\documentclass[aps,prx,twocolumn,superscriptaddress,amsmath,nofootinbib,amssymb,10pt,floatfix]{revtex4-1}
\usepackage{graphicx}
\usepackage{subfigure}
\usepackage{float,caption,subcaption,wrapfig}
\usepackage{mathtools}
\usepackage{bbm}
\usepackage{ulem}   
\usepackage[usenames,dvipsnames]{xcolor}
\usepackage{bbold}
\usepackage{comment}
\usepackage{amsthm}
\usepackage{amssymb} 
\usepackage[colorlinks, linkcolor=blue, citecolor=blue]{hyperref}

\renewcommand{\phi}{\varphi}
\renewcommand{\epsilon}{\varepsilon}

\def\U{\mathrm{U}(1)}

\renewcommand{\sc}[1]{\mathcal{#1}}
\newcommand{\bb}[1]{\mathbb{#1}}

\newtheorem{theorem}{Theorem}[section]

\newtheorem{lemma}[theorem]{Lemma}
\newtheorem{corollary}[theorem]{Corollary}

\theoremstyle{definition}
\newtheorem{definition}[theorem]{Definition}

\newtheorem{remark}[theorem]{Remark}

\newtheorem{example}[theorem]{Example}

\begin{document}
\title{Fractionalization as an alternate to charge ordering in electronic insulators}

\author{Seth Musser}
\affiliation{Department of Physics, Massachusetts Institute of Technology, Cambridge, Massachusetts 02139, USA}
\affiliation{Condensed Matter Theory Center and Joint Quantum Institute, Department of Physics, University of Maryland, College Park, Maryland 20742, USA}
\author{Meng Cheng}
\affiliation{Department of Physics, Yale University, New Haven, Connecticut 06511, USA}
\author{T. Senthil}
\affiliation{Department of Physics, Massachusetts Institute of Technology, Cambridge, Massachusetts 02139, USA}
\begin{abstract}
Incompressible insulating phases of electronic systems at partial filling of a lattice are often associated with charge ordering that breaks lattice symmetry. The resulting phases have an enlarged unit cell with an effective integer filling. Here we explore the possibility of  insulating states\textemdash which we dub ``quantum charge liquids" (QCL)\textemdash at partial lattice filling that preserve lattice translation symmetry. Such QCL phases  must necessarily either have gapped fractionally charged excitations and associated topological order or have gapless neutral excitations.  We establish some general constraints on gapped fermionic QCL phases that restrict the nature of their topological order. We prove a number of results on the \textit{minimal} topological order that is consistent with the lattice filling.  In particular we show that, at rational fillings $\nu = p/q$ with $q$ an even integer, the minimal ground-state degeneracy on a torus of the fermionic QCL is $4q^2$, four times larger than that of the bosonic QCL at the same filling. We comment on models and physical systems which may host fermionic QCL phases and discuss the phenomenology of these phases.
\end{abstract}

\maketitle

\section{Introduction}

The competition between interelectron Coulomb repulsion and electronic kinetic energy is at the heart of many interesting phenomena in condensed-matter physics. If the electrons live on a crystalline lattice, then at any fractional filling $\nu$, defined as the mean number of electrons per unit cell, metallic Fermi-liquid phases will occur if the kinetic energy dominates. When the Coulomb energy dominates, and the filling $\nu$ is rational, an insulating state is expected where the electrons  are localized  in a spatial pattern that breaks lattice translation symmetry. These states are known as Wigner-Mott insulators. Recent experiments \cite{regan_mott_2020, wang_correlated_2020, xu_correlated_2020, mak_semiconductor_2022} on moir\'{e} transition metal dichalcogenide (TMD) materials have found evidence for Wigner-Mott insulating phases at many fractional fillings $\nu$. 

How does the Wigner-Mott insulator evolve into the Fermi liquid as the bandwidth is increased at fixed Coulomb interaction? Are there new phases stabilized at intermediate bandwidth that are distinct from both the Wigner-Mott insulator and the Fermi liquid? Remarkably, questions such as these may be experimentally accessible in the moir\'{e} TMD systems where the bandwidth can be simply controlled by tuning a perpendicular electric field. 

\begin{figure}
    \centering
    \includegraphics[width=\columnwidth]{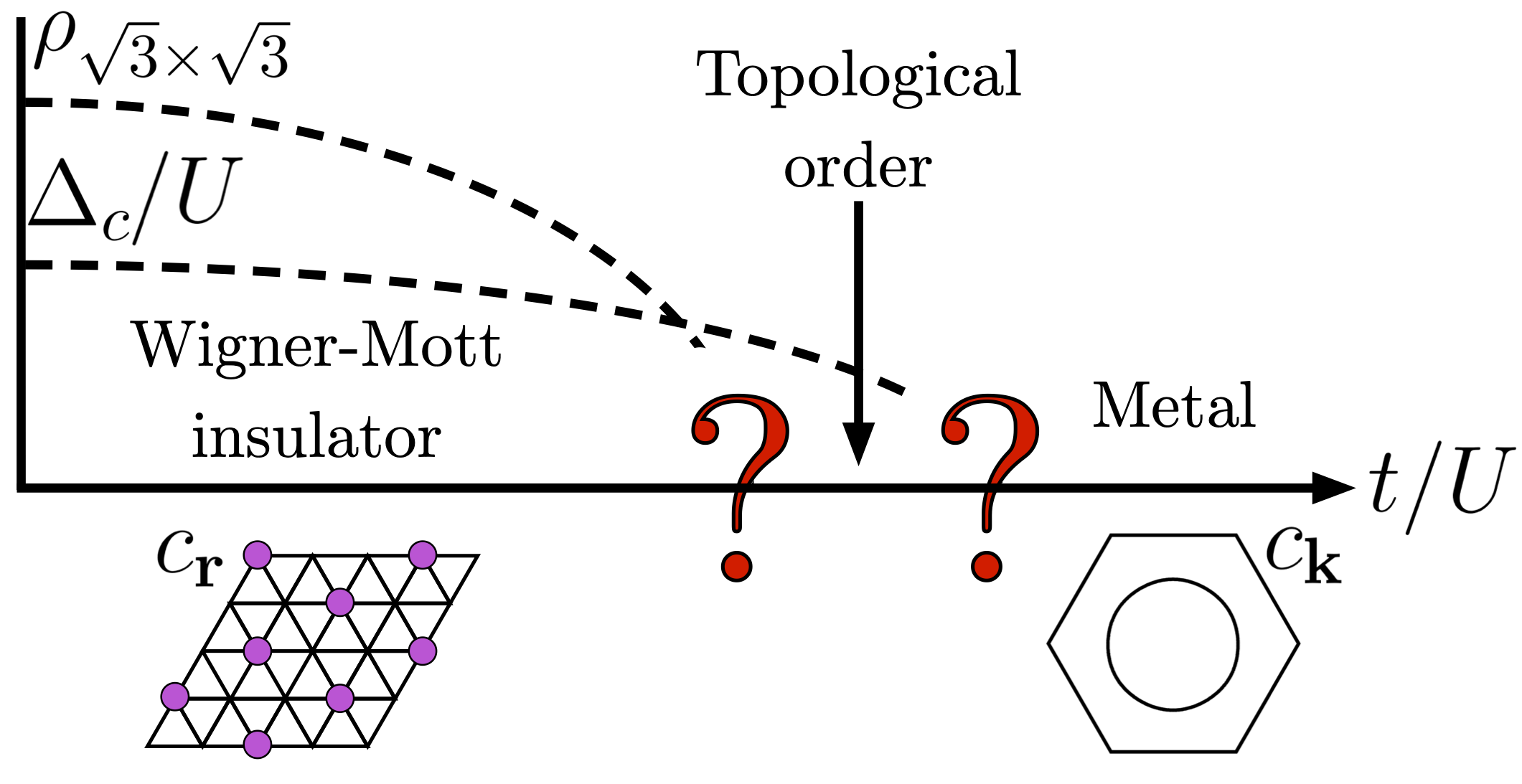}
    \captionsetup{justification=raggedright}
    \caption{A possible evolution between the Wigner-Mott insulator and a Fermi liquid at a fixed fractional lattice filling. $t$ sets the scale of the bandwidth, and $U$ is a measure of the electron-electron interaction strength. We consider an intermediate phase where the charge order is lost before the insulator-metal transition.\footnote{In the extended Hubbard model $t/U\approx 0.1$ in numerical studies of the insulator-metal transition \cite{zhou_quantum_2024}.} This is displayed in the cartoon, where the dimensionless charge gap is $\Delta_c/U$ and the dimensionless charge density wave order parameter is $\rho_{\sqrt{3}\times \sqrt{3}}$. It goes to zero before the charge gap closes. General constraints tell us that this insulating phase can be fully gapped but must exhibit TO, and have excitations that carry fractional charge. For a treatment of the possibility of a continuous phase transition between the metal and the Wigner-Mott insulator without this intervening topologically ordered phase, see Refs.~\cite{xu_interaction-driven_2022, musser_theory_2022, musser_metal_2022}.}
	\label{fig:intermediate}
\end{figure}

In this paper we focus on insulating states at fractional lattice filling. We explore the possibility of states which preserve the translation symmetry of the lattice and are therefore distinct from Wigner-Mott insulators. The well-known Lieb-Schultz-Mattis-Oshikawa-Hastings (LSMOH) theorem \cite{lieb_two_1961, oshikawa_commensurability_2000, hastings_lieb-schultz-mattis_2004} tells us that such states necessarily have fractional-charge excitations and associated topological order (TO), or have gapless neutral excitations.  We will refer to these as ``quantum charge liquids'' (QCLs). We specifically focus on the case where the intervening phase is fully gapped, and hence topological, as indicated in Fig.~\ref{fig:intermediate}.

Although this paper is focused on the experimentally relevant case of fermionic QCLs, we note that the concept, if not the name, of QCLs has arisen before in bosonic systems. One early example is the triangular lattice quantum dimer model, which can be mapped to a hopping problem of hardcore bosons on a kagome lattice \cite{moessner_quantum_2008}. Other tractable microscopic models for bosonic fractionalized phases have also been constructed \cite{balents_fractionalization_2002,senthil_microscopic_2002, motrunich_exotic_2002, motrunich_bosonic_2003}. More fractionalized states of bosons can be accessed theoretically either through parton constructions \cite{senthil_z_2_2000}, or through vortex condensation from a proximate superfluid phase \cite{balents_dual_1999, senthil_z_2_2000,lannert_quantum_2001,balents_putting_2005}.\footnote{We note that many of these constructions are linked to closely related states \cite{read_large-n_1991, wen_mean-field_1991, savary_quantum_2016, broholm_quantum_2020} that occur as gapped quantum spin liquid (QSL) phases of quantum magnets with, at least, a $\U$ spin symmetry \cite{zou_symmetry_2018}. One key difference between quantum magnets and the charge-frustrated systems of interest in this paper is the action of time-reversal symmetry. Electrical charge is even under time reversal while the spin is odd under time reversal. Thus, time reversal acts differently on the generator of global $\U$ in the two cases. The incompressible insulating states of interest in this paper where the lattice filling of particles is a generic rational fraction correspond, in the quantum magnet analog, to magnetization plateau states which, except in the special case of zero magnetization, are not time-reversal invariant. Thus we can even have chiral TO as a possibility for magnetization plateau states \cite{misguich_magnetization_2001} while they are not allowed at fractional filling of a time-reversal invariant charge system. Apart from this conceptual distinction, the natural microscopic interactions of charge-frustrated systems and quantum magnets are of course different. Finally, the concept of a fermionic QCL discussed in this paper has no analog in a spin system.}

In contrast, fermionic QCLs have gotten surprisingly little attention despite their relevance to the aforementioned experiments in moir\'{e} TMDs. Early works \cite{pollmann_charge_2006, pollmann_correlated_2006, pollmann_fermionic_2011} suggest that spinless fermions at $\nu = 1/2$ on a lattice can form a $\bb{Z}_2$ TO insulating state that preserves charge conservation, lattice translation, and time-reversal symmetries. Contrary to these suggestions, and to the familiar bosonic case, we show that such a state is not allowed to exist on general grounds. Instead a filling $\nu = 1$ is needed to observe $\bb{Z}_2$ TO in this system, as recently found in exactly solvable models \cite{han_resonating_2023, cai_quantum_2024}.

For concreteness we now specialize to spinless fermions. This is relevant to the TMD materials where the spin-exchange scale is very small compared with the charge-exchange scale \cite{mak_semiconductor_2022}. The spins are thus easily polarized in a small magnetic field, or become incoherent at small temperatures between the spin- and charge-exchange scales. With this assumption, we present general results about the possible TOs the spinless fermions can form by extending the LSMOH constraints to the fermionic case where: $\U$, lattice translation, and time-reversal symmetries are respected. We remark that fractionalized states of fermions on a lattice in the absence of time-reversal symmetry (either due to explicit or spontaneous breaking) occur in fractional quantum anomalous Hall and fractional Chern insulator states that have been observed recently \cite{park_observation_2023,zeng_thermodynamic_2023, lu_fractional_2024}. As a warm-up we also obtain some results on these time-reversal-broken fractionalized states. We show that, in general, the ``minimal order,'' which we define below, realized by spinless fermions is larger than that of bosons.

Given that LSMOH requires nontrivial TO, we can ask what the simplest or minimal TO is consistent with the lattice filling and symmetries. We propose using ground-state degeneracy (GSD) on the torus to define\footnote{The notion of minimal TO has appeared before in other contexts; see Ref.~\cite{sodemann_composite_2017}.} minimal TOs rather than quantum dimension, as done in a recent paper \cite{jian_minimal_2024}.
We then show that, for fillings of $\nu = p/q$ per unit cell with $p, q$ relatively coprime integers, and $q$ even, the minimal fermionic TO that does not break either translation or time reversal has a GSD on the torus of at least $4q^2$, larger by a factor of four than the minimal GSD of a bosonic TO at the same filling. This was discussed using dual vortex and fermionic tensor network arguments for the special case of $\nu = 1/2$ in Ref.~\cite{bultinck_filling_2018}, although it was not proven. In contrast, for fillings of $\nu = p/q$ with $q$ an odd integer, the minimal fermionic TO will have a GSD of at least $q^2$, exactly the minimal GSD of a bosonic TO at the same filling. We give a proof of these statements, and physically justify the emergence of such TOs using parton arguments. We then prove that there is only a single minimal TO consistent with: lattice translation, charge conservation, and time reversal at these fillings, so the minimal TO is unique.

While our focus in this paper is on gapped QCLs, where we can establish some rigorous results, it is also interesting to contemplate the physics of gapless QCL phases. If we are interested in \textit{insulating} QCLs, the charge must be fully gapped. Thus the gapless sector, if present at all, must consist entirely of electrically neutral excitations.\footnote{This is yet another difference with QSL states of interest in quantum magnetism. Insulating quantum spin liquids may well have gapless spin excitations.} A  specific example of such a gapless insulating QCL with an emergent Fermi surface of neutral fermions is described in Appendix \ref{app:parton_args}.  These may well also appear as intermediate phases between the Wigner-Mott insulator and a Fermi liquid. Finally we could also consider \textit{metallic} QCLs. The Fermi liquid itself could be thought of as a familiar example; non-Fermi-liquid metallic states that preserve the symmetries and filling constraints will surely qualify as metallic QCLs.

Where might we expect to see fermionic QCLs?   In the better understood bosonic systems, fractionalized insulators typically occur in the ``intermediate'' correlation limit where the kinetic energy leads to ring exchange processes. By analogy, and as motivated above,  two-dimensional moir\'{e} TMD materials evolving between Wigner-Mott insulators and Fermi liquids might be an interesting platform to observe fractionalized QCL insulators.  There has also been a recent suggestion \cite{mao_fractionalization_2023} that twisted bilayer graphene at $1/3$ filling might realize a fractionalized phase. Our results on the allowed TO phases at fractional filling should inform further explorations of this suggestion. Finally, we note that the physics of charge frustration appears in a few traditional solid-state systems as well, e.g., in triangular lattice quarter-filled organic materials \cite{hotta_theories_2012}. Although not as tunable as moir\'{e} materials, QCL phases might occur in some such systems as well.

The outline of the rest of the paper is as follows: In Sec.~\ref{sec:definitions} we give a brief review of the mathematical theory of TO and define different classes of ``minimal'' orders. In Sec.~\ref{sec:fractional_fillings} we find the minimal orders at a fixed fractional filling of a lattice $\nu = 1/q$ with and without the presence of time-reversal symmetry. With time-reversal symmetry we show that there is a unique minimal TO, $\bb{Z}_{2q}$ ($\bb{Z}_q$) gauge theory for $q$ even (odd). Finally in Sec.~\ref{sec:discussion} we discuss areas of further study in three related directions: (1) finding microscopic models that realize gapped fermionic QCLs, (2) describing the phenomenology of such QCLs, including their response to doping, and (3) other scenarios where the concept of minimal order could prove clarifying, as well as theoretical questions raised by the concept of minimal order.

\section{Definition of ``minimal" symmetry enriched topological order}
\label{sec:definitions}

We first give a brief review of the theory of TO to establish some common definitions. For a more comprehensive review, see Refs.~\cite{kitaev_anyons_2006} and \cite{bonderson_non-abelian_2007}. A TO $\sc{C}$ is characterized by its anyon types and the chiral central charge $c_-$. We focus on the anyons, since $c_-$ must be 0 in time reversal invariant systems, which is the main subject of the work. Two anyons $a,b\in \sc{C}$ can be fused, resulting in a direct sum of other anyons $a\times b = \sum_c N^{c}_{ab}c$, where $N^c_{ab}\geq 0$ are integers known as the fusion multiplicity. The fusion rules determine the quantum dimensions $d_a\geq 1$ for the anyons, through the equation $d_ad_b=\sum_c N_{ab}^c d_c$. If $d_a=1$ then $a$'s fusion outcome with any other anyon is unique and it is thus an Abelian anyon. One important Abelian anyon is the unique ``identity'' anyon $1$, which corresponds to a local bosonic excitation. It has the property that $1\times a = a$ for all $a\in \sc{C}$. The set of Abelian anyons $\sc{A}\subseteq \sc{C}$ forms an Abelian group where the group multiplication is defined by fusion and the identity element is given by $1$. If $\sc{A}=\sc{C}$ then we refer to $\sc{C}$ as an Abelian TO. The total quantum dimension $\sc{D}_{\sc{C}}$ of $\sc{C}$ is defined as
\begin{equation}
\sc{D}_{\sc{C}} = \sqrt{\sum_{a\in \sc{C}} d_a^2}.
\end{equation}
We thus see that $\sc{D}_{\sc{C}}^2 \geq |\sc{C}|$, where $|\sc{C}|$ is the number of anyons in $\sc{C}$, with equality if and only if $\sc{C}$ is an Abelian TO.

In addition to the fusion rules for anyons in $\sc{C}$ a TO also needs to specify the exchange and braiding statistics of the anyons.  The exchange statistics of an anyon $a$ is denoted by $e^{i\theta_a}$, where $\theta_a$ is referred to as the topological spin of $a$. The effect of braiding an anyon $a$ around $b$ is captured by the monodromy $M_{ab}$. If either $a$ or $b$ is an Abelian anyon then $M_{ab} = e^{i\theta_{a,b}}$, where $\theta_{a,b}$ is referred to as the braiding phase. Restricted to $\sc{A}\subseteq \sc{C}$ the function $\theta_{a,b}$ is a symmetric bilinear form on $\sc{A}$ i.e., a map from $\sc{A}\times \sc{A}$ to $\bb{R}/2\pi\bb{Z}=\mathbb{S}^1$ that satisfies $\theta_{a,b} = \theta_{b,a}$ and $\theta_{a\times b,c} = \theta_{a,c} + \theta_{b,c}$. Moreover, $\theta_{a,a} = 2\theta_a$ i.e., the braiding phase of an anyon with itself is twice its topological spin.
The identity anyon is a boson which braids trivially with all other anyons, and thus $\theta_1 = 0$ and $\theta_{1,a} = 0$ for all $a\in \sc{C}$. 
In a bosonic TO, $1$ is the only anyon type with this property. A fermionic TO will have an Abelian anyon $c\in \sc{C}$ corresponding to a local fermion i.e., the electron. This anyon obeys $c\times c = 1$, $\theta_c = \pi$, and $\theta_{a,c} = 0$ for all $a\in \sc{C}$. In a fermionic TO, $1$ and $c$ are the only anyons which braid trivially with everything else. We note that physically the electron is a local particle, and as such is sometimes not counted as an anyon in $\sc{C}$. However, keeping $c\in \sc{C}$ will prove to be a useful bookkeeping device for us, so we adopt this convention throughout. With this convention, a trivial fermionic insulator\footnote{Here we refer to fermionic invertible states as trivial because they have no nontrivial anyons.} is described by the theory $\{1,c\}$.

We further note that a bosonic TO $\sc{C}_b$ can be trivially turned into the fermionic TO, $\sc{C}_b\boxtimes\{1,c\}$, where $\boxtimes$ means stacking of two anyon theories. Physically, it means that the TO $\sc{C}_b$ is made of bosonic degrees of freedom in a fermionic system. However, not all fermionic TOs can be written in this form,\footnote{A familiar example of a fermionic TO that is not equivalent to a bosonic TO stacked with $\{1,c\}$ is the Moore-Read Pfaffian state.} although Abelian ones can (see Theorem \ref{thm:decomp} in Appendix \ref{app:abelian_results}).

Consistency of the fusion and braiding rules establishes constraints on the possible TOs. For a comprehensive review, see Appendix E of Ref.~\cite{kitaev_anyons_2006}. If, in addition, the TO is required to respect a symmetry $G$, this will impose further constraints, ruling out some symmetry actions and enrichments of G on the TO \cite{barkeshli_symmetry_2019, bulmash_fermionic_2022, bulmash_anomaly_2022, aasen_characterization_2022}. Those that are consistent with a symmetry action of $G$ are referred to as symmetry-enriched TOs (SETs). 

In addition to a symmetry $G$ a system may have a ``fractional topological response'' \cite{cheng_gauging_2023} that requires a nontrivial TO. The example explored in this paper is fractional filling of a lattice, $\nu$, which requires nontrivial TO by the LSMOH theorem \cite{lieb_two_1961, oshikawa_commensurability_2000, hastings_lieb-schultz-mattis_2004}. Another familiar example is the fractional quantum Hall effect \cite{hallbootstrap}. A natural question is what the ``least complicated,'' or minimal, SETs are given a symmetry $G$ and a fractional topological response $\nu$. Surprisingly, despite the well-developed study of SETs, the question of minimality is a relatively recent one. Answering it first requires one to define what is meant by minimal. A recent paper \cite{jian_minimal_2024} endeavored to do this by defining a minimal SET to be one which has the smallest quantum dimension possible, given that it must respect $G$ and have the appropriate fractional topological response. We propose another definition of minimal. We say that a SET is minimal if it has the smallest GSD on a torus, given that it must respect $G$ and $\nu$. We prefer our definition for two reasons. The first is that a lower bound on GSD provides a more useful tool for numerics such as exact diagonalization, where GSD on a torus is easily assessed.\footnote{Note that in electronic systems the GSD may have size dependence \cite{wen_quantum_2003, rao_theory_2021}; when we discuss GSD of these systems we refer to those with an even number of electrons.} The second reason to prefer our definition is that it is more powerful when proving results about minimal Abelian SETs. We make this more precise below.

\begin{figure}
    \centering
    \includegraphics[width=\columnwidth]{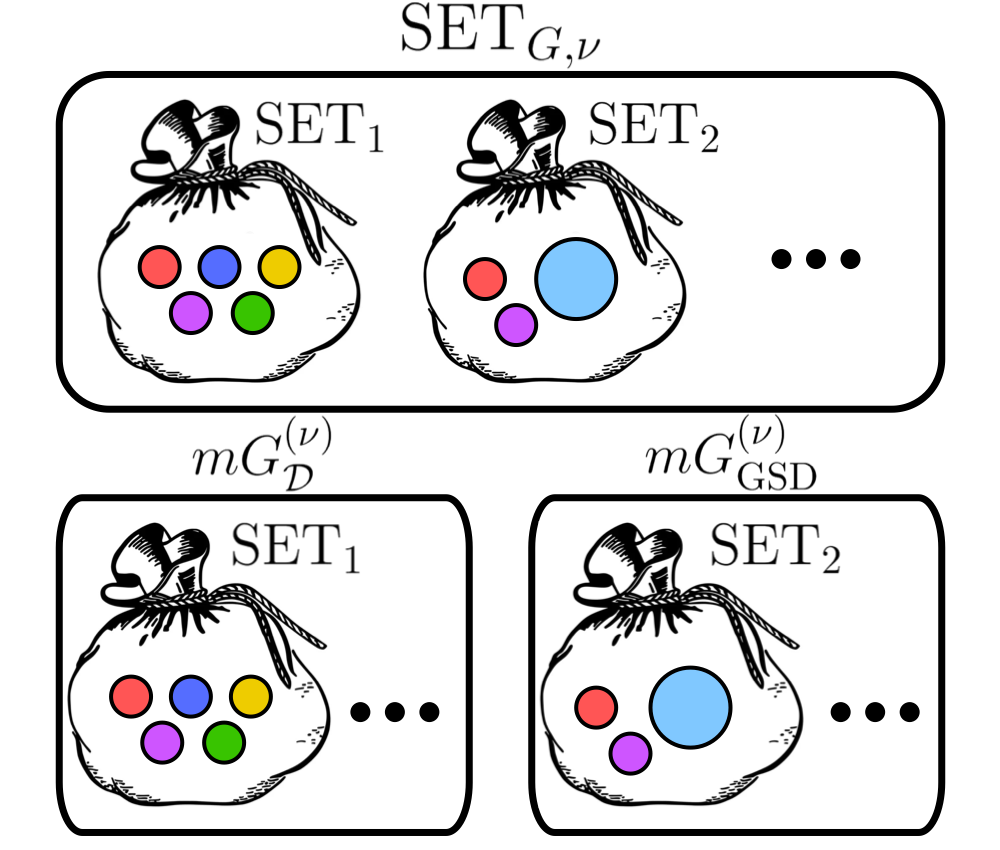}
    \captionsetup{justification=raggedright}
    \caption{A cartoon picture of the minimal order classes. The top black box is the set of SETs consistent with the symmetry $G$ and filling fraction $\nu$, SET$_{G,\nu}$. Contained within SET$_{G,\nu}$, illustrated by bags, are the SETs, SET$_{1}$ and SET$_2$ along with other possibilities. Each of these SETs contains anyons, illustrated by the colorful balls. The radius of each anyon is a proxy for its quantum dimension, with the smaller balls therefore being Abelian and the larger light blue ball being non-Abelian. The bottom-left illustrates the minimal class $mG^{(\nu)}_{\sc{D}}$, containing all SETs in SET$_{G,\nu}$ with the smallest quantum dimension. This will correspond to the total ``weight'' of the balls in the bag, so SET$_2$ is excluded because of its large non-Abelian anyon. The bottom-right illustrates the minimal class $mG^{(\nu)}_{\rm GSD}$, containing all SETs with the smallest number of anyons. The large number of anyons in SET$_1$ thus rules out its inclusion. This illustration makes it clear that $mG^{(\nu)}_{\sc{D}}$ will tend to be biased in favor of Abelian theories, while the reverse is true of $mG^{(\nu)}_{\rm GSD}$.}
	\label{fig:orders_marbles}
\end{figure}

Let SET$_{G,\nu}$ be the set of all SETs respecting $G$ and $\nu$. Then define:
\begin{equation}
mG^{(\nu)}_{\sc{D}} = \{\sc{C}\in {\rm SET}_{G,\nu} \mid \sc{D}_{\sc{C}}\leq \sc{D}_{\sc{C}'} \ \forall \sc{C}'\in {\rm SET}_{G,\nu} \},
\end{equation}
this forms the set of minimal SETs according to the definition of Ref.~\cite{jian_minimal_2024}. Now the GSD of a TO is equal to the number of distinct anyon types \cite{kitaev_anyons_2006}, which is $|\sc{C}|$ for bosonic TOs and $|\sc{C}|/2$ for fermionic TOs, where this factor of two compensates for our including $c\in \sc{C}$.\footnote{The GSD is independent of the periodic or antiperiodic boundary condition. As mentioned before, we assume an even number of electrons.}
Therefore, achieving the smallest GSD will mean achieving the smallest number of anyons. We thus define 
\begin{equation}
mG^{(\nu)}_{\rm GSD} = \{\sc{C}\in {\rm SET}_{G,\nu} \mid |\sc{C}|\leq |\sc{C}'| \ \forall \sc{C}'\in {\rm SET}_{G,\nu} \},
\end{equation}
as the set of SETs with minimal GSD on a torus. We can think of these as \textit{classes} of minimal SETs. For a given $G$ there is \textit{a priori} no reason for these two classes to have any relationship to one another. However, a useful heuristic is that the minimal class $mG^{(\nu)}_{\sc{D}}$ is more likely to contain Abelian TOs than the class $mG^{(\nu)}_{\rm GSD}$. This is illustrated in Fig.~\ref{fig:orders_marbles}.

We make this more precise by stating a result proven in Appendix \ref{app:minimal_classes}. It follows from the inequality $|\sc{C}|\leq \sc{D}_{\sc{C}}^2$, which we noted earlier. If there is an Abelian TO in $mG^{(\nu)}_{\rm GSD}$, then all TOs in $mG^{(\nu)}_{\sc{D}}$ are Abelian and moreover $mG^{(\nu)}_{\sc{D}}\subseteq mG^{(\nu)}_{\rm GSD}$.  On the other hand, there is no reason for there to be any Abelian TOs in $mG^{(\nu)}_{\rm GSD}$ even if one of the TOs in $mG^{(\nu)}_{\sc{D}}$ is Abelian. We give an example of this in Appendix \ref{app:minimal_classes}. It is thus more powerful to establish results about the minimal orders of $mG^{(\nu)}_{\rm GSD}$ being Abelian, than about those of $mG^{(\nu)}_{\sc{D}}$. Likewise, it is more powerful to establish results about the minimal orders of $mG^{(\nu)}_{\sc{D}}$ being non-Abelian.

For our choice of $G$, specified in the next section, we see that \textit{all} TOs in $mG^{(\nu)}_{\rm GSD}$ are Abelian and thus that $mG^{(\nu)}_{\rm GSD} = mG^{(\nu)}_{\sc{D}}$. Thus in this case the two classifications of minimal order turn out to be equivalent. However, it is nonetheless more powerful to work with $mG^{(\nu)}_{\rm GSD}$ since this will imply things about $mG^{(\nu)}_{\sc{D}}$, but not the other way around. We thus henceforth use the term minimal to refer to a SET in $mG^{(\nu)}_{\rm GSD}$, unless otherwise specified. Note that we often discuss minimal TOs rather than minimal SETs, here a minimal TO is the underlying TO of a minimal SET.

\section{Minimal order at fractional fillings}
\label{sec:fractional_fillings}
 
We consider a general system of spinless fermions that respect $\U$ fermion number conservation and are at a fractional filling $\nu = 1/q$ per unit cell; all of our arguments also easily generalize to the case $\nu = p/q$ with $p$ relatively prime to $q$. For our purposes we henceforth take the symmetry group to be given by
\begin{equation}
G = \bb{Z}^2\times [\U_f\rtimes \bb{Z}^T_2], \label{eqn:symm_gp}
\end{equation}
where the first $\bb{Z}^2$ is translation in the $x$ and $y$ directions, the $\U_f$ is $\U$ charge with the restriction that the charge mod $2$ is equal to the fermion parity, and $\bb{Z}_2^T$ is the order-two time-reversal operation which does not change the $\U$ charge. In what follows we will also take the electron to be a Kramers singlet i.e., we assume that $\sc{T}^2 = 1$. This is appropriate in thinking about the possibility of the fractionalized states in TMD moir\'{e} heterostructures, where, as we emphasized, the electron spin exchange is very small compared with other energy scales. Thus the spin may be easily fully polarized in a small magnetic field without affecting the orbital motion. 

As will become clear in the following sections the inclusion of time-reversal symmetry in $G$ proves to be very restrictive on the possible minimal theories. However, it also makes proving results about the minimal order more challenging. As a warm-up we thus begin by supposing that the system breaks time reversal, meaning its symmetry group is now given by:
\begin{equation}
H = \bb{Z}^2 \times \U_f.
\label{eqn:TRS_gp}
\end{equation}
We review some generalities about TOs respecting $H$ and $\nu$ in Sec.~\ref{sec:generalities}. We demonstrate in Sec.~\ref{sec:warm-up} that the minimal TOs in both $mH^{(\nu)}_{\sc{D}}$ and $mH^{(\nu)}_{\rm GSD}$ are Abelian with $2q$ anyons when $q$ is odd, and $4q$ anyons when $q$ is even. The proof is relatively straightforward and illustrative of how we tackle the full $G$.

Having warmed up we dive into a discussion on constraints on the minimal TOs consistent with time reversal as well i.e., the full $G$. In Sec.~\ref{sec:construction} we construct a fermionic Abelian TO which preserves $G$ and has a GSD of $4q^2$ ($q^2$) for $q$ even (odd). 
We then prove a number of results. First, in Sec.~\ref{sec:minimal_Abelian}, we show that any TO at filling $\nu = 1/q$ which respects $G$ and has a GSD equal to or smaller than our construction must be Abelian. Moreover, we show that all such possibilities must have exactly the same GSD as our construction, so our construction is minimal. Then, in Sec.~\ref{sec:uniqueness_of_SET}, we show that our construction is the \textit{unique} minimal TO that obey all the requirements. We defer a full classification of the symmetry enrichments of this minimal TO to future work.

We briefly remark that we have thus demonstrated a hierarchy of lower bounds on torus GSD when the system is gapped.  
\begin{itemize}
    \item If all symmetries are unbroken then we see that we must have a GSD of at least $4q^2$ for even $q$, and $q^2$ for odd $q$.
    \item However, if the time reversal is spontaneously broken but the translation symmetry is preserved, then we must have a GSD of at least $2\times 2q =4q$ for even $q$, $2q$ for odd $q$. The extra factor of two is due to the fact that time reversal is broken spontaneously.
    \item  If our GSD is even smaller than this, then we can conclude that translation symmetry must be spontaneously broken. This will allow us to have a GSD as small as $q$ e.g., in a stripe order. Note that the GSD of $q$ in this case is not topological but arises from the breaking of translation symmetry.
\end{itemize}

\subsection{Generalities about TOs at fractional filling} \label{sec:generalities}
We now discuss a few general results about TOs at fractional filling~\cite{cheng_translational_2016}, which form the basis of our arguments.

In any gapped TO, bosonic or fermionic, with $\U$ symmetry, there exists an anyon $v$ called a vison,\footnote{The name `vison' is used here as a generalization of its original use \cite{senthil_z_2_2000} to denote an Ising vortex in a $\bb{Z}_2$ gauge theory. In the special situation in which the low energy physics is described by a $\bb{Z}_2$ gauge theory with a charge-$1/2$ (under the global $\U$) chargon, and a charge-neutral vison, the latter is indeed nucleated by such a flux insertion. } which is nucleated by flux insertion, whose braiding with other anyons detects their (fractional) $\U$ charge:
\begin{equation}
\label{eq: thetava}
    \theta_{v,b}=2\pi Q_b \pmod{2\pi},
\end{equation}
where $Q_b$ is the charge carried by $b$. Physically, $v$ can be understood as resulting from the adiabatic insertion of a $2\pi$ $\U$ flux. It is important to keep in mind that in a fermionic system with $\U_f$ symmetry, all bosonic (fermionic) local excitations carry even (odd) charges. This means we can take $Q_b$ to be defined modulo two.\footnote{The $\mathrm{U}(1)$ charge respects fusion, so $Q_{z} = Q_x+Q_y \pmod{2}$ for any $z$ in the fusion product of $x$ and $y$.}

In addition, the classic Laughlin argument shows that $Q_v = \sigma_H$ (the Hall conductance in units where electronic charge and Planck's constant are set to 1).  It follows from Eq.~\eqref{eq: thetava} that $\theta_{v,v}=2\pi Q_v=2\pi\sigma_H$. The topological spin of $v$ must therefore be $\theta_v=\pi\sigma_H$ or $\pi\sigma_H+\pi$. If continuous rotation symmetry is assumed then it can be shown that $\theta_v = \pi \sigma_H$ \cite{Goldhaber:1988iw, Greiter:2022iph}. This can also be shown without continuous rotation symmetry \cite{Kapustin:2020bwt}, so we always have $\theta_v = \pi\sigma_H$.

In a translation-invariant state with fractional charge per unit cell $\nu = 1/q$, there should be a background anyon $a$ which has $\U$ charge $Q_a = \nu$. 
The presence of such an anyon can be argued using the  aforementioned vison which detects the background $\U$ charge per unit cell \cite{cheng_translational_2016, bultinck_filling_2018}. However, it can also be understood to arise naturally from the requirement that the TO be translation invariant; the most natural way to do this is to fractionalize the fundamental particle into $q$ anyons with charge $1/q$ and place one in every unit cell. This is depicted in Fig.~\ref{fig:boson}. 

We note that the background anyon $a$ must be Abelian. This was shown more formally in Ref.~\cite{cheng_translational_2016}, but we sketch two complementary arguments here. First suppose that $a$ is non-Abelian, with quantum dimension $d_a>1$. Then the Hilbert-space dimension associated with having non-Abelian anyons on every site will scale as $d_a^N$, where $N$ is the number of sites \cite{kitaev_anyons_2006}. This macroscopic degeneracy must be lifted by some interaction, meaning that this description is at best useful at an intermediate energy scale, and does not correspond to the true infrared (IR) limit. 
Flowing towards the true IR, any further macroscopic degeneracy will continue to be broken until we have fully broken it, at which point we must have a background anyon with $d_a=1$ i.e., $a$ must be Abelian.

Another complementary perspective can show the same conclusion. Operationally, the background anyon is defined by the following thought experiment:
 (1) excite some anyon pair $x$ and $\overline{x}$, (2) adiabatically transport $x$ around a unit cell, (3) and annihilate the pair. This will return the system to its ground state, so it can only lead to a global phase $\varphi_{x}$ for the wave function. We can think of this phase as the braiding of $x$ around the background anyon $a$. Importantly, this phase is additive under fusion: if $z\in x\times y$, then $\varphi_z=\varphi_x+\varphi_y$.  
 Via a result in Ref.~\cite{barkeshli_symmetry_2019} [see Eq. (45)], this is then enough to establish that $\varphi_x=\theta_{a,x}$ for some Abelian anyon $a$, which we identify as the background anyon.

\subsection{Time-reversal broken minimal order}
\label{sec:warm-up}

\begin{figure}
    \centering
    \subfigure[Bosonic TO]{
        \includegraphics[width=0.45\columnwidth]{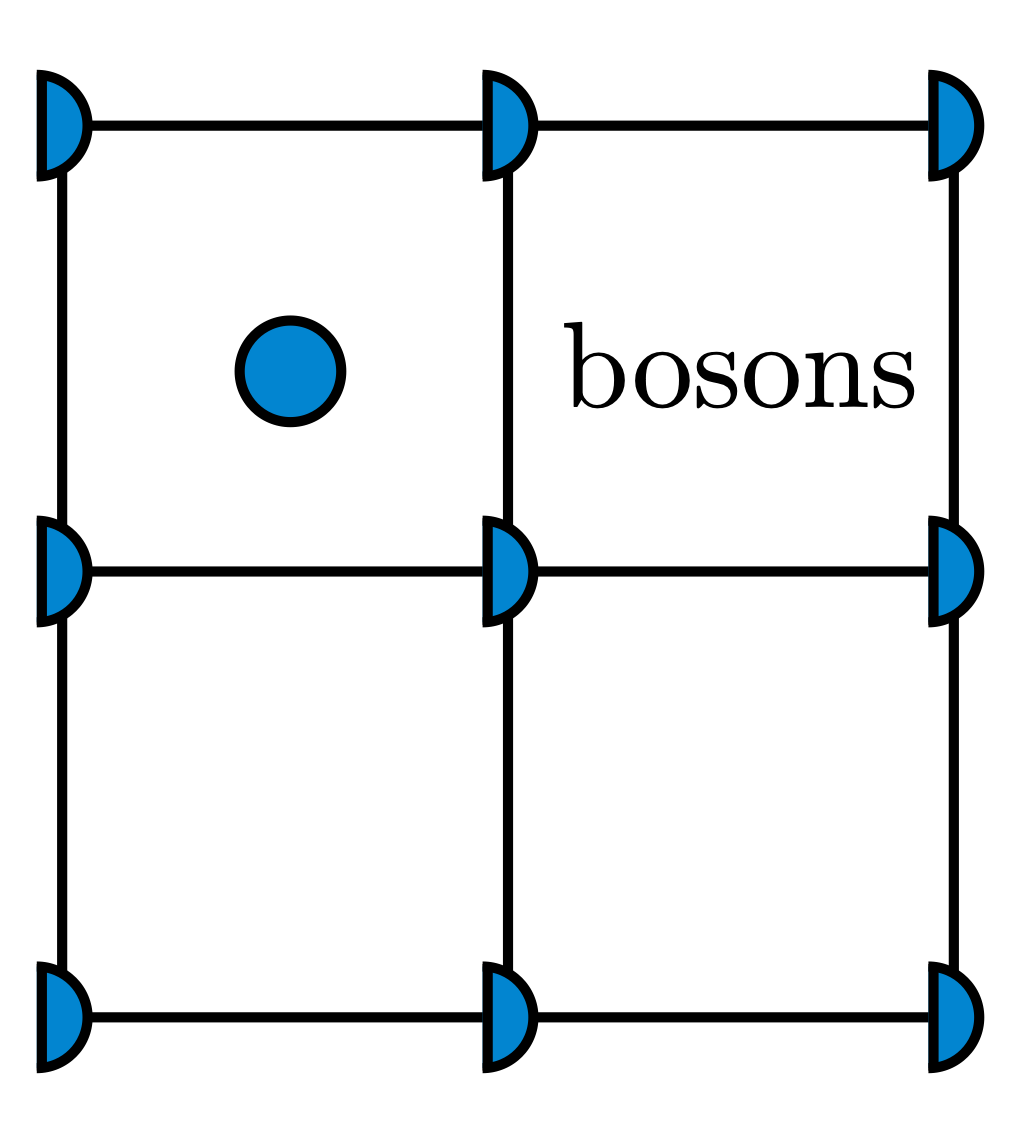}
        \label{fig:boson}
            }
    ~ 
    \subfigure[Fermionic TO]{
        \includegraphics[width=0.45\columnwidth]{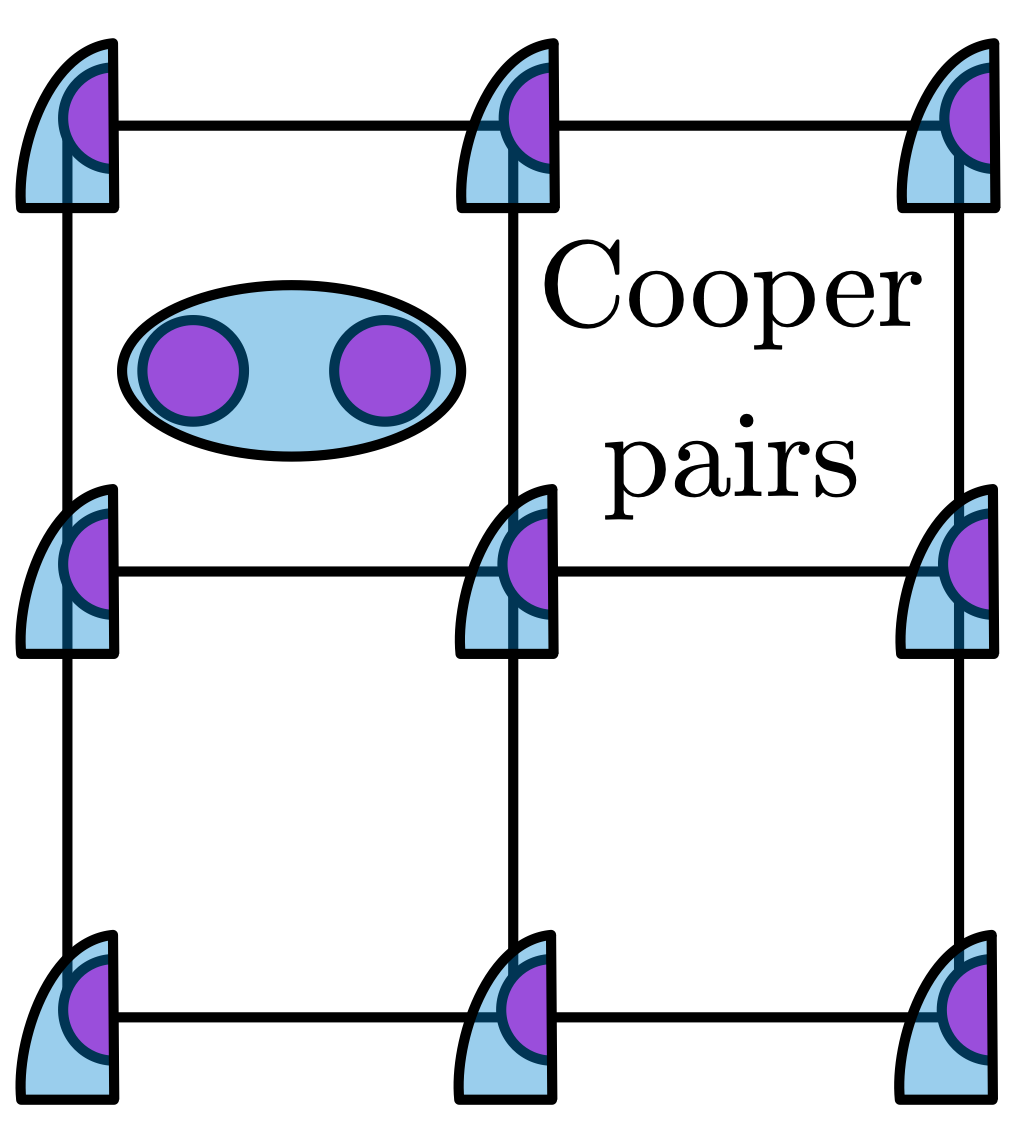}\
        \label{fig:cooper}
            }
    \captionsetup{justification=raggedright}
    \caption{A cartoon picture of the easiest way to construct a gapped state at fractional filling without breaking translation. Panel \subref{fig:boson} considers bosons at $\nu=1/2$ filling per unit cell on the square lattice. The bosons, illustrated by a blue circle, can fractionalize into a background anyon $a$ with charge $1/2$ which will be placed on each lattice site. In panel \subref{fig:cooper} we now consider spinless fermions at $\nu=1/2$. By the arguments in the main text the same fractionalization cannot occur for the fermions. However, if the fermions pair into Cooper pairs at $\nu/2 = 1/4$ filling per unit cell, then they may fractionalize into a quarter of the Cooper pair on each site.}
	\label{fig:TOs}
\end{figure}

For bosons at $1/q$ filling, with $q$ even, $\U_q$ (the TO of the bosonic $1/q$ Laughlin state) is a natural candidate for time-reversal breaking TOs in bosonic QCLs. We now demonstrate that for $q$ even, $\U_q\boxtimes\{1,c\}$ cannot be a fermionic TO at filling $1/q$. First we note that for even $q$ $\U_q$ by itself is a bosonic TO, and local operators in the theory are all bosons. It follows that the local excitation with minimal nonzero charge in $\U_q$ is a charge-$2$ Cooper pair. Since every anyon in this theory has order $q$, the minimal fractional charge in this theory is $2/q$. But this means this theory does not have a background anyon of charge $Q_a=\nu = 1/q$, contradicting the presence of $\U_f$ symmetry and the filling constraint.

We can construct a possible fermionic TO by pairing two fermions into Cooper pairs. The Cooper pairs are bosons at a filling $\nu/2 = 1/2q$ per unit cell. We then make a TO, $\sc{C}$, which is translation invariant by considering $\U_{2q}$, the bosonic Laughlin state of the Cooper pairs. Strictly speaking we have $\U_{2q}\boxtimes \{1,c\}$, where $c$ is the fundamental fermion. This will be a fermionic TO in SET$_{H,\nu}$, and is depicted in Fig.~\ref{fig:cooper}. This theory has a GSD of $2q$. Of course if $q$ is odd the theory $\U_q$ is a fermionic TO in SET$_{H,\nu}$ with a smaller GSD, $|\U_q|/2=q$.

We now prove that these two constructions belong to $mH^{(\nu)}_{\rm GSD}$ for the respective choice of $q$. We do this by showing that the Abelian sector in any fermionic TO in SET$_{H,\nu}$ is at least as large as our constructions. So then if our theory is to be minimal according to $mH^{(\nu)}_{\rm GSD}$ it must be Abelian with exactly the same anyon count as our constructions. This will prove that every TO in $mH^{(\nu)}_{\rm GSD}$ is Abelian and thus that $mH^{(\nu)}_{\rm GSD}=mH^{(\nu)}_{\sc{D}}$, so the two definitions of minimality coincide. Our proof strategy is illustrated in Fig.~\ref{fig:vice}.

In addition to the background anyon $a$ we know that the fundamental fermion is an Abelian anyon with integer charge, so $a\neq c$ and we have two Abelian anyons $a,c\in \sc{A}$. Let us now consider how many different anyons we can generate from the powers of $a$ and $c$. We start with the case of $q$ odd. The order of $a$ must be a multiple of $q$ due to its $\U$ charge being $1/q$. Since we have already constructed a theory with $2q$ anyons, $\U_q$, it must be the case that if our theory is to be minimal the order of $a$ is either $q$ or $2q$. However, the order of $a$ cannot be $q$ because $a^q$ has charge one, which cannot be identified with the vacuum. Thus the order of $a$ is $2q$. Then for our theory to be minimal i.e., have at most $2q$ anyons, it must be the case that $c\in \langle a\rangle$. The only possibility consistent with charge conservation is $c=a^q$. Then we see that our theory is Abelian and $\sc{A} = \langle a\rangle$. Here we use $\langle\cdot\rangle$ to represent an Abelian group in terms of its generators.
It is useful to write $\sc{A}=\sc{A}'\boxtimes\{1,c\}$, where $\sc{A}'=\langle ac\rangle$ has $q$ anyons and is in fact a bosonic TO because $(ac)^q=1$. All such bosonic TOs are completely classified~\cite{Moore:1988qv, bonderson_non-abelian_2007}  and can be labeled by $\theta_a=2\pi n/q$ with $n$ coprime to $q$. Note that $\U_q$ corresponds to $n=2$.

Although we have determined the possible minimum TOs, we can still ask whether there are distinct symmetry enrichments by $\bb{Z}^2\times \U_f$. To that end, we need to specify the vison anyon. It is easy to see that $v=a^{k}$, where $k$ is such that $2nk=1\pmod{q}$. Since $n$ is coprime with $q$ and $q$ is odd, such $k$ always exists and is uniquely fixed. Thus the possible minimal orders in the time reversal breaking case have Hall conductances of $\sigma_H = k/q \pmod{1}$ where $k$ can take all values relatively prime to $q$. An example of such a QCL phase was recently demonstrated in \cite{chen_how_2025}. It is worth noting that we can (trivially) modify the TO by stacking copies of invertible topological states\footnote{In two dimensions, the basic such state is the so-called $E_8$ state with chiral central charge $c = 8$.} of charge-neutral local bosonic excitations, which does not change the anyon content.

Next we tackle $q$ even. In this case we again must have that $a^q\neq 1$. So then $a$ must have order a multiple of $2q$. Suppose now that $a^q=c$, then $a^q$ must braid trivially with $a$. This means that $\theta_{a,a^q}=2q\theta_a = 0$. Since $q$ is even we multiply by the integer $q/2$ to see that $0=q^2\theta_a=\theta_{a^q}$ and thus $a^q$ is a boson, a contradiction with $a^q=c$. So $c\notin \{1,a,\ldots,a^{2q-1}\}$ and we must have that there are at least $4q$ anyons in our theory. For our theory to be minimal it must thus be Abelian with $\sc{A} = \langle a\rangle\boxtimes \{1,c\}$ where $a$ has order exactly $2q$. All such theories are again classified \cite{Moore:1988qv, bonderson_non-abelian_2007}, with $\theta_a=\pi n/2q$ for any $n$ such that $n$ is coprime with $q$. In particular this means $n$ must be odd. We can then choose $v=a^k$, where $k$ is such that $nk=2\pmod{2q}$. This means $k$ must be even, and more specifically that $k=2k'$ with $k'$ relatively prime to $2q$. Note $n=1$ corresponds to the $\U_{2q}$ theory, where we can choose $v=a^2$. Thus when $q$ is even the minimal orders have Hall conductances $\sigma_H = 2k'/q \pmod{1}$, where $k'$ can take all values relatively prime to $q$.

This general proof strategy will prove useful in dealing with the full $G$. We first construct a theory in SET$_{G,\nu}$, then show that all minimal theories must be at least as big as our construction. We again refer to reader to Fig.~\ref{fig:vice} for an illustration of this strategy. Having obtained the anyon count of each minimal theory we then show that in the presence of time-reversal symmetry the minimal TO is in fact unique.

\begin{figure}
    \centering
    \includegraphics[width=\columnwidth]{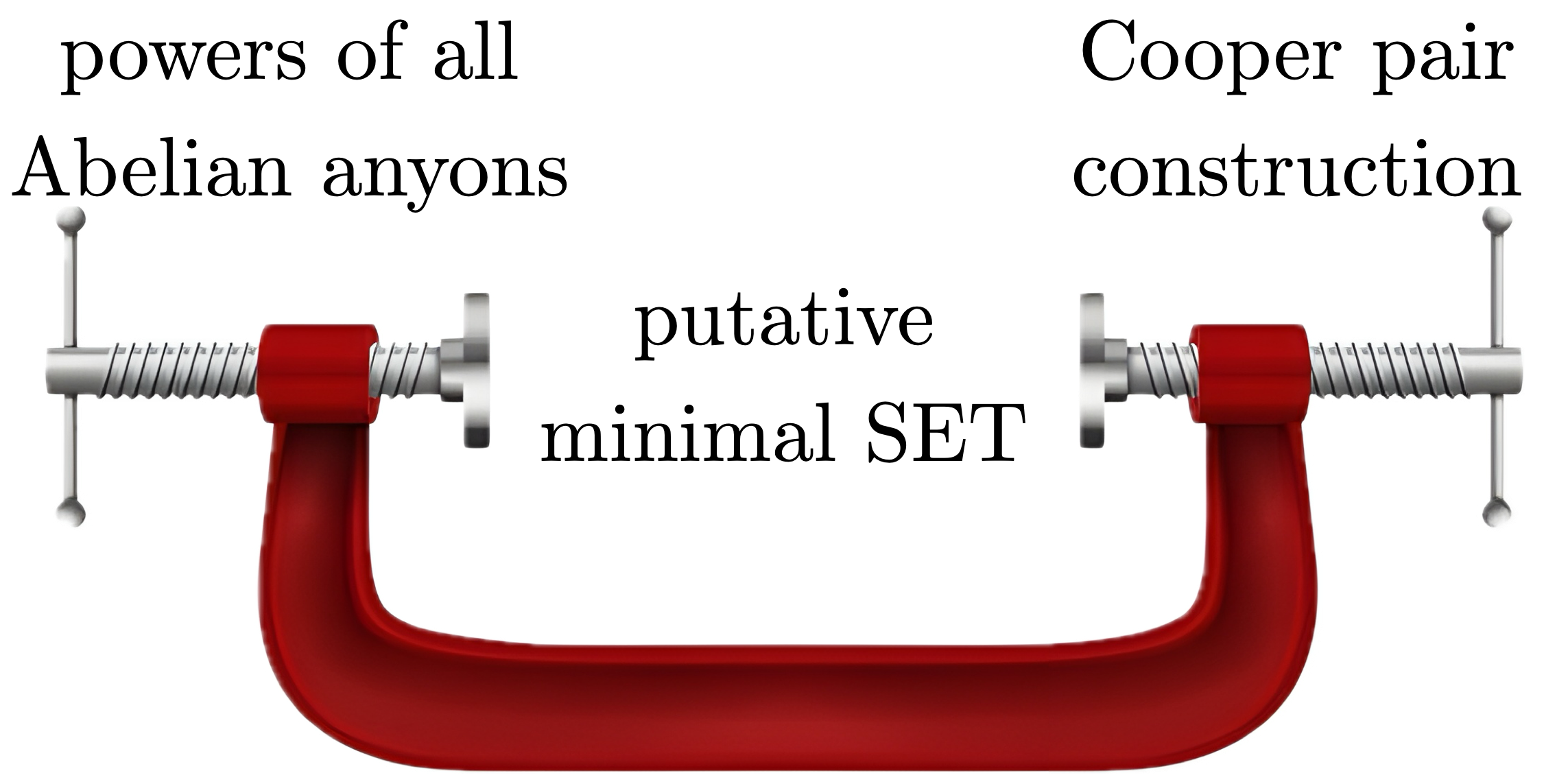}
    \captionsetup{justification=raggedright}
    \caption{Our proof strategy for demonstrating the number of anyons a minimal theory must have. Using Cooper pairing we construct a theory in SET$_{G,\nu}$. This theory's anyon count must upper bound the anyon count of any minimal SET. We further consider all possible powers of the Abelian anyons we know our theory must possess. This will lower-bound the anyon count. Any putative minimal SET will be caught in the vice between these two bounds.}
	\label{fig:vice}
\end{figure}

\subsection{Construction of a fermionic SET of $G$}
\label{sec:construction}

Let us briefly review the construction of a bosonic TO which preserves the symmetries of the full group $G$ in Eq.~\eqref{eqn:symm_gp}, when the system is bosonic. We know from the discussion in the last section that we must possess an Abelian background anyon $a$, with $Q_a = \nu$. A natural bosonic TO which has such a background anyon and preserves the symmetries of $G$ is $\mathbb{Z}_q$ gauge theory with the usual action of time reversal. It will be a consequence of our later arguments that this is in fact \textit{the minimal} bosonic TO that preserves the symmetries of $G$.
However, for even $q$, one can again see that the $\bb{Z}_q$ TO must have excitations with minimal fractional charge $2/q$, and thus cannot be a fermionic TO at $\nu = 1/q$ without breaking any symmetries.

We can again turn to Cooper pairing to construct fermionic TOs, as illustrated in Fig.~\ref{fig:TOs} and the previous section. The Cooper pairs are bosons at a filling $\nu/2 = 1/2q$ per unit cell. We then make a TO, $\sc{C}$, which is time-reversal and translation invariant by considering the $\bb{Z}_{2q}$ gauge theory\footnote{Stacked with $\{1,c\}$.} of the Cooper pairs. This state is a fermionic Abelian TO with a GSD of $|\sc{C}|/2 = 4q^2$.

We can achieve a smaller quantum dimension if $q$ is odd i.e., $q=2k+1$. Consider a band insulator which is made up of a state with one electron per site, and holes at a filling $\nu_h = 1-\nu = 2k/(2k+1)$. We then Cooper pair the holes to make a bosonic insulator at filling $\nu_h/2 = k/(2k+1)$. A symmetry preserving bosonic insulating state can then be made with a $\bb{Z}_{2k+1} = \bb{Z}_{q}$ gauge theory of these paired holes. Thus for $q$ odd we have found a fermionic TO with GSD $|\sc{C}|/2=q^2$.

While the use of Cooper pairs to construct these states may not seem energetically natural if the dominant interaction is repulsive, in Appendix \ref{app:parton_args} we provide an alternate construction of these TOs using a parton construction which may yield a more natural route. 

\subsection{Any minimal TO is Abelian}
\label{sec:minimal_Abelian}

Having constructed a possible fermionic TO via Cooper pairing, we now prove that this is a minimal order consistent with $G$ in Eq.~\eqref{eqn:symm_gp}. We do this by again following the strategy of Fig.~\ref{fig:vice}. We begin by reviewing the basic Abelian anyons we know our theory possesses.

First, we have a background anyon $a$ with charge $Q_a=1/q$. As a consequence, the order $n_a$ must be an integer multiple of $2q$. Here the order $n_a$ is defined as the smallest nonzero integer such that $a^{n_a}=1$.
 
We further know that there exists another Abelian anyon, $v$, nucleated by $2\pi$ flux threading. Since $\theta_{v,a}=2\pi/q$, the order of $v$ must be an integer multiple of $q$. 

Lastly, we know that $\theta_v = \pi \sigma_H \pmod{2\pi}$. Since our theory is time-reversal invariant the Hall conductivity must be zero and thus $\theta_v = 0$ i.e., $v$ is a charge-neutral boson. 

\subsubsection{If $q$ is odd then the minimal order must be Abelian with a GSD of $q^2$} 

We start with $q$ odd as a warm-up. By our earlier arguments we know that the subset of Abelian anyons $\sc{A}\subseteq \sc{C}$ must contain at least three distinct anyons $a,v,c\in \sc{A}$. We consider all possible powers of these anyons in $\sc{A}$; the size of $|\sc{A}|$ is lower-bounded by the number of their unique powers, as in Fig.~\ref{fig:vice}. We see that there are $2q^2$ Abelian anyons generated, proving that our earlier construction was minimal. Note again that we are counting the electron as an anyon.

One approach to showing the number of distinct powers of $a,v,c$ is $2q^2$ is by simply checking this manually. Consider the set $\sc{S}=\{a^kv^l|0\leq k,l\leq q-1\}$. Let us show that all anyons in this set are distinct. Otherwise, it means that there exists some $r,s\in\{0,1,\dots,q-1\}$ with $a^rv^s=1$ other than $r=s=0$. Since $v$ is a boson, $\theta_{a^rv^s,v}=2\pi r/q$, so we must have $r=0$. But then for $r=0$ we have $\theta_{v^s,a}=2\pi s/q$, which forces $s=0$, a contradiction. Our argument also shows that any anyon in this set braids nontrivially with at least one other anyon, so the electron $c\notin \sc{S}$. We thus conclude that the minimal number of anyons in $\sc{C}$ is at least $2q^2$, which is saturated by the $\bb{Z}_q$ toric code.

\subsubsection{If $q$ is even then the minimal order must be Abelian with a GSD of $4q^2$}

We use a very similar argument as above, but we need to be a little more careful since the upper bound on our minimal number of anyons is now $8q^2$, where we are again counting the electron as an anyon for bookkeeping purposes. In the main text we first make the simplifying assumption that translation does not permute anyons unlike, e.g., Wen's plaquette model \cite{wen_quantum_2003}. We show in Appendix \ref{app:non_abelian_results} that under this assumption the time-reversal invariance of our state will imply that $a=\sc{T}a$, and $a$ is thus a boson or fermion. We can then show that if $a=\sc{T}a$ there must be at least $8q^2$ Abelian anyons. In Appendix \ref{app:non_abelian_results} we further allow for the possibility that translation permutes anyons and hence $a$ may not be the same as $\sc{T}a$, but nonetheless show that there must be at least $8q^2$ Abelian anyons.

We begin by noting that our logic for the case of $q$ odd revealed that there are $2q^2$ unique anyons in the set $\sc{S} = \{v^ma^nc^k\mid 0\leq m,n\leq q-1, k=0,1\}\subseteq \sc{C}$. Furthermore, the only anyons in $\sc{S}$ which braid trivially with all other anyons in $\sc{S}$ are $1,c$. Now consider $a^q$. Since $a$ is a boson or fermion it must be the case that $a^q$ braids trivially with all powers of $a$. Moreover, since $a^q$ has integer charge, then it must braid trivially with all powers of $v$ as well. So then $a^q$ braids trivially with all anyons in $\sc{S}$. However, $a^q\neq 1$, because it has odd charge and the vacuum does not in a fermionic theory. Additionally $a^q\neq c$ since $a^q$ is a boson. So then we see $a^q\notin \sc{S}$. 

We can now conclude that there are $4q^2$ anyons contained in $\sc{S}'=\{v^ma^nc^k\mid 0\leq m\leq q-1, 0\leq n\leq 2q-1, k=0,1\}$. Suppose now that $a^{2q}\notin \sc{S}'$. Then there are at least $8q^2$ anyons in the powers of $a,v,c$. For our theory to be minimal we can have at most this many anyons, so this must be all the anyons in the theory. But this is clearly a contradiction, since $a^q\neq 1,c$ but braids trivially with all anyons. Thus we must have $a^{2q}\in \sc{S}'$. Since $a^{2q}$ is a boson that braids trivially with all other anyons in $\sc{S}'$ it must be the case that $a^{2q}=1$ or $a^q$. The latter is not possible since $a^q\neq 1$. So then $a^{2q}=1$.

Next we know that since $a^q\neq 1,c$ but braids trivially with all anyons in $\sc{S}'$ there must exist some $\gamma \in \sc{C}\setminus\sc{S}'$ which braids nontrivially with $a^q$. Furthermore, since $a^{2q}=1$, we must have that:
\begin{equation}
\theta_{a^q,\gamma} = \pi.
\end{equation}
{At this stage, the} anyon $\gamma$ can then either be Abelian or non-Abelian. Suppose that it is Abelian. Since $\gamma$ was not contained in the powers of $a,v,c$, we conclude that there are at least $8q^2$ unique powers of $a,\gamma,c$ and we are done.

Suppose then that $\gamma$ is non-Abelian. In Appendix \ref{app:non_abelian_results} we prove that $\gamma$ can be taken to have the fusion products:
\begin{equation}
\gamma \times \gamma = v+va^qc \text{ and } a^qc\times \gamma = \gamma.
\label{eqn:gamma_fusion_rules}
\end{equation}
We first sketch the argument for these fusion rules, with full proofs in Appendix \ref{app:non_abelian_results}. Since $\gamma$ has $\pi$ braiding with $a^q$ we can always choose $v$ to be in the fusion product of $\gamma$ with itself. But since $\gamma$ is non-Abelian there must be some other anyon in this fusion product. By minimality the only possibilities are powers of $v,a,c$. Then the only other possibility by charge assignment is $va^qc$. This gives us the first fusion rule up to multiplicities of the fusion products, which we show in Appendix \ref{app:non_abelian_results} must be one. The second fusion rule is a consequence of minimality as well; if it were not true there would be more than $8q^2$ anyons. Finally in Appendix \ref{app:non_abelian_results}, we use a result in Ref.~\cite{lapa_anomaly_2019} to show that time-reversal invariance will require that $\gamma$ have $\theta_\gamma = 0$.

Now we condense the vison $v$ to produce a TO with the three anyons $1,\gamma,a^qc$ that has the same fusion rules as the Ising TO, but with a non-Abelian anyon, $\gamma$, with $\theta_\gamma=0$. But all TOs with Ising-like fusion rules have $8\theta_\gamma=\pi$~ \cite{kitaev_anyons_2006}. We thus obtain a contradiction and conclude that $\gamma$ must be Abelian. As noted earlier we are now done.

We can reach the conclusion that $\gamma$ must be Abelian through a different logic that does not directly rely on knowing that $\theta_\gamma = 0$. Consider first the action of time reversal on $v$ and its powers. Clearly $v$ will transform to $v^{q-1}$. However, for even $q$, $v^{\frac{q}{2}}$ will map to itself. We can then ask if it is a Kramers singlet or a Kramers doublet. Both possibilities are allowed and correspond to different SET states. First consider the case where it is a Kramers singlet. Then the state obtained by condensing $v$ will preserve time reversal. However, as we have just seen, with a non-Abelian $\gamma$, this state will have the fusion rules of the Ising TO which is certainly not time-reversal invariant, and we have a contradiction. Next consider the case where $v^{\frac{q}{2}}$ is a Kramers doublet. Now time reversal will have a complicated action on $v$, and the state obtained by condensing $v$ will not be time-reversal invariant. Instead we first argue that for every TO where $v^{\frac{q}{2}}$ is a Kramers doublet, there is a partner TO where it is a Kramers singlet. To see this, we note that $v^{\frac{q}{2}}$ has a $\pi$ braiding phase with $a$. If $a$ is a boson or a Kramers doublet fermion, we form a topological insulator of these particles. As is well known \cite{lu_theory_2012}, this has the effect of toggling the $\pi$ flux seen by this particle between Kramers singlet and Kramers doublet. If instead $a$ is a Kramers singlet fermion, we repeat the same procedure with $ac$ which is a boson that has $\pi$ braiding with $v^{\frac{q}{2}}$. Importantly, this procedure will not alter the quantum dimension of $\gamma$. Thus since we already argued that a phase with non-Abelian $\gamma$ and a Kramers singlet $v^{\frac{q}{2}}$ does not exist, it follows that such a phase also cannot exist with a Kramers doublet $v^{\frac{q}{2}}$.

\subsection{Uniqueness of the minimal TO}
\label{sec:uniqueness_of_SET}

We have established that all minimal TOs consistent with the symmetry group $G$ of Eq.~\eqref{eqn:symm_gp} must be Abelian and have an anyon count of $8q^2 \ (2q^2)$ for $q$ even (odd). As mentioned earlier this means that the definitions of minimality by either quantum dimension or ground-state degeneracy will coincide i.e., $mG^{(\nu)}_{\sc{D}}=mG^{(\nu)}_{\rm GSD}$. We now show that our construction is always the \textit{unique} minimal TO. We also comment on whether there is a unique symmetry enrichment of this TO. Throughout we assume that translation does not permute anyons, so $a$ must be a boson or fermion. In Appendix \ref{app:uniqueness} we present more general proofs without relying on this assumption.

\subsubsection{The case of $q$ odd}

This case is very straightforward. We know that $a$ is a particle with charge $1/q$ and $a^q$ must braid trivially with everything. So then $a$ must be a fermion and $a^q=c$ by charge constraints. We further know that the vison is a charge-zero boson and that all anyons in the minimal theory are powers of these two. Since the braiding of $a$ and $v$ is given by $\theta_{a,v} = 2\pi/q$ the topological spin and braiding of all the anyons are then fully specified. It is simple to check that this will just be $\bb{Z}_q$ gauge theory. 

In fact, this theory will not just be the unique minimal TO, but the unique minimal SET as well. To see this note that the charge assignment is totally fixed by the charges of $a$ and $v$. It is easy to see that the time reversal acts as $\mathcal{T}v=v^{-1}, \mathcal{T}a=a$. Because $a^q=1$ and $q$ is odd, $a$ must be a Kramers singlet. Furthermore, translation on $v$ is projective, and the corresponding phase is fully fixed by the mutual braiding between $a$ and $v$. Finally, we note that translation action on $a$ cannot be projective. For a projective translation we have to place a background $v$ or some power of it at each plaquette center, as was considered in a time reversal breaking context in Ref.~\cite{song_translation-enriched_2022}. However such a background is not time-reversal invariant. Thus the translation action on $a$ (and its powers) is also totally fixed.

\subsubsection{The case of $q$ even}

When $q$ is even we know that $a$ is an order $2q$ boson or fermion with charge $1/q$. Furthermore, as discussed in the previous section there must exist some Abelian anyon $\gamma$ that has $\pi$ braiding with $a^q$ and all anyons in our theory must be made up of powers of $a,\gamma,c$. To fully specify the topological spin of our theory, we only need to know the topological spin of the Abelian $\gamma$. Without loss of generality we can assume that $\gamma^2 = v$, so $\theta_\gamma = 0, \pm \pi/2,  \pi$,\footnote{Note that earlier in Sec. \ref{sec:minimal_Abelian} we found $\theta_\gamma=0$ under the assumption that $\gamma$ was non-Abelian; having ruled out that scenario we now must determine $\theta_\gamma$ when $\gamma$ is Abelian.} and $\theta_{\gamma,a}=\pi/q$. We can further take $\gamma$ to have charge zero; if it did not we could consider $\gamma a^q$ which would have the same braiding with $a$ and would also square to the vison.

We can now use a result from Ref.~\cite{lapa_anomaly_2019} that
\begin{equation}
e^{2\pi i(c_{-}-\sigma_H)/8} = \frac{1}{\sqrt{2}\sc{D}}\sum_{b\in \sc{C}} d_b^2 e^{i(\theta_b + \pi Q_b)}
\end{equation}
for a $\U_f$ symmetry enriched fermionic TO. Since the theory is time-reversal invariant, both $c_-$ and $\sigma_H$ must vanish and the left-hand side will be equal to one. The right-hand side can then be evaluated using the fact that $4\theta_\gamma = 0$ and $\theta_{a,\gamma} = \pi/q$. The end result is
\begin{equation}
1 = \frac{1}{2}\left(1 - e^{i\theta_{a^q}} + e^{i\theta_\gamma} + e^{i(\theta_\gamma + \theta_{a^q})}\right),
\end{equation}
and since $a^q$ must be a boson this tells us that $\theta_\gamma = 0$. It is then simple to check that the braiding of the powers of $a,\gamma,c$ is just that of $\bb{Z}_{2q}$ gauge theory.

However, we note that in this case we have many choices of symmetry enrichment, as already alluded to in the previous section.  Without loss of generality take the $e$ particle in $\bb{Z}_{2q}$ gauge to have charge $1/q$ and the $m$ particle to be charge neutral. Since $a$ is allowed to be a charge $1/q$ boson or fermion there are two possibilities: $a=e$, or $a=em^q$. The vison in this theory will always be given by $v = m^2$. In addition since we saw that the charge-neutral $\gamma$ had to be a boson, then $\gamma = m$. It is also possible that $a$ or $m^q$, or both, may be a Kramers doublet since we have a Kramers singlet fermion.\footnote{In particular for Kramers singlet fermions with $\U_f \rtimes \bb{Z}_2^T$ symmetry, the anomaly structure \cite{wang_erratum_2015} does not prevent this possibility.} Besides these possibilities we will not attempt to classify all possible SETs in this case, but will leave it to future works.

\section{Discussion}
\label{sec:discussion}

Here we explored the possiblity of a gapped quantum charge liquid (QCL) at a fractional filling of the underlying lattice which nevertheless does not break translation or time-reversal symmetries. Such a QCL must display topological order (TO), and have excitations carrying fractional charge due to LSMOH constraints. Thus fractionalization can appear as an alternate to charge ordering at fractional lattice filling. Here we raised the question of the \textit{minimal} topological order that is consistent with the symmetries and filling constraints, and established rigorous results. We showed that a fermionic QCL will generically be more complicated than the bosonic QCL at the same filling. A simple illustration of our arguments is that a system of spinless fermions at $\nu = 1/2$ cannot form a $\bb{Z}_2$ TO, but can form a $\bb{Z}_4$ TO. Thus the change in the statistics of the particles at filling $\nu = 1/2$ must result in a dramatic shift in the TOs realized. Below we discuss a number of questions related to the physics of QCLs that naturally follow from our work.

\subsection{Microscopic models for fermionic QCLs}
It would be interesting to explore microscopic models that realize gapped fermionic QCLs. A natural strategy is to start with  existing models for bosonic QCLs and replace the bosons with fermions. Our results show that, if the model continues to fractionalize, there will be a drastic shift of the ground-state degeneracy on a torus. Concretely, one possible route to take is to study fermionic realizations of the bosonic quantum dimer \cite{moessner_quantum_2008} and Balents-Fisher-Girvin \cite{balents_fractionalization_2002} models. Both models can be mapped to hardcore boson models on the kagome lattice with a fixed cluster charge per kagome hexagon \cite{isakov_topological_2011}. They are at fillings $\nu = 1/2$ and $\nu = 3/2$ of the kagome unit cell, respectively. Thus the minimal bosonic TO they can realize is $\bb{Z}_2$ gauge theory, and both models realize this. Suppose that the hardcore boson is replaced with a spinless fermion. Our proof above demonstrates that the minimal TO that can be realized in such a system is $\bb{Z}_4$ gauge theory. More generally, it will thus be interesting to study fermionic models of charge frustration and search for gapped QCL insulating phases. 

Of particular interest in the context of moir\'{e} TMD materials is to explore the physics of microscopic models of these systems in an intermediate coupling regime where neither the Coulomb interaction nor the bandwidth is overwhelmingly strong. This can possibly be pursued with density-matrix renormalization-group methods \cite{kiely_continuous_2023, zhou_quantum_2024}. Some caution is required in reaching conclusions about experiments because there is still some uncertainty in the details of the microscopic models of these systems due to poorly understood effects of lattice relaxation.

\subsection{Phenomenology of QCLs} 
One could also ask, for instance in the TMD setting, about the physical properties of QCLs, should they occur, and their experimental  manifestation. Choosing a filling with many closely degenerate charge ordered states would promote charge fluctuations, possibly allowing us to realize the intervening QCL phase in this experimental setting. The charge order has been imaged directly in scanning tunneling microscopy experiments deep in the insulator \cite{li_imaging_2021}. If the QCL is present, then as the displacement field used to tune the phase diagram is increased, the charge order will vanish before the transition to the metal. Alternately, since, at rational $\nu$, the charge order  will be described by a discrete order parameter, there will be a finite-temperature phase transition which could be detected experimentally. If the corresponding transition temperature vanishes within the regime where the ground state is incompressible, then an intervening QCL phase must exist, as depicted in Fig.~\ref{fig:phase_diagram}.

\begin{figure}
    \centering
    \includegraphics[width=\columnwidth]{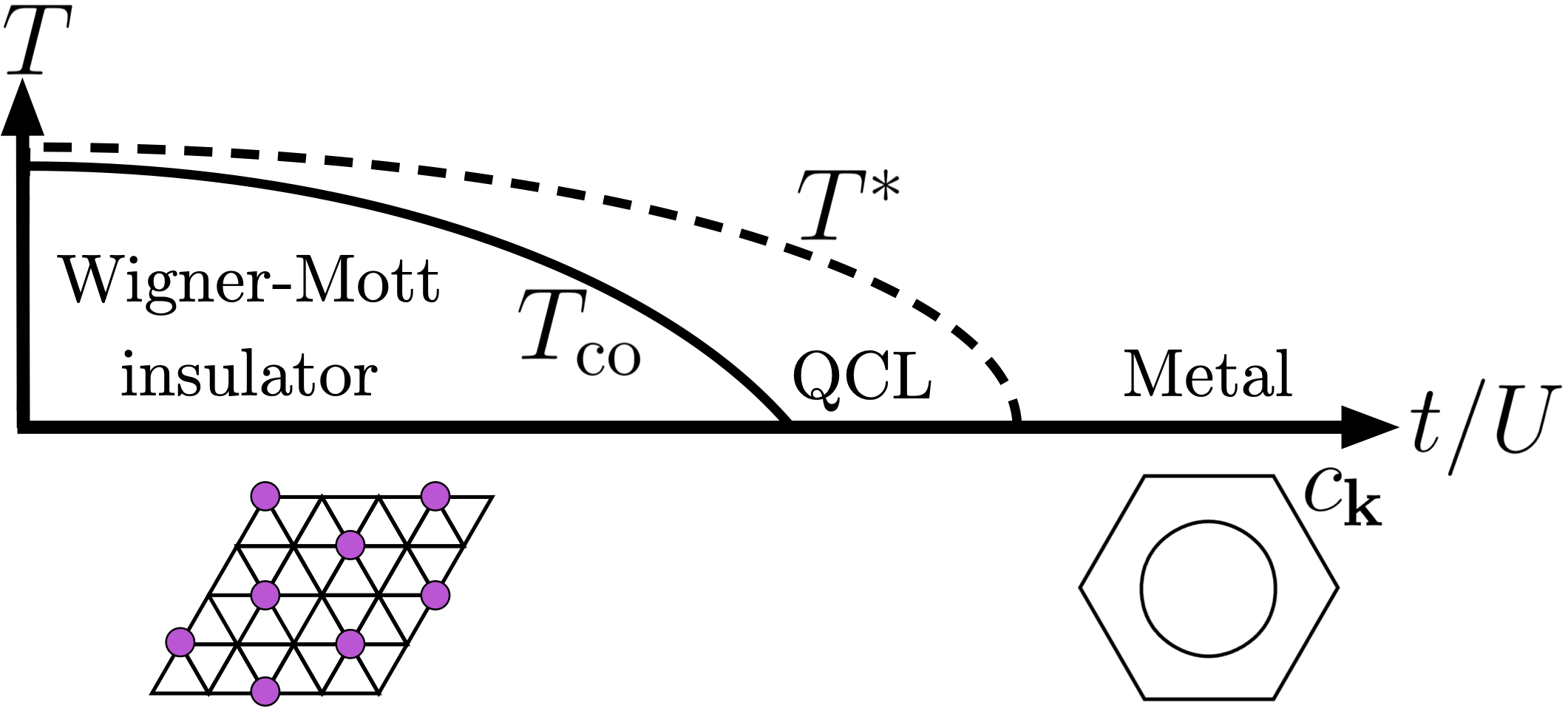}
    \captionsetup{justification=raggedright}
    \caption{A temperature $T$ vs bandwidth $t/U$ schematic phase diagram showing the possibility of an intervening QCL phase. For small bandwidths a Wigner-Mott insulator is stabilized at a fixed fractional $\nu$; here a cartoon of a possible $\sqrt{3}\times \sqrt{3}$ charge order is displayed. As the bandwidth is increased, the transition temperature to this phase, $T_{\mathrm{co}}$, vanishes. However, below a crossover temperature, $T^*$, the system continues to be nearly incompressible. Then, following our discussion, a gapped QCL must be stabilized at zero temperature, even though TO cannot be stabilized at finite temperature in $(2+1)$d. As the bandwidth is increased further the gapped QCL will transition into a compressible metal, shown by $T^*$ going to zero. The metal's Fermi surface is shown below the axis.}
	\label{fig:phase_diagram}
\end{figure}

The phase transitions at a fixed lattice filling out of the gapped QCL, assumed to be in the minimal $\bb{Z}_q$ or $\bb{Z}_{2q}$ TO, are also interesting. The transition between the QCL and the Wigner-Mott insulator is simply driven by condensing the $m$ particle of the gauge theory.\footnote{$m$ can be identified as the $v$ anyon in the odd-$q$ case, and the $\gamma$ anyon in the even-$q$ case in Sec. \ref{sec:uniqueness_of_SET}.} This confines the TO but, as $m$ has projective translation symmetry, will lead to a charge ordered insulator. The transition will be in the universality class known \cite{senthil_deconfined_2023} as $XY^*$ with known critical exponents.\footnote{Although not directly pertinent to the present context, interesting multicritical confinement transitions out of $\bb{Z}_n$ TO, where both $e$ and $m$ attempt to simultaneously condense, have been described recently \cite{shi_analytic_2024}.} The transition from the gapped QCL to the Fermi-liquid metal is mysterious though, and could be a target for future work.  

We could also ask about the effects of doping the QCL. A simple possibility is that the cheapest charged excitation in the insulator is the $e$-particle which will then have a nonzero density at a small doping.  If the $e$-particle is a boson, an allowed possibility in the minimal TO, then a natural doped state is a condensate of $e$, which is a superconductor that breaks the global $\U$ symmetry. The $e$ condensation will completely destroy all the anyons in the theory so that, as is known from other contexts, this state is smoothly connected to a BCS superconductor though it is stabilized by a different mechanism. Note that this superconductor will be fully gapped because it evolves out of a gapped insulator and will have odd parity pairing (because this is the only possibility for spinless fermions).

If, instead the $e$-particle is a fermion, a natural doped state is a non-Fermi-liquid metal where the charge carriers are ``orthogonal'' to the physical electrons, known as an ``orthogonal metal'' \cite{nandkishore_orthogonal_2012}. If the electron filling is changed to $\nu = \frac{p}{q} + \delta$, the doped $e$-particles are at a filling $q\delta$ due to their fractional charge. Thus the $e$-particles will form a Fermi surface with an area that is $q$ times larger than what would be inferred from the density of doped electrons. The increased Fermi momentum may be measured in a local experiment by looking for $2K_F$ signatures (Friedel oscillations) in the charge density around impurities. The fractional charge of the carriers in this orthogonal metal can also be directly evinced through shot noise measurements.

Finally, it is possible that despite the fractionalization in the QCL, the cheapest charged excitation is actually just the electron. Then upon doping the charge will just go in through doping electrons on top of the TO of the QCL. The resulting state will have a conventional Fermi surface of electrons with an area set by the doping density in the usual way. However, the TO of the QCL will remain unchanged, and we will have a fractionalized Fermi liquid \cite{senthil_fractionalized_2003} (FL$^*$) that violates the conventional Luttinger theorem. 

While other possibilities, e.g., a first-order transition to a conventional Fermi liquid, are not ruled out, these observations show that a QCL, if stabilized at the rational filling, can also lead to fascinating physics in a nearby doping range. 

We note that throughout our discussion we have assumed the electrons are spinless, but if the electron spin is not polarized then there are additional phenomenological possibilities. These include the possibility of phases which are both QCLs and QSLs.

\subsection{Theoretical questions on minimal TO} 
On the more theoretical front, one might also be able to apply the proof strategy we used to determine the minimal TO of other symmetry groups. We highlight a few situations where this question is pertinent. Electronic quantum Hall bilayers at a total Landau level filling $\nu = 1/2 + 1/2$ have long been studied, both in experiment and in theory. The standard fate is the development of interlayer coherent order (an exciton condensate) accompanied by an integer quantum Hall effect. References \cite{sodemann_composite_2017, potter_realizing_2017} discussed an alternate symmetric gapped phase, which could be stabilized for some range of microscopic parameters, that has $\bb{Z}_4$ TO. The global symmetry of this system has two $\U$s for each layer, as well as an anti-unitary symmetry that flips particles to holes in both the layers. Reference \cite{sodemann_composite_2017} further conjectured that the $\bb{Z}_4$ TO is minimal (in terms of GSD). A proof of this conjecture is accessible with our techniques~\cite{hallbootstrap}.

The quantum Hall bilayer system has an anomaly involving charge conservation and particle-hole symmetries \cite{sodemann_composite_2017,potter_realizing_2017}. More generally, the $(2+1)d$ surface state of a symmetry protected topological (SPT) phase in $(3+1)d$ has an anomalous symmetry realization and cannot be trivially gapped. In a large class of states, a symmetric gapped state with TO is, however, possible. What then is the minimal TO for a system with a global symmetry $G$ that has an anomaly?  

Another quantum Hall bilayer system where the minimality question arose recently is discussed in the paper \cite{jian_minimal_2024} which considered a system with a quantized fractional spin Hall response. The symmetry group is then a product of two $\U$s for each spin (equivalently valley) along with time reversal that interchanges them \cite{wen_cheshire_2024}. With time reversal, the authors of Ref.~\cite{jian_minimal_2024} were able to show that the minimal order \textit{according to quantum dimension} was also $\bb{Z}_4$ gauge theory at a fractional spin Hall coefficient of one half.\footnote{Although the $\bb{Z}_4$ gauge theory appears in both problems, there is a slight difference in the action of the anti-unitary symmetry in this system and the quantum Hall bilayer system considered in Refs.~\cite{sodemann_composite_2017, potter_realizing_2017}.} However, their proof leaves open the possibility that there might exist a non-Abelian SET with a smaller anyon count but larger quantum dimension. 
We suspect that some variant of the proofs used here can rule out such a possibility. Indeed, the proofs here should be useful to demonstrate the minimal size of the subset of Abelian anyons in any SET which respects an order two anti-unitary symmetry.

Finally, the concept and definition of minimal order raises some other interesting theoretical questions. First, to what extent should we expect minimal order to be realized in a physical system respecting a symmetry group $G$? Ultimately, whatever order is stabilized will depend sensitively on energetics. A recent example is illustrated by the numerical identification of nonminimal non-Abelian quantum spin Hall states in twisted MoTe$_2$ \cite{abouelkomsan_non-abelian_2024, jian_minimal_2024}. In this situation minimal order only seems to be energetically favored if pairing between spin species is sufficiently strong \cite{zhang_vortex_2024, zhang_non-abelian_2024}. However, it might nonetheless be the case that minimal order is generically favored in other situations, e.g., if the interactions are short-ranged in real space. Second, one could ask if there are other sensible classes of minimal order and what their relationship to $mG^{(\nu)}_{\sc{D}}$ and $mG^{(\nu)}_{\rm GSD}$ is. Finally, one could ask if it is possible to pick any TO in SET$_{G,\nu}$ and ``flow'' to a minimal order by successive condensation of anyons which do not break the symmetries $G$. If this is true then minimal orders could be regarded as the ``fixed points'' of all TOs in SET$_{G,\nu}$.

\acknowledgements 
We thank Chao-Ming Jian for a discussion about Ref.~\cite{jian_minimal_2024}, Fiona Burnell for a discussion about the possibility of non-Abelian background anyons, and Arkya Chatterjee, Mike Crommie, Hart Goldman, Ho Tat Lam, Kin Fai Mak, Salvatore Pace, Zhengyan ``Darius'' Shi, Daniel Spiegel, Ashvin Vishwanath, and Liujun Zou for their willingness to discuss the work and suggest references. We also thank Zhihuan Dong, Nandini Trivedi, Ruben Verresen, and Mike Zaletel for helpful discussions of a related upcoming work. S.M. is supported by the Laboratory for Physical Sciences. M.C. was supported by NSF grant DMR-1846109.  T.S. was supported by NSF grant DMR-2206305, and partially through a Simons Investigator Award from the Simons Foundation.  This work was also partly supported by the Simons Collaboration on Ultra-Quantum Matter, which is a grant from the Simons Foundation (Grant No. 651446, T.S.).

\appendix

\section{Results about minimal classes}
\label{app:minimal_classes}

To make this section maximally readable we define the minimal SET classes again.

\begin{definition}
Let SET$_{G,\nu}$ be the set of all SETs that respect a symmetry $G$ and fractional topological response $\nu$. Then define:
\begin{equation}
mG^{(\nu)}_{\sc{D}} = \{\sc{C}\in \text{SET}_{G,\nu} \mid \sc{D}_{\sc{C}}\leq \sc{D}_{\sc{C}'} \ \forall \sc{C}'\in \text{SET}_{G,\nu}\},
\end{equation}
and
\begin{equation}
mG^{(\nu)}_{\rm GSD} = \{\sc{C}\in \text{SET}_{G,\nu} \mid |\sc{C}|\leq |\sc{C}'| \ \forall \sc{C}'\in \text{SET}_{G,\nu}\},
\end{equation}
where $|\sc{C}|$ is the number of anyons in the TO $\sc{C}$.
\end{definition}

\begin{lemma}
Suppose that there exists some SET $\sc{C}\in mG^{(\nu)}_{\rm GSD}$ such that $\sc{A} = \sc{C}$ i.e., $\sc{C}$ is Abelian. Then every SET in $mG^{(\nu)}_{\sc{D}}$ is Abelian. In fact $mG^{(\nu)}_{\sc{D}}$ is precisely the subset of Abelian TOs in $mG^{(\nu)}_{\rm GSD}$.
\label{lem:minimal_class}
\end{lemma}
\begin{proof}
Consider $\sc{C}'\in mG^{(\nu)}_{\sc{D}}$, by definition we know that $\sc{D}_{\sc{C}'}^2 \leq \sc{D}_{\sc{C}}^2 = |\sc{C}|$ since $\sc{C}$ is Abelian. Since $|\sc{C}'|\leq \sc{D}_{\sc{C}'}^2$ we see that $|\sc{C}'|\leq \sc{D}_{\sc{C}'}^2\leq |\sc{C}|$. But since $\sc{C}\in mG^{(\nu)}_{\rm GSD}$ we must have $|\sc{C}|\leq |\sc{C}'|$. These inequalities are thus equalities. This means $\sc{D}_{\sc{C}'}^2 = |\sc{C}'|$, so $\sc{C}'$ is an Abelian TO. It also means that $|\sc{C}'|=|\sc{C}|$, so $\sc{C}'\in mG^{(\nu)}_{\rm GSD}$ as well.

On the other hand if $\sc{C}'\in mG^{(\nu)}_{\rm GSD}$ is Abelian, then we see that $\sc{D}_{\sc{C}'}^2 = |\sc{C}'|\leq |\sc{C}''|\leq \sc{D}_{\sc{C}''}^2$ for every $\sc{C}''\in SET_{G,\nu}$. Thus $\sc{C}'\in mG^{(\nu)}_{\sc{D}}$. We therefore conclude that $mG^{(\nu)}_{\sc{D}}$ is precisely the subset of Abelian TOs in $mG^{(\nu)}_{\rm GSD}$.
\end{proof}

\begin{remark}
Note that an obvious corrollary of this is that if $mG^{(\nu)}_{\rm GSD}$ contains only Abelian TOs, then $mG^{(\nu)}_{\rm GSD} = mG^{(\nu)}_{\sc{D}}$.
\end{remark}

\begin{remark}
The opposite direction is not true i.e., it may be the case that $mG^{(\nu)}_{\sc{D}}$ contains an Abelian SET, but every SET in $mG^{(\nu)}_{\rm GSD}$ is non-Abelian. An example of this is fermionic TOs with $\U_f$ symmetry, with the requirement that the Hall condutivity $\sigma_H=1/2$. The set $mG^{(\sigma_H)}_{\sc{D}}$ contains multiple theories including the Moore-Read state, $\U_8$ and others, all with $\sc{D}^2=16$. However, one can show that all theories in $mG^{(\sigma_H)}_{\rm GSD}$ have anyon count 12, saturated by the Moore-Read state.
Moreover, as the next proof will make clear, all of these theories must be non-Abelian. 
\end{remark}

\begin{corollary}
Suppose that there exists some SET $\sc{C}\in mG^{(\nu)}_{\sc{D}}$ such that $\sc{D}_{\sc{C}}^2>|\sc{C}|$ i.e., $\sc{C}$ is non-Abelian. Then every SET in $mG^{(\nu)}_{\rm GSD}$ is non-Abelian.
\end{corollary}
\begin{proof}
This follows from the negation of Lem.~\ref{lem:minimal_class}.
\end{proof}

\section{Parton construction}
\label{app:parton_args}

The use of Cooper pairs to construct the states in Section \ref{sec:construction} may not seem energetically natural if the dominant interaction is repulsive. However alternate constructions of states with the same SET exist. We consider a standard parton decomposition of the physical spinless fermion annihilation operator $c_r$ at each lattice site $r$
\begin{equation}
 c_r = f_r \Phi_r
\end{equation}
Here $f_r$ and $\Phi_r$ are fermionic and bosonic partons. We take the boson to have the physical electric charge and the fermions to be electrically neutral. As usual these partons are coupled to a dynamical $\U$ gauge field. Both $f$ and $\Phi$ will be at the same mean filling $\nu$ as the electrons. Now consider a state where the bosons have formed a gapped $\bb{Z}_q$ topologically ordered state. A simple option for the fermions is that they form a Fermi surface. The resulting state is a gapless incompressible insulator with a neutral Fermi surface which is closely analogous to the spinon Fermi surface of spinful electronic Mott insulator at half filling. It is readily seen that the low-energy theory of this state is described by a field theory with the structure
\begin{equation}
{\cal L} = {\cal L}[f, a] + \frac{1}{2\pi} (a + A)  db - \frac{q}{2\pi} bdc
\end{equation}
Here ${\cal L}[f,a]$ is the Lagrangian of a Fermi surface coupled to a dynamical $\U$ gauge field\footnote{To be  precise, we should  regard $a$ as a spin$_c$ connection to signify that fields that couple to it with unit charge are fermions. The internal gauge fields $b$ and $c$ are ordinary $\U$ gauge fields while the background gauge field $A$ is a spin$_c$ connection.} $a$. This theory is coupled to a topological sector (described by the internal $\U$ gauge fields $b$ and $c$) associated with the $\bb{Z}_q$ gauge theory formed out of the $\Phi$ partons. 

We can now use this gapless theory as a parent state to construct gapped topological states. The simplest option is to imagine pairing the $f$-fermions. The resulting state is described by a Topological Quantum Field Theory (TQFT) with the action 
\begin{equation}
{\cal L} = -\frac{\hat{b} da}{\pi} + \frac{1}{2\pi} (a + A)  db - \frac{q}{2\pi} bdc \label{eqn:TQFT}
\end{equation}
where $\hat{b}$ is a new dynamical $\U$ gauge field.\footnote{$\hat{b}$ is an ordinary gauge field and not a spin$_c$ connection.} This TQFT actually is exactly the same phase as the $\bb{Z}_{2q}$ gauge theory constructed earlier out of Cooper pairs. We now demonstrate this.

The simplest way to see this is to integrate out the gauge field $a$, leaving a constrain $b=2\hat{b}$. Then we find
\begin{equation}
    \sc{L}=- \frac{2q}{2\pi}\hat{b}dc-\frac{2}{2\pi}Ad\hat{b}.
\end{equation}
This is exactly the Lagrangian for the $\bb{Z}_{2q}$ gauge theory formed out of Cooper pairs with physical electrical charge $2$.  

This parton construction does not give us the $\bb{Z}_q$ gauge theory that we argue is possible at odd $q$. However, we can access that state by a different parton construction. We illustrate this at a filling $\nu = 1/3$.  We (schematically) write the fermion as a product of 3 other fermions:
\begin{equation}
c_r = d_{1r}d_{2r}d_{3r}
\end{equation}
We assign physical charge $1/3$ to each of these 3 fermions. This representation comes with an $\mathrm{SU}(3)$ gauge redundancy. Now consider a `Higgs' condensate where 
\begin{equation}
T_{ab} = \langle d^\dagger_{ar} d_{br} \rangle ~~\neq 0
\end{equation}
such that the $\mathrm{SU}(3)$ is broken down to $\bb{Z}_3$. Then the low-energy theory will be a deconfined $\bb{Z}_3$ gauge theory where the elementary gauge charge is a fermion, and the gauge flux is a boson. The Higgs condensation also implies that the three fermion species $d_a$ (with $a = 1,2,3$) can all be identified with each other, and therefore there is just a single fermion species. These $d$-fermions are then at integer filling and can form a band insulator. Thus we get a gapped deconfined $\bb{Z}_3$ gauge theory with fermionic charges that carry physical electric charge $1/3$.  

The state just described is actually identical to the topologically ordered state constructed out of Cooper pairs of holes earlier. To see this note that we can always bind the local fermion to flip the statistics of the  $\bb{Z}_3$ gauge charge between boson and fermion. Thus if we start with the $\bb{Z}_3$ gauge theory with fermionic charges with physical charge $1/3$, binding the local fermion gives us a theory where the $\bb{Z}_3$ charge is a boson, and carries physical charge $-2/3$. The square of this is another bosonic anyon with physical charge $-4/3$, and the local excitation created by fusing three $\bb{Z}_3$ charges is just a charge $-2$ Cooper pair. Thus we get precisely the $\bb{Z}_3$ TO of the Cooper pair of holes. 

\section{Results about Abelian TOs}
\label{app:abelian_results}

In this section we prove a number of results about Abelian TOs. Most of the proofs rely only on basic facts from group theory i.e., the first isomorphism theorem and Lagrange's theorem. They are as self-contained as possible. We later find these proofs useful even when we consider non-Abelian TOs $\sc{C}$ because we will be able to make arguments to apply them to the subset of Abelian anyons $\sc{A}\subseteq \sc{C}$. Moreover, once we have proven that all minimal TOs are Abelian all of these results will be able to be used in full to constrain the possible Abelian TOs.

The structure of this appendix is as follows. In Appendix \ref{app:sub_condensation} we introduce our group-theoretic approach by discussing condensation in an Abelian TO. In Appendix \ref{app:sub_time} we introduce a single symmetry, time reversal, and give some results about the size of a time reversal invariant Abelian TO. In Appendix \ref{app:sub_charge_time} we introduce a second symmetry, and consider Abelian SETs of $\U_f\rtimes \bb{Z}_2^T$.

\subsection{Results about Abelian TO and anyon condensation}
\label{app:sub_condensation}

Let $\sc{A}$ be a (bosonic or fermionic) parent Abelian TO. We suppose that there is a boson $v\in \sc{A}$ which has order $n_v$. We then want to form a child TO $\sc{A}'$ from $\sc{A}$ by condensing $v$. We give a more precise definition of condensation
and show that $|\sc{A}| = n_v^2|\sc{A}'|$. Much of this section follows ideas discussed in Ref.~\cite{burnell_anyon_2018}, but we reformulate them in terms of group-theoretic language to build fluency. The tools developed here will prove useful once we introduce symmetries. 

We now give a number of definitions. We have implicitly used some of them before, but we restate every necessary fact here for completeness.

\begin{definition}
The braiding phase $\theta \colon \sc{A}\times \sc{A}\rightarrow \bb{S}^1$ is a symmetric bilinear form i.e., it has the properties that 
\begin{equation}
\theta_{ab,d} = \theta_{a,d}+\theta_{b,d} \text{ and } \theta_{a,b} = \theta_{b,a}
\end{equation}
for all $a,b,d\in \sc{A}$. Additionally $\theta_{a,a}=2\theta_a$, where $\theta_a$ is the topological spin of $a\in \sc{A}$.
\end{definition}

\begin{definition}
We call $\sc{A}$ a \textit{bosonic} TO if braiding is nondegenerate on $\sc{A}$ i.e., if $\theta_{a,b} = 0$ for all $b\in \sc{A}$ implies that $a=1$, where $1$ is a boson. We call $\sc{A}$ a \textit{fermionic} TO if $\theta_{a,b} = 0$ for all $b\in \sc{A}$ implies that $a=1,c$ where $c^2=1$ and $\theta_c=\pi$ i.e., if braiding is nondegenerate up to the fundamental fermion.
\end{definition}

\begin{definition}
Let $\sc{A}$ be an Abelian TO. Then take $\sc{S}\subseteq \sc{A}$ to be a subgroup of $\sc{A}$. We define the \textit{complement} of $\sc{S}$ to be given by
\begin{equation}
\sc{S}^\perp = \{a \in \sc{A} \mid \theta_{a,b} = 0 \ \forall b\in \sc{S}\} \subseteq \sc{A}.
\end{equation}
\end{definition}

\begin{remark}
We note that if $a,d\in \sc{S}^\perp$ then $ad\in \sc{S}^\perp$ since $\theta_{ad,b} = \theta_{a,b}+\theta_{d,b} =0$ for all $b\in \sc{S}$. Thus $\sc{S}^\perp$ is a subgroup. In particular we note that if $\sc{A}$ is a fermionic TO, then $\{1,c\}\subseteq \sc{S}^\perp$ will be a normal subgroup of $\sc{S}^\perp$.

We further note that $\sc{S}^\perp$ looks very similar to the complement of a vector space equipped with an inner product, with the braiding phase playing the role of the inner product. However, it can behave differently. For example if we take $\sc{A}$ to be the usual $\bb{Z}_2$ gauge theory and $\sc{S} = \{1,e\} \subseteq \sc{A}$, then we can see that $\sc{S}^\perp = \sc{S}$. It will never be the case that the complement of a nontrivial vector space is equal to the vector space, but this behavior is allowed for the complement as we define it here.
\end{remark}

\begin{theorem}
If $\sc{A}$ is a bosonic Abelian TO then for any subgroup $\sc{S}\subseteq \sc{A}$ we must have that:
\begin{equation}
|\sc{A}| = |\sc{S}||\sc{S}^\perp|.
\end{equation}
\label{thm:character_theory}
\end{theorem}
\begin{proof}
We use a result from character theory to establish this. First we recall that a character $\chi$ of a finite Abelian group $\sc{A}$ is a homomorphism $\chi \colon \sc{A}\rightarrow \bb{S}^1$. The set of characters of $\sc{A}$ is labeled as $\mathrm{Hom}(\sc{A},\bb{S}^1)$, where this can be treated as an Abelian group under pointwise addition i.e., $(\chi_1+ \chi_2)(a) = \chi_1(a)+\chi_2(a)$ for $\chi_1, \chi_2 \in \mathrm{Hom}(\sc{A},\bb{S}^1)$ and $a\in \sc{A}$. A result from character theory tells us that the number of distinct characters of a finite Abelian group $\sc{A}$ is precisely $|\sc{A}|$ \cite{conrad_characters_nodate}. 

We know that $\theta_{a,\cdot} \colon \sc{A}\rightarrow \bb{S}^1$ is a homomorphism for every $a\in \sc{A}$. Furthermore, because $\sc{A}$ is a bosonic TO the braiding phase is nondegenerate. Thus if $a,b\in \sc{A}$ and $a\neq b$ then there exists some $d\in \sc{A}$ such that $\theta_{a,d}\neq \theta_{b,d}$. Then we see that $\theta_{a,\cdot} = \theta_{b,\cdot}$ as an element of $\mathrm{Hom}(\sc{A},\bb{S}^1)$ if, and only if, $a=b$. This means that we have found $|\sc{A}|$ distinct elements of $\mathrm{Hom}(\sc{A},\bb{S}^1)$, and by the logic above we see that for every character $\chi \in \mathrm{Hom}(\sc{A},\bb{S}^1)$ there exists one, and only one, $a\in \sc{A}$ such that $\chi = \theta_{a,\cdot}$.

We now want to consider a function $B\colon \sc{A} \rightarrow \mathrm{Hom}(\sc{S},\bb{S}^1)$ i.e., $B(a) \colon \sc{S} \rightarrow \bb{S}^1$ is a homomorphism for every $a\in \sc{A}$. We define it as $B(a)(b) = \theta_{a,b}$, where $b$ is restricted to lie in the subgroup $\sc{S}$. We note that $B(ad)=B(a)+B(d)$ for $a,d\in \sc{A}$ so $B$ is a homomorphism from $\sc{A}$ to $\mathrm{Hom}(\sc{S}, \bb{S}^1)$. Now suppose that $a \in \mathrm{Ker}(B)$, then $B(a)(b) = 0$ for all $b \in \sc{S}$. But then we see that $\mathrm{Ker}(B)\simeq \sc{S}^\perp$. Now consider $\chi \in \mathrm{Hom}(\sc{S},\bb{S}^1)\subseteq \mathrm{Hom}(\sc{A},\bb{S}^1)$, then by the logic above we know that there is some $a\in \sc{A}$ such that $\chi = B(a)$. Thus $\mathrm{Im}(B) \simeq \mathrm{Hom}(\sc{S},\bb{S}^1)$. The first isomorphism theorem tells us that
\begin{equation}
\sc{A}/\sc{S}^\perp \simeq \mathrm{Hom}(\sc{S},\bb{S}^1).
\end{equation}
From character theory we know that $|\mathrm{Hom}(\sc{S},\bb{S}^1)| = |\sc{S}|$ and the theorem follows immediately.
\end{proof}

\begin{example}
One simple example of this is the usual $\bb{Z}_2$ gauge theory. If $\sc{S} = \{1,e\}\subseteq \sc{A}$ then $\sc{S}^\perp = \sc{S}$. Clearly $|\sc{A}| = 4 = |\sc{S}||\sc{S}^\perp|$.
\end{example}

\begin{example}
Now consider the Abelian bosonic TO $\U_q\boxtimes \U_{-q}$ for $q$ even. If we take $\sc{S}$ to be $\U_q$ then we see that $\sc{S}^\perp$ will be $\U_{-q}$ and we obviously have that $|\sc{S}||\sc{S}^\perp| = q^2 = |\sc{A}|$.
\end{example}

\begin{lemma}
If $\sc{A}$ is a fermionic Abelian TO and $\sc{S}\subseteq \sc{A}$ is a subgroup then we have that:
\begin{equation}
|\sc{A}|=|\sc{S}||\sc{S}^\perp|\times \begin{cases} 1 &\mbox{if } c\notin \sc{S}\\ 1/2 &\mbox{if } c\in \sc{S}\end{cases}.
\end{equation}
\label{lem:fermion_char_theory}
\end{lemma}
\begin{proof}
We prove this by cases. Suppose that $c\notin \sc{S}$. This means that we can freely quotient by $\{1,c\}$ to obtain a bosonic TO $\sc{A}_b$ with $\sc{S}\subseteq \sc{A}_b$. Then using Thm.~\ref{thm:character_theory} we see that $|\sc{A}_b| = |\sc{S}||\sc{S}^\perp_b|$, where $\sc{S}^\perp_b = \{a \in \sc{A}_b \mid \theta_{a,b} = 0 \ \forall b\in \sc{S}\}\subseteq \sc{A}_b$. For every $a\in \sc{S}^\perp_b$ we must have that $ac$ also braids trivially with every element of $\sc{S}$, and thus $\sc{S}^\perp = \sc{S}^\perp_b\times \{1,c\}$. We thus see that $|\sc{A}| = 2|\sc{A}_b|$ and $|\sc{S}^\perp| = 2|\sc{S}^\perp_b|$, so we have $|\sc{A}| = 2|\sc{S}||\sc{S}^\perp_b| = |\sc{S}||\sc{S}^\perp|$.

Now suppose that $c\in \sc{S}$, we again quotient by $\{1,c\}$. Using Thm.~\ref{thm:character_theory} we then see that $|\sc{A}_b|=|\sc{S}_b||\sc{S}_b^\perp|$. By the logic above we have that $|\sc{S}^\perp| = 2|\sc{S}_b^\perp|$, $|\sc{A}|=2|\sc{A}_b|$. Now we also have that $|\sc{S}|=2|\sc{S}_b|$. Thus we see that $|\sc{A}|=|\sc{S}||\sc{S}^\perp|/2$.
\end{proof}

\begin{corollary}
Let $\sc{A}$ be a bosonic Abelian TO with a subgroup $\sc{S}\subseteq \sc{A}$. Then $(\sc{S}^\perp)^\perp = \sc{S}$. If instead $\sc{A}$ is a fermionic Abelian TO then $(\sc{S}^\perp)^\perp = \sc{S}\times \{1,c\}$.
\label{cor:perp_of_perp}
\end{corollary}
\begin{proof}
Let $a\in \sc{S}$, then $0 = \theta_{b,a} = \theta_{a,b}$ for every $b\in \sc{S}^\perp$ by definition. So $\sc{S}\subseteq (\sc{S}^\perp)^\perp$ by definition. Thus we must have that $|(\sc{S}^\perp)^\perp| \geq |\sc{S}|$ with equality if, and only if, $(\sc{S}^\perp)^\perp = \sc{S}$.

Now suppose that $\sc{A}$ is a bosonic TO. Then by Thm.~\ref{thm:character_theory} we know that $|\sc{A}| = |\sc{S}||\sc{S}^\perp|$ and $|\sc{A}| = |\sc{S}^\perp||(\sc{S}^\perp)^\perp|$ since $\sc{S}^\perp$ is a subgroup of $\sc{A}$. Equating these and dividing by $|\sc{S}^\perp|$ reveals that $|(\sc{S}^\perp)^\perp| = |\sc{S}|$. From earlier we see this must mean that $(\sc{S}^\perp)^\perp = \sc{S}$.

Suppose now that $\sc{A}$ is a fermionic TO. Then we automatically see that $\sc{S}\times \{1,c\} \subseteq (\sc{S}^\perp)^\perp$. If $c\notin \sc{S}$ then this means that $|(\sc{S}^\perp)^\perp|\geq 2|\sc{S}|$, with equality if, and only if, $(\sc{S}^\perp)^\perp = \sc{S}\times \{1,c\}$. If $c\in \sc{S}$ then we have that $\sc{S}\times \{1,c\} = \sc{S}$ and thus $|(\sc{S}^\perp)^\perp|\geq |\sc{S}|$ with equality if, and only if, $(\sc{S}^\perp)^\perp = \sc{S} = \sc{S}\times \{1,c\}$.

Suppose first that $c\notin \sc{S}$, then from Lem.~\ref{lem:fermion_char_theory} we have that $|\sc{A}| = |\sc{S}||\sc{S}^\perp|$. Now since $\sc{S}^\perp$ is a subgroup of $\sc{A}$ with $c\in \sc{S}^\perp$ then we also know by Lem.~\ref{lem:fermion_char_theory} that $|\sc{A}| = |\sc{S}^\perp||(\sc{S}^\perp)^\perp|/2$. Equating these we see that $|(\sc{S}^\perp)^\perp| = 2|\sc{S}|$ and thus $(\sc{S}^\perp)^\perp = \sc{S}\times \{1,c\}$. On the other hand if $c\in \sc{S}$ then Lem.~\ref{lem:fermion_char_theory} shows that $|\sc{A}|=|\sc{S}||\sc{S}^\perp|/2$ and that $|\sc{A}| = |\sc{S}^\perp||(\sc{S}^\perp)^\perp|/2$. Equating these reveals that $|(\sc{S}^\perp)^\perp| = |\sc{S}|$ and thus $(\sc{S}^\perp)^\perp = \sc{S} = \sc{S}\times \{1,c\}$.
\end{proof}

\begin{definition}
Let $\sc{A}$ be a bosonic Abelian TO with two subgroups $\sc{S}$ and $\sc{R}$. Further suppose that braiding is nondegenerate when restricted to these subgroups and that $\theta_{r,s}=0$ for all $r\in \sc{R}$ and all $s\in \sc{S}$. Then we say that $\sc{A}$ factorizes as a TO and write $\sc{A} = \sc{S}\boxtimes \sc{R}$.
\end{definition}

\begin{theorem}
    Suppose $\sc{A}$ is a bosonic Abelian TO, and $\sc{S}$ is a subgroup. If the braiding restricted to $\sc{S}$ is nondegenerate, then $\sc{A}=\sc{S}\boxtimes \sc{S}^\perp$.
\end{theorem}
\begin{proof}
    First, we show that the braiding on $\sc{S}^\perp$ is also nondegenerate. If it is not, then there exists $x\in \sc{S}^\perp, x\neq 1$ such that $x\in (\sc{S}^\perp)^\perp$. However, by Collorary 
 \ref{cor:perp_of_perp} $(\sc{S}^\perp)^\perp=\sc{S}$, thus $x\in \sc{S}$ as well, which contradicts the assumption that $\sc{S}$ has nondegenerate braiding. The argument also shows that $\sc{S}\cap \sc{S}^\perp=\{1\}$.

    Next, consider the Abelian TO $\sc{A}'$ generated by $\sc{S}$ and $\sc{S}^\perp$. It is easy to see that $\sc{A}'=\sc{S}\boxtimes \sc{S}^\perp$, and $|\sc{A}'|=|\sc{S}||\sc{S}^\perp|$. On the other hand, $\sc{A}'$ is a subgroup of $\sc{A}$, but from Theorem \ref{thm:character_theory} we know $|\sc{A}|=|\sc{S}||\sc{S}^\perp|$. We thus conclude that $\sc{A}'=\sc{A}$ and we are done.
\end{proof}

\begin{remark}
    The theorem is a special case of a more general result about factorization of modular tensor categories (MTC), Theorem 3.13 in Ref.~\cite{drinfeld_braided_2010}. Namely, if a MTC $\sc{C}$ has a modular subcategory $\sc{B}$, then the MTC can be factorized as $\sc{C}=\sc{B}\boxtimes \sc{B}^\perp$, where $\sc{B}^\perp$ is the ``complement'' of $\sc{B}$ in $\sc{C}$.
\end{remark}

\begin{definition}
Let $\sc{A}$ be a bosonic or fermionic Abelian TO. Further take $v\in \sc{A}$ to be a boson and let $\langle v\rangle = \{1,v,\ldots, v^{n_v-1}\}\subseteq \sc{A}$. Then define $\sc{A}' = \langle v\rangle^\perp/\langle v\rangle$ as the TO formed from $\sc{A}$ by condensing $v$.
\end{definition}

\begin{remark}
Let us discuss the logic of this definition informally. Since we want to identify $v$ with the vacuum (after all this is what it means to condense $v$ \cite{burnell_anyon_2018}) we need to make sure our TO does not contain anything that braids nontrivially with $v$. Thus we must consider $\langle v\rangle^\perp$, the subgroup of all anyons that braid trivially with $v$. This obviously is not yet a TO since every element of it will braid trivially with every power of $v$. We thus consider the group $\langle v\rangle^\perp/\langle v\rangle$ which will identify $v$ with the vacuum. We formally prove this is a proper TO below.
\end{remark}

\begin{theorem}
If $\sc{A}$ is a bosonic (fermionic) Abelian TO with a boson $v\in \sc{A}$ then the TO formed from $\sc{A}$ by condensing $v$ is a bosonic (fermionic) Abelian TO. Moreover, we have that
\begin{equation}
|\sc{A}| = n_v^2|\sc{A}'|,
\end{equation}
where $n_v$ is the order of $v\in \sc{A}$.
\label{thm:condensation_relation}
\end{theorem}
\begin{proof}
First suppose that $\sc{A}$ is a bosonic Abelian TO. Since $v$ is a boson it braids trivially with itself, and thus will with all of its powers. So we have that $\langle v\rangle$ is a subgroup of $\langle v\rangle^\perp\subseteq \sc{A}$. Thus $\sc{A}' = \langle v\rangle^\perp/\langle v\rangle$ is well defined. Now we use Cor.~\ref{cor:perp_of_perp} to see that $(\langle v\rangle^\perp)^\perp = \langle v\rangle$. Thus if $a\in \sc{A}$ braids trivially with every element of $\langle v\rangle^\perp$ it must lie in $\langle v\rangle$. Since in $\sc{A}'$ every element of $\langle v\rangle$ is identified with the vacuum, we must then have that braiding is nondegenerate on $\sc{A}'$. Thus $\sc{A}'$ is a bosonic TO.

Since we have that $\sc{A}' = \langle v\rangle^\perp/\langle v\rangle$ we know that $|\langle v\rangle^\perp | = |\langle v\rangle||\sc{A}'| = n_v|\sc{A}'|$. From Thm.~\ref{thm:character_theory} we must have that $|\sc{A}| = |\langle v\rangle||\langle v\rangle^\perp | = n_v|\langle v\rangle^\perp | = n_v^2 |\sc{A}'|$. We therefore have the proof.

Suppose now that $\sc{A}$ is a fermionic Abelian TO. We first note that since $v$ is a boson no power of $v$ can be a fermion and hence $c\notin \langle v\rangle$. For the same reasons as above we have that $\sc{A}' = \langle v\rangle^\perp/\langle v\rangle$ is well-defined. By Cor.~\ref{cor:perp_of_perp} we see that $(\langle v\rangle^\perp)^\perp = \langle v\rangle \times \{1,c\}$. Then we see that braiding is nondegenerate on $\sc{A}'$ up to the fundamental fermion $c$. Thus $\sc{A}'$ is a fermionic TO. By Lem.~\ref{lem:fermion_char_theory} we see that $|\sc{A}| = |\langle v\rangle||\langle v\rangle^\perp| = n_v^2|\sc{A}'|$.
\end{proof}

\begin{example}
A simple example of this is given by $\bb{Z}_n$ gauge theory. There we have a boson $e\in \sc{A}$ with order $n$. If we condense $e$ then the resulting theory is trivial i.e., $\sc{A}' = \{1\}$. We then have that $|\sc{A}|=n^2|\sc{A}'|$.
\end{example}

\subsection{Time reversal invariant Abelian TOs}
\label{app:sub_time}

A particularly powerful algebraic result about time reversal invariant Abelian bosonic TOs was given in Ref.~\cite{lee_study_2018}. We state this result as a theorem, although we must first give a definition of time-reversal invariant Abelian TOs. At the end of this section we extend the results of Ref.~\cite{lee_study_2018} to include fermionic Abelian TOs as well as strengthen the bounds on anyon count for Abelian bosonic TOs.

\begin{definition}
Let $\sc{A}$ be an Abelian TO. We call $\sc{A}$ a time-reversal invariant Abelian TO if it has an automorphism $\sc{T} \colon \sc{A} \rightarrow \sc{A}$ with the following properties. First, $\sc{T}^2$ is the identity on $\sc{A}$. Also, if $a\in \sc{A}$ then $\theta_{\sc{T}a} = -\theta_a$ so time reversal takes the topological spin to its negative. Finally, if $\sc{A}$ is a fermionic TO then $\sc{T}c = c$ i.e., time reversal sends the fundamental fermion to itself.
\end{definition}

\begin{remark}
In this section we have not specified whether $c$ is a Kramers doublet under $\sc{T}$, although the anomaly structure of our theory will depend on this \cite{lapa_anomaly_2019}. We further note that other studies of the action of time reversal on Abelian TOs have allowed $\sc{T}^2 = C$, where $C$ is charge conjugation \cite{delmastro_symmetries_2021}. However, this action of time reversal is not consistent with $\U$ charge symmetry. Thus we restrict $\sc{T}^2$ to be the identity on $\sc{A}$.
\end{remark}

\begin{theorem}
Suppose that $\sc{A}_b$ is a time-reversal invariant Abelian bosonic TO. Then $|\sc{A}_b| = l^2$ for some integer $l$.
\label{thm:size_is_square}
\end{theorem}
\begin{proof}
This proof follows immediately from (3.11) of Ref.~\cite{lee_study_2018}, which showed that $|\sc{A}_b|=l^2$ for some integer $l$.

We will review the derivation of $|\sc{A}_b|=l^2$ for clarity, since only part of the discussion in Ref.~\cite{lee_study_2018} is necessary for our purposes. Define:
\begin{align}
f\colon \sc{A}_b &\rightarrow \sc{A}_b\\
 a &\mapsto a\times \sc{T}a, \nonumber
\end{align}
 Further define:
\begin{align}
g\colon \sc{A}_b &\rightarrow \sc{A}_b\\
 a &\mapsto a\times\sc{T}\overline{a}, \nonumber
\end{align}
where $\overline{a} = a^{-1}$ is the dual of $a$. Suppose that $a \in \mathrm{Im}(f)$ i.e., $a=b\times\sc{T}b$ for some $b\in \sc{A}_b$. Then we see that $\sc{T}a = a$ since $\sc{T}$ squares to the identity. Thus we see that $a\in \mathrm{Ker}(g)$ and hence that $\mathrm{Im}(f)$ is a subgroup of $\mathrm{Ker}(g)$. We can thus define:
\begin{equation}
\sc{C} = \mathrm{Ker}(g)/\mathrm{Im}(f).
\end{equation}
Now suppose $a \in \mathrm{Ker}(g)$, then $\sc{T}a=a$ by definition. Then we see that $a^2 = a\times\sc{T}a$ and thus $a^2\in \mathrm{Im}(f)$. Thus every nontrivial element of $\sc{C}$ must have order two and hence $\sc{C} = \bb{Z}_2^n$ for some integer $n$.

We now use the fact that $\theta_{\sc{T}a,\sc{T}b} = -\theta_{a,b}$ for all $a,b\in \sc{A}_b$. Note that this means that $\theta_{a,\sc{T}b} = -\theta_{\sc{T}a,b} = \theta_{\sc{T}\overline{a},b}$. If we add $\theta_{a,b}$ to both sides then we can use its bilinearity to see that $\theta_{a,b\times\sc{T}b} = \theta_{a\times\sc{T}\overline{a},b}$ i.e., that $\theta_{a,f(b)}=\theta_{g(a),b}$. Then consider $a\in \mathrm{Ker}(g)$ we see that $0 = \theta_{1,b} = \theta_{g(a),b} = \theta_{a,f(b)}$ for every $b\in \sc{A}_b$. Similarly we see that for $b\in \mathrm{Ker}(f)$ we have $0 = \theta_{g(a),b}$ for all $a\in \sc{A}_b$. We thus have
\begin{equation}
\mathrm{Ker}(g)\subseteq \mathrm{Im}(f)^\perp \text{ and } \mathrm{Ker}(f)\subseteq \mathrm{Im}(g)^\perp.
\end{equation}
We then see that:
\begin{align}
|\mathrm{Ker}(g)| &\leq |\mathrm{Im}(f)^\perp| \text{ by above}\\
&= \frac{|\sc{A}_b|}{|\mathrm{Im}(f)|} \text{by Thm.~\ref{thm:character_theory}}\\
|\mathrm{Ker}(f)| &\leq \frac{|\sc{A}_b|}{|\mathrm{Im}(g)|} \text{ similarly.}
\end{align}
Now we know by the first isomorphism theorem that $|\sc{A}_b|=|\mathrm{Ker}(g)||\mathrm{Im}(g)|=|\mathrm{Ker}(f)||\mathrm{Im}(f)|$. Substituting these into the above relations reveals that the above inequalities are equalities and thus we have that
\begin{equation}
\mathrm{Ker}(g) = \mathrm{Im}(f)^\perp \text{ and } \mathrm{Ker}(f) = \mathrm{Im}(g)^\perp.
\end{equation}
Then we see that $\sc{C} = \mathrm{Im}(f)^\perp/\mathrm{Im}(f)$. By Thm.~\ref{thm:character_theory} we know that $|\sc{A}_b| = |\mathrm{Im}(f)||\mathrm{Im}(f)^\perp| = |\mathrm{Im}(f)|^2|\sc{C}|$. Thus to prove that $|\sc{A}_b|$ is a square, all we need to show is that $|\sc{C}|$ is.

By Cor.~\ref{cor:perp_of_perp} we know that $(\mathrm{Im}(f)^\perp)^\perp = \mathrm{Im}(f)$. Thus the only elements $a\in \sc{A}_b$ that braid trivially with all elements of $\mathrm{Im}(f)^\perp$ are those in $\mathrm{Im}(f)$. This means that since $\sc{C} = \mathrm{Im}(f)^\perp/\mathrm{Im}(f)$ we must have that braiding is nondegenerate on $\sc{C}$.

We next want to argue that $\theta_{a,a}=0$ for every element in $a\in \sc{C}$. First consider $a\in \mathrm{Ker}(g)$, where $\sc{T}a=a$. Then we have that:
\begin{equation}
\theta_a = \theta_{\sc{T}a} = -\theta_a \pmod{2\pi} \label{eqn:top_spin_trs},
\end{equation}
and thus $0 = 2\theta_a = \theta_{a,a}$ modulo $2\pi$. We thus have that $\theta_{a,a}=0$ for every $a\in \sc{C}=\mathrm{Ker}(g)/\mathrm{Im}(f)$. Now we note that $2\theta_{a,b}=\theta_{a^2,b}=0$ since $a\in \sc{C}$ has order two.\footnote{Note that here we are working in $\sc{C}$ so $a^2 =f(a)\sim 1$.} Thus $\theta_{a,b}=0,\pi$ for every $a,b\in \sc{C}$. 

Now take some $a\in \sc{C}$, then $\theta_{a,a} = 0$ as demonstrated above. Since braiding is nondegenerate on $\sc{C}$ there must be some $b\in \sc{C}$ such that $\theta_{a,b} = \pi$. Then take $\sc{S} = \{1,a,b,ab\}\subseteq \sc{C}$. Then by Thm.~\ref{thm:character_theory} we see that $|\sc{C}|=|\sc{S}||\sc{S}^\perp|=4|\sc{S}^\perp|$. We further know by Corr.~\ref{cor:perp_of_perp} that $(\sc{S}^\perp)^\perp = \sc{S}$, so the only elements which braid trivially with everything in $\sc{S}^\perp$ lie in $\sc{S}$. But every nontrivial element of $\sc{S}$ braids nontrivially with at least one other element of $\sc{S}$ and thus $\sc{S}\cap \sc{S}^\perp = \{1\}$. We conclude that braiding is nondegenerate when restricted to $\sc{S}^\perp$. Thus $\sc{S}^\perp\subseteq \sc{C}$ is itself a bosonic TO. We can then find some $\sc{S}'=\{1,a',b',a'b'\}\subseteq \sc{S}^\perp$ where $\theta_{a',b'}=\pi$ and again apply Thm.~\ref{thm:character_theory} to see that $|\sc{S}^\perp| = 4|\sc{S}^{\prime \perp}|$, where $\sc{S}^{\prime \perp}\subseteq \sc{S}^\perp$ is another bosonic TO. Continuing to apply this logic until we have exhausted the finite set $\sc{C}$ reveals that $|\sc{C}| = 4^m = 2^{2m}$ for some $m$. Thus $|\sc{C}|$ is a square and we see $|\sc{A}_b| = |\mathrm{Im}(f)|^2|\sc{C}|$ must also be a square.
\end{proof}

\begin{lemma}
Suppose that $\sc{A}$ is a time-reversal invariant fermionic Abelian TO. Then we either have that $|\sc{A}|=l^2$ or $2l^2$. Additionally, if there is no anyon $a\in \sc{A}$ such that $\sc{T}a=ac$ then $|\sc{A}|=2l^2$ for some integer $l$.
\label{lem:size_is_square}
\end{lemma}
\begin{proof}
If $\sc{A}$ is a time-reversal invariant fermionic Abelian TO then we write $\sc{A}_b = \sc{A}/\{1,c\}$. Braiding will become nondegenerate on this TO. This is thus a bosonic TO with a $\pi$-ambiguity in the topological spin, since we have identified the fundamental fermion $c$ with the vacuum. However, there is no ambiguity of braiding.

We can now repeat the initial arguments of Thm.~\ref{thm:size_is_square} to see that $|\sc{A}_b| = |\mathrm{Im}(f)|^2|\sc{C}|$, where $\sc{C} = \mathrm{Ker}(g)/\mathrm{Im}(f)$. We continue to have that every element of $\sc{C}$ is order two, so $\sc{C}=\bb{Z}_2^n$ for some $n$. These results did not rely on topological spin, only on braiding, so they continue to hold in $\sc{A}_b$. Then we know that $|\sc{A}|=2|\sc{A}_b|=|\mathrm{Im}(f)|^22^{n+1}$. If $n$ is odd then $|\sc{A}|=l^2$ for some $l$, while if $n$ is even then $|\sc{A}|=2l^2$ for some $l$.

We need to work a little harder to prove the second result since Eq.~\eqref{eqn:top_spin_trs} \textit{did} rely on topological spin. We now revisit the argument surrounding this equation. Suppose that $a \in \mathrm{Ker}(g)$ i.e., that $\sc{T}a = a$ (modulo $c$). Then we have that
\begin{equation}
\theta_a = \theta_{\sc{T}a} = - \theta_a \text{ or }  -\theta_a + \pi \pmod{2\pi},
\end{equation}
since we have a $\pi$ ambiguity in the topological spin. Then we see that $\theta_{a,a}=2\theta_a = 0,\pi \pmod{2\pi}$. Clearly $\theta_{a,a} = \pi$ if, and only if, $\sc{T}a=ac$. Suppose then that there is no $a\in \sc{A}$ such that $\sc{T}a=ac$, then $\theta_{a,a} = 0$ for every $a\in \sc{C}$. We thus thus repeat the same argument as in Thm.~\ref{thm:size_is_square} to see that $|\sc{C}|=2^{2m}$ for some $m$. We then conclude that $|\sc{A}_b| = |\mathrm{Im}(f)|^2|\sc{C}|= l^2$ for some $l$ and hence that $|\sc{A}|=2l^2$.

\end{proof}

\begin{example}
Consider the semion-fermion fermionic Abelian TO, $\sc{A}=\U_2\boxtimes\{1,c\}$. If the semion in this theory is referred to as $a$ then a consistent action of time reversal is $\sc{T}a = ac$. Thus by the above theorem the size can be a square, and indeed $|\sc{A}|=4$.
\end{example}

\begin{remark}
We note that the semion-fermion fermionic Abelian TO is well known to be anomalous \cite{fidkowski_non-abelian_2013, wang_interacting_2014}. Even more clearly, if we impose $\U_f\rtimes \bb{Z}_2^T$ symmetry then we know that $Q_a = Q_{\sc{T}a} = Q_a+Q_c=Q_a+1\pmod{2}$, which is a contradiction. Thus we conclude that every fermionic Abelian TO with symmetry group $\U_f\rtimes \bb{Z}_2^T$ must have size $|\sc{A}|=2l^2$ for some $l$. We nonetheless see that the other possibility will be useful when we generalize some results to the non-Abelian case.
\end{remark}

\begin{lemma}
Let $\sc{A}$ be a time-reversal invariant Abelian bosonic TO. Suppose that it has an anyon $a\in \sc{A}$ which has order $n$ i.e., $n$ is the smallest integer such that $a^n = 1$. Then we have that
\begin{equation}
|\sc{A}| = (kn)^2.
\end{equation}
for some positive integer $k$. In particular this means $ |\sc{A}|\geq n^2$.
\label{lem:bound_b}
\end{lemma}
\begin{proof}
We first establish something simple. Suppose that $1 = (\sc{T}a)^m = \sc{T}a^m$, then we can use the fact that $\sc{T}^2$ is the identity and $\sc{T}1 = 1$ to see that $a^m = 1$. Thus $m$ must be a multiple of $n$. We further note that $(\sc{T}a)^{n} = \sc{T}a^{n} = \sc{T}1 = 1$. We may then conclude that the order of $\sc{T}a$ is exactly $n$ i.e., equal to the order of $a$.

Next we note that $\theta_{a,\sc{T}a} = 0,\pi$ since $\theta_{b,d} = -\theta_{\sc{T}b,\sc{T}d}$ for $b,d\in \sc{A}$.

Suppose that $0 = \theta_{a,\sc{T}a}$. This will mean that $\theta_{a^n, \sc{T}a^m} = nm\theta_{a,\sc{T}a} = 0$, so every power of $a$ will braid trivially with every power of $\sc{T}a$. Then if $\sc{S} = \{1,a,\ldots, a^{n-1}\}\subseteq \sc{A}$ is the subgroup of powers of $a$, we see that $\sc{R} = \{1,\sc{T}a,\ldots, \sc{T}a^{n-1}\}\subseteq \sc{S}^\perp$. Moreover, since $\sc{R}$ is a subgroup of $\sc{S}^\perp$ we know that $|\sc{S}^\perp|$ must be a multiple of $|\sc{R}|$ by Lagrange's theorem i.e., $|\sc{S}^\perp| = mn$ for some nonzero integer $m$. By Thm.~\ref{thm:character_theory} we then have that:
\begin{equation}
|\sc{A}| = |\sc{S}||\sc{S}^\perp| = mn^2.
\end{equation}
By Thm.~\ref{thm:size_is_square} we must have that $|\sc{A}| =l^2$ for some $l\in \bb{Z}$. Thus we must have that $l^2 = mn^2$. Then $\sqrt{m} = l/n$, but the only integers whose square root is a rational number are the squares, so $m=k^2$ and thus $|\sc{A}| = (nk)^2$. 

Now suppose that $\theta_{a,\sc{T}a}=\pi$. Then we must have that $\theta_{a,\sc{T}a^2} = 2\theta_{a,\sc{T}a} = 0$, so $a$ braids trivially with $\sc{T}a^2$. Next suppose, by contradiction, that $n$ is odd, then $n+1$ is even and $(n+1)/2$ is an integer. Then we can see that $0 = [(n+1)/2]\theta_{a,\sc{T}a^2} = \theta_{a,\sc{T}a^{n+1}} = \theta_{a,\sc{T}a}$ since $\sc{T}a^n=1$. This contradicts $\theta_{a,\sc{T}a} = \pi$, however. Thus if $\theta_{a,\sc{T}a}=\pi$ then $n$ is even. Then the set $\sc{R} = \{1,\sc{T}a^2, \sc{T}a^4,\ldots, \sc{T}a^{n-2}\}$ is a subgroup of $\sc{A}$ with $|\sc{R}|=n/2$. Further if we take $\sc{S} = \{1,a,\ldots, a^{n-1}\}\subseteq \sc{A}$ then we can see that $\sc{R}\subseteq \sc{S}^\perp$. Since $\sc{R}$ is a subgroup of the group $\sc{S}^\perp$ this tells us that $|\sc{S}^\perp|$ must be a multiple of $n/2$ by Lagrange's theorem i.e., $|\sc{S}^\perp| = mn/2$ for some nonzero integer $m$. By Thm.~\ref{thm:character_theory} we have that $|\sc{A}| = mn^2/2$. Then by Thm.~\ref{thm:size_is_square} there must be some integer $l$ such that $l^2 = mn^2/2$. Since $n$ is even we can write $n=2n'$ for some integer $n'$. Then we see that $l^2 = 2mn^{\prime 2}$ and thus $\sqrt{2m} = l/n'$. Since the only integers whose square root is a rational number are the squares we know that $2m=k^{\prime 2}$ for some integer $k'$. Then we see that $k^{\prime 2}$ is even and thus $k'$ is even. Thus $k' = 2k$ for some integer $k$ and $m=2k^2$. We thus have $|\sc{A}| = mn^2/2 = (nk)^2$.
\end{proof}

\begin{example}
Consider $\mathbb{Z}_n$ gauge theory. We know that it is time-reversal invariant under the standard action of time reversal. Further we know that there exists $e\in \sc{A}$ which has order $n$. We also know that $|\sc{A}|=n^2$. Thus there is at least one time reversal invariant theory which saturates the lower bound given by Lem.~\ref{lem:bound_b}.
\end{example}

\begin{lemma}
Let $\sc{A}$ be a time-reversal invariant fermionic Abelian TO. Suppose it has an anyon $a\in \sc{A}$ with order $n$ such that $a^k\neq c$ for any $k$. Then we have that
\begin{equation}
|\sc{A}| = (kn)^2 \ \mathrm{or} \ 2(nk)^2.
\end{equation}
for some positive integer $k$. If there is no anyon $a\in \sc{A}$ such that $\sc{T}a = ac$ then $|\sc{A}|=2(nk)^2$.
\label{lem:bound_f}
\end{lemma}
\begin{proof}
First we note that since $a^k\neq c$ for any $k$ then $\sc{S} = \{1,a,\ldots, a^{n-1}\}$ forms a subgroup of $\sc{A}_b=\sc{A}/\{1,c\}$ of size $n$. We work in this quotient for the rest of the proof. Recall that topological spin is then only defined modulo $\pi$, but braiding is unchanged. Now we must have that:
\begin{equation}
\theta_{a,\sc{T}a} = -\theta_{\sc{T}a,a}=-\theta_{a,\sc{T}a} \pmod{2\pi},
\end{equation}
so we conclude that $\theta_{a,\sc{T}a} = 0,\pi$. 

Suppose first that $\theta_{a,\sc{T}a}=0$, then by the same logic as above we see that $\sc{R} = \{1,\sc{T}a,\ldots, \sc{T}a^{n-1}\}\subseteq \sc{S}^\perp$. Now since $\sc{R}$ is a subgroup of $\sc{S}^\perp$ we must have that $n$ divides $|\sc{S}^\perp|$, by Lagrange's theorem. So then $|\sc{S}^\perp| = mn$ for some integer $m$. By Thm.~\ref{thm:character_theory} we see that $|\sc{A}_b| = mn^2$ and hence $|\sc{A}|=2mn^2$. But then by Lem.~\ref{lem:size_is_square} there must be some integer $l$ such that $l^2$ or $2l^2 = 2mn^2$, and hence $m=k^2,2k^{\prime2}$ for some integer $k$. Thus we see $|\sc{A}|=(nk)^2$ or $2(nk)^2$ where $k'=2k$.

Suppose instead that $\theta_{a,\sc{T}a}=\pi$. Then by the same logic as above we see that $\sc{R} = \{1,\sc{T}a^2,\ldots, \sc{T}a^{n-2}\}\subseteq \sc{S}^\perp$ and $n$ is even. Thus $n/2$ must divide $|\sc{S}^\perp|$ so we can write $|\sc{S}^\perp| = mn/2$ for some $m$. By Thm.~\ref{thm:character_theory} we see that $|\sc{A}_b| = mn^2/2$ and thus $|\sc{A}|=mn^2$. Then by Lem.~\ref{lem:size_is_square} we see that $|\sc{A}|=l^2$ or $2l^2$ for some $l$. We conclude that $m=k^2$ or $2k^2$ and thus that $|\sc{A}|=(nk)^2$ or $2(nk)^2$.

If there is no $a\in \sc{A}$ such that $\sc{T}a=ac$ then we know from Lem.~\ref{lem:size_is_square} that $|\sc{A}|=2l^2$ and we thus conclude that $|\sc{A}|=2(nk)^2$.
\end{proof}

\begin{remark}
Note that there would have been an easier proof of these statements had $a\in \sc{A}$ been a boson such that $\sc{T}a = a^k$. First since $\sc{T}a$ and $a$ must have the same order, we note that $k$ must be relatively prime to $n$. Then we see that $\sc{T}\langle a \rangle = \langle a\rangle$. Further suppose that $b\in \sc{T}\langle a\rangle^\perp$, then $b = \sc{T}d$ for some $d$ such that $\theta_{d,a}=0$. We see this means $0 = \theta_{d,a^k} = -\theta_{\sc{T}d,\sc{T}a^k} = -\theta_{b,a}$ and hence $b\in \langle a\rangle^\perp$. Since $\sc{T}$ is invertible we see that $\sc{T}\langle a\rangle^\perp = \langle a\rangle^\perp$. Thus $\sc{T}$ will restrict to an automorphism on the TO formed by condensing $a$, $\sc{A}' = \langle a\rangle^\perp/\langle a\rangle$. 

From Thm.~\ref{thm:condensation_relation} we know that $|\sc{A}| = n^2|\sc{A}'|$. Since $\sc{A}'$ is a time-reversal invariant Abelian TO then we can use the earlier results to finish the proof.

The case where $a\in \sc{A}$ is a boson which is mapped to one of its powers under time reversal may seem a little restrictive. However, we saw in the main text that this can often be accomplished in a time reversal invariant theory if $a$ is taken to be the vison.
\end{remark}

\subsection{$\U_f\rtimes \bb{Z}_2^T$ conserving fermionic Abelian TOs}
\label{app:sub_charge_time}

We review some powerful results about $\U_f\rtimes \bb{Z}_2^T$ conserving fermionic Abelian TOs that were given in Ref.~\cite{lapa_anomaly_2019} and conclude by extending them slightly. First we state a very useful theorem about fermionic Abelian TOs.

\begin{theorem}
Every fermionic Abelian TO $\sc{A}$ can be written as $\sc{A} = \sc{A}_b\boxtimes \{1,c\}$ where $\sc{A}_b$ is an Abelian bosonic TO. In particular $\sc{A}_b$ is closed under fusion and $c\notin \sc{A}_b$.
\label{thm:decomp}
\end{theorem}
\begin{proof}
The proof of this is given in Corollary A.19 of Ref.~\cite{drinfeld_braided_2010} and discussed in Ref.~\cite{lapa_anomaly_2019}.
\end{proof}

\begin{remark}
Note that this theorem is totally general, and does not rely on any underlying symmetries.
\end{remark}

\begin{example}
This can be demonstrated with a simple example. Consider $\U_q$, where $q$ is odd. We can express this as $\sc{A} = \{1,a,\ldots, a^{2q-1}\}$, where $\theta_a = \pi/q$. We see that $\theta_{a^q,a^n} = 2nq\theta_a = 2n\pi =0\pmod{2\pi}$ and $\theta_{a^q} = q^2\theta_a = q\pi = \pi \pmod{2\pi}$ since $q$ is odd. So $a^q=c$ is a transparent fermion. Thus we can write $\sc{A} = \{1,a^2,a^4,\ldots, a^{2q-2}\} \times \{1,c\}$. The even powers of $a$ are closed under fusion and braiding is nondegenerate on this set.
\end{example}

\begin{lemma}
Let $\sc{A} = \sc{A}_b\boxtimes \{1,c\}$ be a fermionic Abelian TO which respects $\U_f$. Then there exists a unique $\gamma \in \sc{A}_b$ such that:
\begin{equation}
\theta_{\gamma, b} = \pi Q_b \ \mathrm{for \ all} \ b\in\sc{A}_b.
\end{equation}
Moreover, $\gamma^2 = v$ or $\gamma^2 = vc$, where $v$ is the vison.
\label{lem:gamma_exist}
\end{lemma}
\begin{proof}
We repeat the arguments of Ref.~\cite{lapa_anomaly_2019}. Since $\U$ charge is defined modulo two in this fermionic theory, then we must have that $\pi Q_a + \pi Q_b = \pi Q_d\pmod{2\pi}$ if $ab=d$. If we restrict to the bosonic TO $\sc{A}_b$, then we can use arguments about the nondegeneracy of the $S$-matrix \cite{barkeshli_symmetry_2019} [see Eq.~(45)] to see that there must be a unique $\gamma \in \sc{A}_b$ such that $\theta_{\gamma,b} = \pi Q_b$ for every $b\in \sc{A}_b$.

Then since every element of $b\in \sc{A}_b\times c$ is such that $b=dc$ for some $d\in \sc{A}_b$ and $c$ braids trivially with $\gamma$ we must have that:
\begin{equation}
\theta_{\gamma,b} = \pi Q_b + \begin{cases} 0&\mbox{if } b\in \sc{A}_b\\ \pi &\mbox{if } b\in\sc{A}_b\times c\end{cases}.
\end{equation}
We then have that $\theta_{\gamma^2,b} = 2\pi Q_b$ for all $b\in \sc{A}$. Thus $\gamma^2$ has identical braiding to the vison, and since braiding is nondegenerate up to fusion with $c$ we must have that $\gamma^2 = v$ or $vc$.
\end{proof}

\begin{lemma}
Let $\sc{A} = \sc{A}_b\boxtimes \{1,c\}$ be a fermionic Abelian TO which respects $\U_f\rtimes \bb{Z}_2^T$. Then the $\gamma$ discussed in Lem.~\ref{lem:gamma_exist} also has the properties that $4\theta_\gamma = 0$. This will imply that $\gamma^2 = v$, so $v\in \sc{A}_b$. Moreover, $Q_\gamma$ will be an integer, so the charge of $v$ will be even.
\label{lem:gam_almost_boson}
\end{lemma}
\begin{proof}
We want to consider the braiding of $\gamma \sc{T}\gamma$ with some anyon $b\in \sc{A}_b$. We see that:
\begin{align}
\theta_{\gamma \sc{T}\gamma, b} &= \theta_{\gamma,b} + \theta_{\sc{T}\gamma, b}\\
&= \pi Q_b - \theta_{\gamma, \sc{T}b}\\
&= \begin{cases} 0 &\mbox{if } \sc{T}b\in \sc{A}_b\\ \pi &\mbox{if } \sc{T}b\in \sc{A}_b\times c\end{cases}
\end{align}
Since $\gamma \in \sc{A}_b$ this will mean that
\begin{align}
0 &= 2\theta_{\gamma \sc{T}\gamma, \gamma}\\
&= 4\theta_\gamma + 2\theta_{\sc{T}\gamma, \gamma}.
\end{align}
We know from time reversal that $\theta_{\sc{T}\gamma,\gamma} = -\theta_{\gamma, \sc{T}\gamma}$, so $2\theta_{\sc{T}\gamma,\gamma} = 0$. From the above we thus conclude that $4\theta_{\gamma} = 0$. Then we conclude $\gamma^2$ is a boson. 

Now from Lem.~\ref{lem:gamma_exist} we know that $\gamma^2 = v$ or $vc$. Since $\gamma^2$ is a boson, we must have $\gamma^2 = v$.

We finally have that $2\pi Q_\gamma = \theta_{\gamma,v}=\theta_{\gamma,\gamma^2}=4\theta_\gamma = 0$. So then $\gamma$ has integer charge. Since $v=\gamma^2$, $v$ must have even charge.
\end{proof}

\begin{example}
Consider $\bb{Z}_2$ gauge theory stacked with the fundamental fermion. Every anyon in this theory will square to the vacuum. Since the vacuum has even charge, this means every particle will have integer charge. In particular, the vison must braid trivially with every particle, so $v=1$. Given the natural charge assignment of the theory, we have $\gamma = m$. We see that $\gamma$ has integer charge and $\gamma^2 = v$ as it must.
\end{example}

\section{Results about non-Abelian TOs}
\label{app:non_abelian_results}

We begin in Appendix \ref{app:non_sub_G} by arguing a result about braiding with translation invariant anyons in a SET $\sc{C}$ which is assumed to be symmetric under the full symmetry group $G$ of Eq.~\eqref{eqn:symm_gp}. We see in the special case that translation does not permute anyons $a=\sc{T}a$, where $a$ is the background anyon. This result will also later prove useful when we consider the uniqueness of $\bb{Z}_{2q}$ gauge theory even if translation is allowed to permute anyons.

We then switch gears to prove some results bounding the numbers of anyons in a non-Abelian theory. In Appendix \ref{app:non_sub_trs} we first use the results of Appendix \ref{app:abelian_results} to prove a few things about the size of the subgroup of Abelian anyons $\sc{A}\subseteq \sc{C}$ when $\sc{C}$ is a time reversal invariant TO. Since $|\sc{C}|\geq |\sc{A}|$ with equality only when $\sc{C}$ is an Abelian TO this gives us a powerful tool to bound the size of $|\sc{C}|$. Finally in Appendix \ref{app:non_sub_cont} we use the results of Appendix \ref{app:non_sub_trs} to demonstrate that every minimal order which satisfies the full symmetry group $G$ of Eq.~\eqref{eqn:symm_gp} at filling $\nu = 1/q$ with $q$ even must contain exactly $8q^2$ anyons.

\subsection{SETs of $G$}
\label{app:non_sub_G}

In this section we offer an argument for the following result. Suppose that $\sc{C}$ is a SET of the full group $G$ in Eq.~\eqref{eqn:symm_gp}. Then the background anyon and its time-reversal partner must braid identically with all translation invariant anyons, i.e.
\begin{equation}
\theta_{a,d} = \theta_{\sc{T}a,d} \ \forall d\in \sc{C} \ \mathrm{such \ that} \ T_xd=d=T_yd.
\label{eqn:identical_braiding}
\end{equation}
Before giving our argument we note an obvious corollary of it. If translation does not permute anyons then this result means that $a = \sc{T}a$ or $a = c\times\sc{T}a$ since the two braid identically with every anyon. Only the former has the correct charge assignment modulo two, however, so $a = \sc{T}a$. This means that $\theta_a = \theta_{\sc{T}a} = -\theta_a$ so $\theta_a=0,\pi$ and $a$ is a boson or fermion.

The argument for Eq.~\eqref{eqn:identical_braiding} is very straightforward, and relies on the fact that our state must be invariant under time reversal. Suppose that we think of our state as possessing a background anyon $a$ at the center of every plaquette. Then under time reversal the state will be transformed to a state with $\sc{T}a$ at the center of every plaquette. Since we have not broken time-reversal symmetry we should not be able to detect the difference between these two states. As the way we detect differences between anyons is via braiding, na\"{i}vely we should expect that this means $a$ and $\sc{T}a$ must braid identically with all anyons. However, there is a subtlety to consider if translation is allowed to permute anyons.

Let us then be careful about defining the braiding of an anyon with the background anyon. First we create a pair of anyons $d$ and $\overline{d}$. Now we want to take $d$ around a plaquette, thus braiding it with $a$, but leave $\overline{d}$ in place. This will require first defining a string operator which is able to translate only the anyon $d$ by one lattice site in either direction. Let us specialize to the $x$ direction and suppose such an operator exists. The existence of such a local operator will imply that $d$ and $d$ translated by one lattice site, call this $T_xd$, must be in the same superselection sector. But, this is precisely the definition of anyon type \cite{kitaev_anyons_2006}, so the existence of such an operator implies that $d=T_xd$. The same will be true in the $y$ direction. Thus we see that our ability to take $d$ around a plaquette and leave $\overline{d}$ in place is exactly equivalent to the statement that $T_xd=d=T_yd$.

\begin{figure}
    \centering
    \subfigure[Original state]{
        \includegraphics[width=0.45\columnwidth]{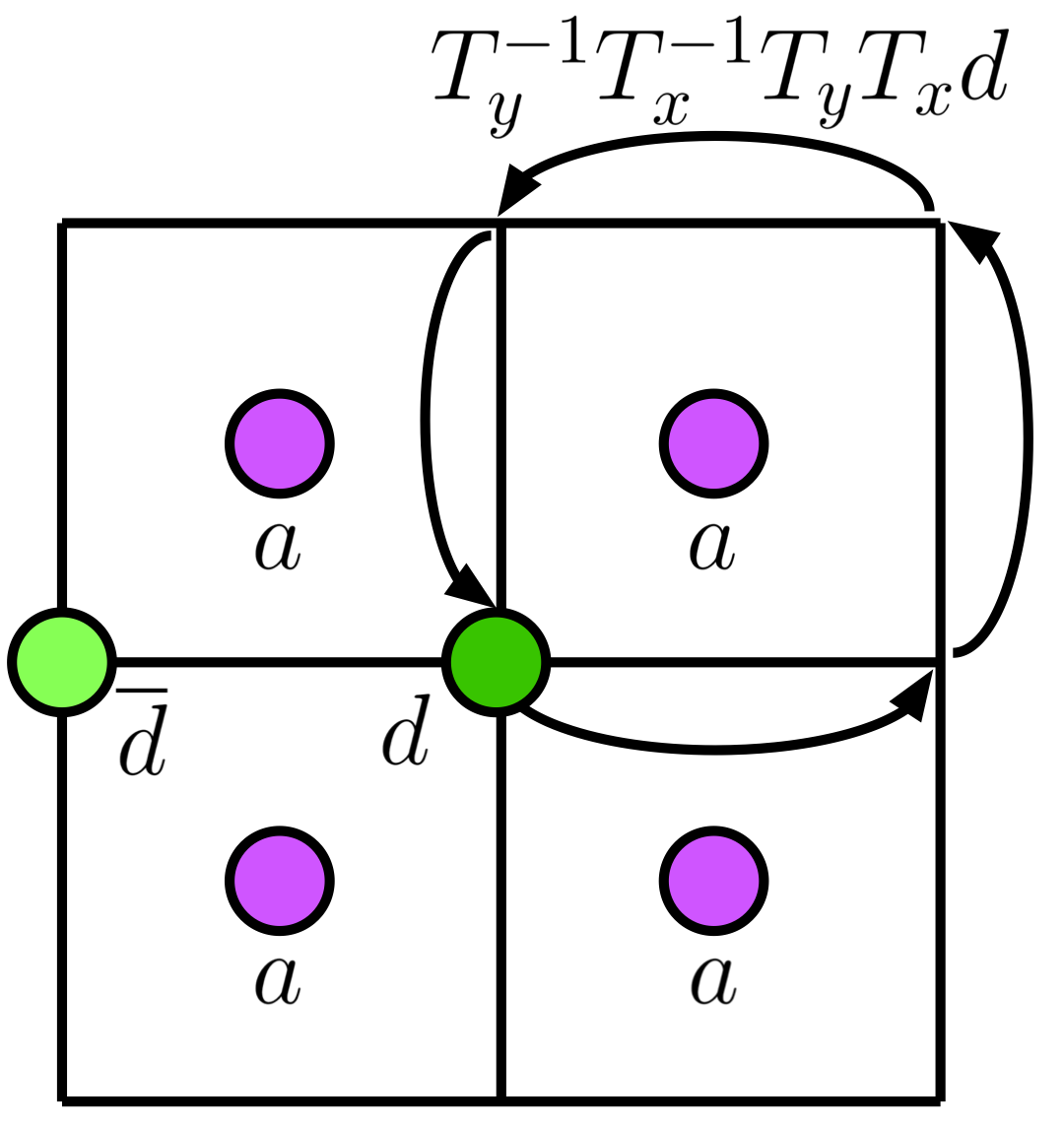}
        \label{fig:background_a}
            }
    ~ 
    \subfigure[Time reversed]{
        \includegraphics[width=0.45\columnwidth]{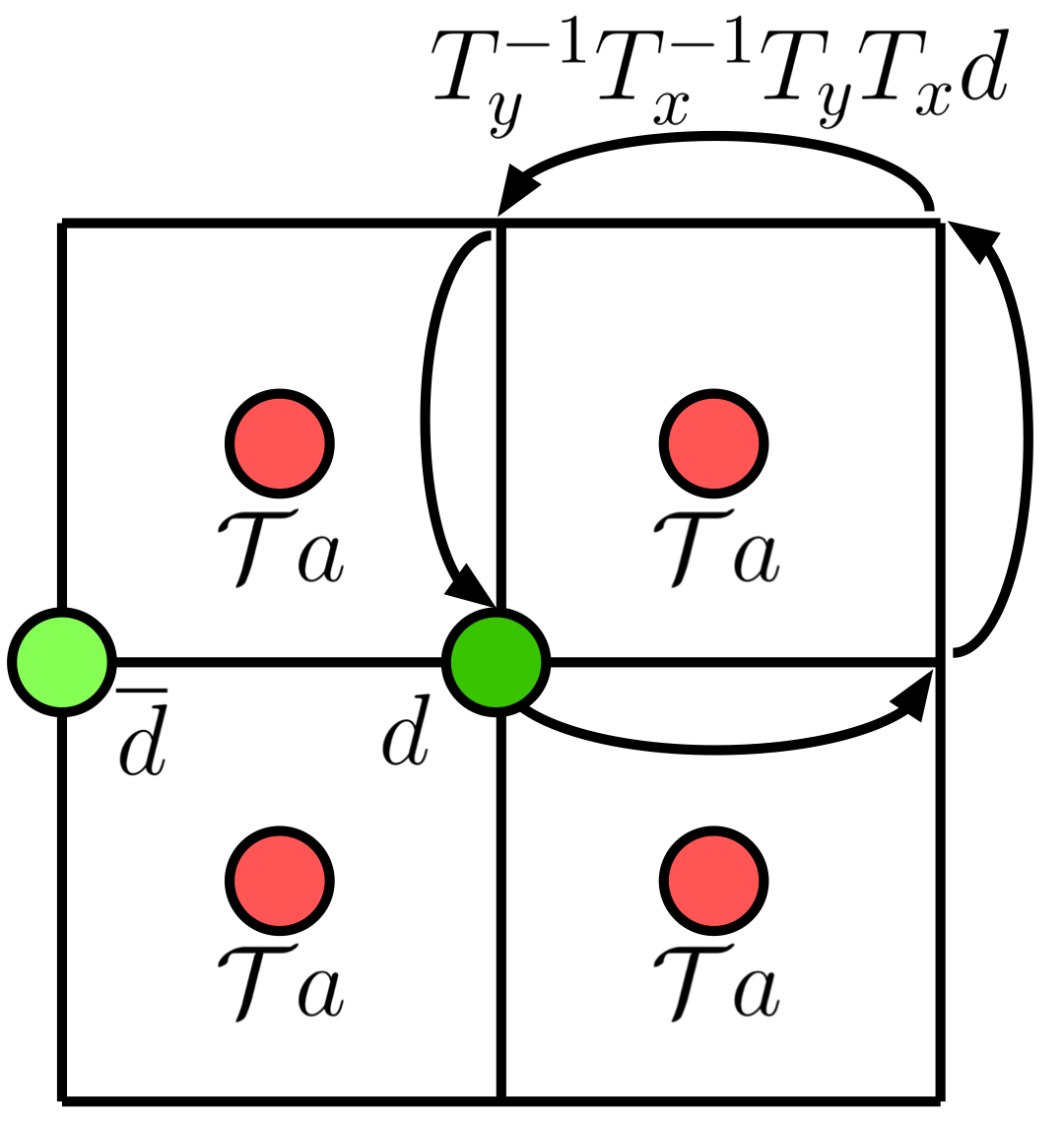}\
        \label{fig:background_Ta}
            }
    \captionsetup{justification=raggedright}
    \caption{An illustration of the requirement that $a$ and $\sc{T}a$ must braid identically with translation invariant anyons. We consider first creating $\overline{d}$, shown as a light green ball, and $d$, shown as a dark green ball. If $d$ is translation invariant then it is possible to construct a string operator $T_y^{-1}T_x^{-1}T_yT_xd$ which will take it around the plaquette as shown. In the original state, panel \subref{fig:background_a}, there is a background anyon $a$ in each plaquette so this process will cause the wave function to acquire a phase $\theta_{a,d}$. In the time reversed state, panel \subref{fig:background_Ta}, this will lead to a phase $\theta_{\sc{T}a,d}$. Since time reversal is unbroken these phases must be identical.}
	\label{fig:background}
\end{figure}

For an anyon $d$ which is left invariant by translation this process is possible, and we can find a string operator which takes $d$ around a plaquette. This will lead to the wavefunction acquiring a phase of $\theta_{a,d}$. Doing the same thing in the time reversed state will lead to a phase of $\theta_{\sc{T}a,d}$. Since we have not broken time-reversal symmetry, these two phases should be identical and we have the result. This is illustrated in Fig.~\ref{fig:background}.

\subsection{Bounds on time reversal invariant TOs}
\label{app:non_sub_trs}

In this section we will focus on proving results about a generic time reversal invariant TO; we assume no other symmetries.

\begin{remark}
We note that when restricted to $\sc{A}\subseteq \sc{C}$ the braiding phase $\theta\colon \sc{A}\times \sc{A}\rightarrow\bb{S}^1$ is still a symmetric group homomorphism. However, we no longer know that braiding is nondegenerate, even up to the fundamental fermion. Indeed it may be the case that there is some $a\in \sc{A}$ which braids trivially with all other anyons in $\sc{A}$, but nontrivially with some non-Abelian anyon in $\sc{C}\setminus \sc{A}$. Having noted that we make the following definition.
\end{remark}

\begin{definition}
Define $\sc{N}\subseteq \sc{A}$ as:
\begin{equation}
\sc{N} = \{a\in \sc{A} \mid \theta_{a,b} = 0 \ \forall b\in \sc{A}\}.
\end{equation}
\end{definition}

\begin{lemma}
The set $\sc{N}\subseteq \sc{A}$ forms a subgroup of $\sc{A}$. Furthermore, all elements of $\sc{N}$ are bosons or fermions. If we consider $\sc{A}/\sc{N}$ as a group, then braiding is nondegenerate on this group and $|\sc{A}|=|\sc{A}/\sc{N}||\sc{N}|$.
\end{lemma}
\begin{proof}
Suppose that $a,b\in \sc{N}$ then $\theta_{a,d} = 0 = \theta_{b,d}$ for all $d\in \sc{A}$. We then have $0 = \theta_{a,d} + \theta_{b,d} = \theta_{ab,d}$ so $ab\in \sc{N}$. Thus $\sc{N}$ is a subgroup. 

Now suppose $a\in \sc{N}\subseteq \sc{A}$, then we must have that $0 = \theta_{a,a}=2\theta_a$. Thus $a$ must be a boson or fermion.

Since $\sc{N}$ is precisely the subgroup of $\sc{A}$ that braids trivially with every element in $\sc{A}$ we see that braiding must be nondegenerate on $\sc{A}/\sc{N}$. 

The relation $|\sc{A}|=|\sc{A}/\sc{N}||\sc{N}|$ follows immediately from Lagrange's theorem.
\end{proof}

\begin{remark}
We note that if any $a\in \sc{N}$ is a fermion then topological spin in $\sc{A}/\sc{N}$ will only be defined mod $\pi$. This follows because for $b\in \sc{A}$ we see that $\theta_{ab} = \theta_a + \theta_b +\theta_{a,b} = \theta_b+\pi$. This is clear even if $\sc{C} = \sc{A}$ is a fermionic Abelian TO, where $\sc{N} = \{1,c\}$. Then $\sc{A}_b = \sc{A}/\{1,c\}$ is a bosonic TO where the topological spin is only defined mod $\pi$.
\end{remark}

\begin{lemma}
If $\sc{C}$ is a time reversal invariant TO then $\sc{T}$ restricts to an automorphism on $\sc{A}\subseteq \sc{C}$.
\end{lemma}
\begin{proof}
We assume that $\sc{T}\colon \sc{C}\rightarrow \sc{C}$ is such that $\sc{T}^2$ is the identity on $\sc{C}$ and that $\theta_{\sc{T}\psi} = -\theta_\psi$ for every $\psi \in \sc{C}$.

Now suppose that $a\in \sc{A}$, then $d_a = 1$. Since $\sc{T}$ is a symmetry of $\sc{C}$ we must have that $d_{\sc{T}a} = d_a=1$ \cite{barkeshli_symmetry_2019}. Thus we see that $\sc{T}(\sc{A})\subseteq \sc{A}$. Since we take $\sc{T}^2$ to be the identity we must then have that $\sc{T}(\sc{A}) = \sc{A}$. Thus $\sc{T}$ restricts to a bijection on $\sc{A}$. Furthermore, since $\sc{T}$ is a symmetry of $\sc{C}$ we must also have that the fusion data of $\sc{C}$ is left invariant by action of $\sc{T}$ \cite{barkeshli_symmetry_2019}. Thus we must also have that $\sc{T}(ab) = (\sc{T}a)(\sc{T}b)$ for $a,b\in \sc{A}$. Since $\sc{T}\colon \sc{A}\rightarrow \sc{A}$ is a bijective homomorphism it must restrict to an automorphism on $\sc{A}$.
\end{proof}

\begin{theorem}
Let $\sc{C}$ be a time reversal invariant TO. Suppose further that there is an Abelian anyon $a\in \sc{A}\subseteq \sc{C}$ such that  $a\in \sc{A}/\sc{N}$ has order $n$. Then $|\sc{A}/\sc{N}| = (kn)^2$ or $(kn)^2/2$ for some integer $k$. If there are no anyons $a \in \sc{A}$ such that $\sc{T}a = af$ for some fermion $f\in \sc{N}$ then we have that $|\sc{A}/\sc{N}|=(kn)^2$.
\label{thm:bound_C}
\end{theorem}
\begin{proof}
Suppose that $a\in \sc{N}$, then we must have that $0 = \theta_{a,\sc{T}b}$ for every $b\in \sc{A}$, since $\sc{T}b\in \sc{A}$. But then we see that $0 = -\theta_{\sc{T}a,b}$ for every $b\in \sc{A}$, so $\sc{T}a\in \sc{N}$. Thus we conclude that $\sc{T}$ also restricts to an automorphism on $\sc{N}$. Hence it must also restrict to an automorphism on $\sc{A}/\sc{N}$. 

We further know that $\sc{A}/\sc{N}$ is an Abelian group where braiding is nondegenerate. Then we can run the same arguments of Appendix \ref{app:abelian_results} to show that if there is some element of $\sc{A}/\sc{N}$ that has order $n$ then $|\sc{A}/\sc{N}|=(kn)^2/2$ or $(nk)^2$ for some $k$. Further if there are no anyons $a\in \sc{A}$ such that $\sc{T}a = af$ for some fermion $f\in \sc{N}$ then we know that every time-reversal invariant anyon in $\sc{A}/\sc{N}$ must have $0=2\theta_a= \theta_{a,a}$. Thus the arguments of Lem.~\ref{lem:size_is_square} reveal that $|\sc{A}/\sc{N}|$ must be a square and hence $|\sc{A}/\sc{N}|=(nk)^2$.
\end{proof}

\begin{remark}
Note that Lem.~\ref{lem:bound_f} actually follows from Thm.~\ref{thm:bound_C}. If $\sc{A}$ is a fermionic Abelian TO then $\sc{N} = \{1,c\}$. Further if $a\in \sc{A}$ has order $n$ and is such that $a^k \neq c$ for any $k$, we conclude that $a\in \sc{A}/\sc{N}$ has order $n$. Then we see that Thm.~\ref{thm:bound_C} will tell us that $|\sc{A}|= |\sc{N}|(kn)^2=2(kn)^2$ or $|\sc{A}|=|\sc{N}|(kn)^2/2=(kn)^2$ for some integer $k$.
\end{remark}

\subsection{Proof by contradiction when $q$ even}
\label{app:non_sub_cont}

In this appendix we show that $|\sc{A}|=8q^2$ for any putative minimal order when $q$ is even. In the main text we assumed that $\sc{T}a=a$ in order to show this; here we do not assume this.

Suppose that $\sc{C}\in mG^{(\nu)}_{\rm GSD}$ is the minimal TO for the group $G$ in Eq.~\eqref{eqn:symm_gp}. Then let $\sc{A}\subseteq \sc{C}$ be the group of Abelian anyons and $\sc{N}\subseteq \sc{A}$ be the subgroup of $\sc{A}$ containing all anyons that braid trivially with every anyon in $\sc{A}$.
Let $c$ be the fundamental fermion.

Since we have already constructed a TO with $8q^2$ anyons (including the fundamental fermion), we must have $|\sc{A}|\leq |\sc{C}|\leq 8q^2$.

\begin{lemma}
 If $c$ is the electron, then there does not exist $b\in \sc{A}$ such that $b^n=c$ for even $n$.   
 \label{lemma:nopower}
\end{lemma}

\begin{proof}
If $b^n=c$ then it must braid trivially with $b$, so $0=\theta_{b,c}=n\theta_{b,b}=2n\theta_b$. We can then multiply by the integer $n/2$ to see that $0 = n^2\theta_b = \theta_{b^n}$. But this is a contradiction since $\theta_c=\pi$.
\end{proof}

\begin{lemma}
One of the following two is true:
\begin{itemize}
    \item $|\sc{A}/\sc{N}|=4q^2, \sc{N}=\{1,c\}$ and $\sc{C}=\sc{A}$.
    \item $|\sc{A}/\sc{N}|=q^2$ and
\begin{equation}
\sc{N} = \{1,a^q,a^qc,c\}.
\end{equation}
In addition, $a^q$ is a boson with order exactly two in $\sc{A}$, and $v$ has order exactly $q$ in $\sc{A}$.
\end{itemize}
\end{lemma}
\begin{proof}

First, it is evident that the order of $v$ in $\sc{A}/\sc{N}$ is a multiple of $q$ because $\theta_{v,a}=2\pi/q$. Suppose the order of $v$ in $\sc{A}/\sc{N}$ is $k_vq$ with $k_v>1$. Then $|\sc{A}/\sc{N}|$ is at least $k_v^2q^2$, for the following reason: without loss of generality, we can assume there is an Abelian anyon $x$ with $\theta_{x,v}=2\pi/k_vq$. Consider the set $\sc{S}=\{x^k v^l| 0\leq k,l\leq k_vq-1\}$. For any $l$, $x^kv^l$ braids nontrivially with $v$ as long as $k\neq 0$. If $k=0$, then $v^l$ braids nontrivially with $x$ if $l\neq 0$. Thus the braiding on $\sc{S}$ is nondegenerate, and $\sc{S}$ must be fully contained in $\sc{A}/\sc{N}$.
Because $|\sc{N}|$ is at least 2 (it contains $\{1,c\}$), we find $|\sc{A}|\geq 2k_v^2q^2>8q^2$ if $k_v>2$. Thus from the minimality assumption, we conclude $k_v=1,2$. If $k_v=2$, then the minimality assumption requires $|\sc{A}|=8q^2$ and $\sc{A}=\sc{S}\boxtimes\{1,c\}$, which is entirely Abelian. This is the first case in the statement. Thus in the following we set $k_v=1$.

First we consider the case $\sc{N}=\{1,c\}$. We then recall by Thm.~\ref{thm:bound_C} that if there are no anyons $b\in \sc{A}$ such that $\sc{T}b=bf$ for some fermion $f\in \sc{N}$ then $|\sc{A}/\sc{N}|$ is $k^2q^2$ for some $k\in \bb{Z}$. But the only fermion in $\sc{N}$ is $c$ and charge conservation forbids $\sc{T}b=bc$ for any $b$. Thus we have $|\sc{A}/\sc{N}|=k^2q^2$. If $k>2$, then
\begin{equation}
    |\sc{A}|=|\sc{A}/\sc{N}|\cdot |\sc{N}|>4q^2\times 2=8q^2,
\end{equation}
contradicting the minimality assumption. We conclude that $k=1$ or $2$ are the only possibilities. If $k=1$, then $\sc{S}=\sc{A}/\sc{N}$, which means $a^q\in \sc{N}=\{1,c\}$. Since $Q_{a^q}=1$, $a^q\neq 1$. Lemma~\ref{lemma:nopower} shows that $a^q\neq c$. Then $\sc{N}\neq \{1,c\}$ contradicting our initial assumption.
Thus we conclude that if $\sc{N}=\{1,c\}$ we must have $|\sc{A}/\sc{N}|=4q^2$, and minimality require $\sc{C}=\sc{A}$.

Next we assume $\sc{N}\neq \{1,c\}$, which implies $|\sc{N}|\geq 4$. If $a^q\notin \sc{N}$, then $|\sc{A}/\sc{N}|\geq 2q^2$ (since $\sc{A}/\sc{N}$ already contains $\sc{S}$). Together $|\sc{A}|\geq 2q^2\times 4=8q^2$. However, in this case there has to be at least one (non-Abelian) anyon in $\sc{C}$ to braid nontrivially with $\sc{N}/\{1,c\}$, so $|\sc{C}|$ must be strictly greater than $8q^2$, violating the minimality assumption. We thus conclude that $a^q\in \sc{N}$.

Now we show that in this case $|\sc{A}/\sc{N}|=q^2$. Because $|\sc{N}|\geq 4$ and $|\sc{A}|\leq 8q^2$, it follows that $|\sc{A}/\sc{N}|\leq 2q^2$. The equality is impossible, since then $|\sc{A}|=8q^2$ so $\sc{C}=\sc{A}$ but we would have a transparent anyon $a^q\neq 1,c$. So according to Thm. ~\ref{thm:bound_C} the only possible value of $|\sc{A}/\sc{N}|$ is $q^2$ ($|\sc{A}/\sc{N}|=\frac12 q^2$ is impossible because $|\sc{A}/\sc{N}|\geq |\sc{S}|=q^2$).

In the following we focus on the case $|\sc{A}/\sc{N}|=q^2$, and prove the rest of the claim about the orders of $a$ and $v$ in $\sc{A}$. 

Because $a^q$ has charge $1\pmod{2}$,  the order of $a^q\in \sc{N}\subseteq \sc{A}$ must be a multiple of two i.e., $n_a=2qr$ for some integer $r$.
Furthermore, we find
\begin{equation}
    \theta_{a^q}=q^2 \theta_a=\frac{q}{2}\cdot 2q\theta_a=\frac{q}{2}\theta_{a^q,a}=0,
\end{equation}
since $a^q\in \sc{N}$ it must braid trivially with itself. 
So $a^q$ is a boson. This means that no power of $a^q$ will be equal to $c$. Thus $\{1,a^q,\ldots,a^{q(2r-1)}\}\times \{1,c\}\subseteq \sc{N}$ is a subgroup of $\sc{N}$ with size $4r$. 
Then $|\sc{A}|=|\sc{A}/\sc{N}||\sc{N}|=4rq^2$. If $r>2$, then this cannot be a minimal TO since $\sc{A}$ has more anyons than our construction. Suppose then that $|\sc{N}|=8$, then we see that $|\sc{A}|=8q^2$. For this to be a minimal TO we must then have that it is an Abelian TO and $\sc{A}=\sc{C}$. But then $a^q\in \sc{A}$ will braid trivially with every other anyon, but $a^q\neq 1,c$. This is a contradiction.  So we conclude that $|\sc{N}|=4$. This tells us that $r=1$ so $n_a=2q$.

Next, we know that $v$ has order exactly $q$ in $\sc{A}/\sc{N}$. This means that $v^q\in \sc{N}$. Since $v^q$ is a boson it must be that $v^q=1,a^q$. But $v$ had an integer charge so $v^q$ must have even charge, since $q$ is even. However, we know that $a^q$ has odd charge. So $v^q\neq a^q$ and $v^q=1$. Thus $v$ has order $q$ in $\sc{A}$.
\end{proof}

\begin{lemma}
If $|\sc{A}/\sc{N}|=q^2$, then for every $\sigma \in \sc{C}\setminus \sc{A}$ we must have $\theta_{a^q,\sigma}=\pi$. Furthermore, for every pair of non-Abelian anyons $\sigma, \sigma'\in \sc{C}\setminus \sc{A}$ we must have that 
\begin{equation}
\sigma \times \sigma' = b + ba^qc,
\end{equation}
for some $b\in \sc{A}$. Thus $d_\sigma^2 =2$ for every non-Abelian anyon. This will further imply that the number of non-Abelian anyons must be precisely $2q^2$ and $a^qc\times \sigma = \sigma$ for every non-Abelian anyon $\sigma$.
\end{lemma}
\begin{proof}
We know that $a^{2q}=1$, so $\sc{C}$ has a $\mathbb{Z}_2$ grading given by braiding with $a^q$. Namely, $\sc{C}=\sc{C}_0\oplus \sc{C}_{\pi}$, where $\sc{C}_{0}$ and $\sc{C}_\pi$ contains anyons with $0$ and $\pi$ braiding phase with $a^q$, respectively. We know that $\sc{A}\subset \sc{C}_0$ since $a^q\in \sc{N}$.  Now let us prove that all non-Abelian anyons must be in $\sc{C}_{\pi}$.

Suppose on the contrary that there is a non-Abelian anyon $\gamma\in \sc{C}_0$. Consider the size of the set $\sc{A}\times \gamma$. If $b\times \gamma\neq \gamma$ for every $b\in \sc{A}$, then this set has $4q^2$ anyons all different from $\sc{A}$, which means $|\sc{C}_0|\geq 8q^2$ and $|\sc{C}|>8q^2$ since $\sc{C}_{\pi}$ cannot be empty, contradicting the minimality assumption.
Thus there must exist some nontrivial $b\in \sc{A}$ such that $b\times \gamma = \gamma$. It means that $\gamma$ and $b\times \gamma$ must braid identically with all anyons in $\sc{A}$. Thus 
\begin{equation}
\theta_{x,\gamma} = \theta_{x,b\times \gamma} = \theta_{x,b}+\theta_{x,\gamma} \ \forall x\in \sc{A}.
\end{equation}
Thus $\theta_{b,d}=0$ for every $d\in \sc{A}$ so $b\in \sc{N}$. We cannot have that $b=c,a^q$ since $b\times\gamma$ has charge $Q_\gamma + 1\neq Q_\gamma \pmod{2}$. Thus we see that $b=a^qc$ is the only option. But $a^qc$ is a fermion, so $a^qc\times \gamma=\gamma$ means $\theta_{a^qc, \gamma}=\pi$, a contradiction.  

Thus we have shown that $\sc{C}_0=\sc{A}$, and $\sc{C}_{\pi}$ is entirely non-Abelian. For any $\sigma, \sigma'\in \sc{C}_{\pi}$, suppose $\beta$ is contained in their fusion product $\sigma\times\sigma'$. Clearly $\beta\in \sc{C}_0$, so $\beta\in \sc{A}$. The braiding of $\beta$ with anyons in $\sc{A}$ is fixed by $\sigma$ and $\sigma'$, thus $\beta$ is unique up to fusion with an anyon in $\sc{N}$. 
However, only two of these anyons have even charge, $1$ and $a^qc$. The requirement that the fusion products have the same charge (modulo two) means that:
\begin{equation}
\sigma \times \sigma' = N^{b}_{\sigma\sigma'}b+ N^{ba^qc}_{\sigma\sigma'}ba^qc,
\end{equation}
for some $b\in \sc{A}$, where we have included fusion multiplicities. But we must have that $N^{b}_{\sigma\sigma'}=N^{\sigma'}_{\overline{\sigma}b}$ (see Ref.~\cite{kitaev_anyons_2006}). We further know that the fusion multiplicities for $\overline{\sigma}\times b$ are always zero or one since $b$ is an Abelian anyon, and thus $N^{\sigma'}_{\overline{\sigma}b}=0,1$. This tells us that the above are either zero or one. Using Eq.~(21) of Ref.~\cite{barkeshli_symmetry_2019} we further see that this means:
\begin{equation}
N^b_{\sigma\sigma'} +N^{ba^qc}_{\sigma\sigma'} = d_{\sigma}d_{\sigma'} >1,
\end{equation}
since these are non-Abelian anyons. So each of these fusion multiplicities must be exactly one. This immediately implies the fusion rule ($N_{\sigma\sigma'}^b=N_{\sigma\sigma'}^{ba^qc}=1$) and that $d_\sigma^2=2$. 

Now we invoke a theorem in \cite{barkeshli_symmetry_2019} [see Eq.~(244)], that because of the $\mathbb{Z}_2$ grading on $\sc{C}$, 
we must have $\sc{D}_{\sc{C}_0}^2=\sc{D}_{\sc{C}_\pi}^2=4q^2$. Thus there are precisely $2q^2$ anyons in $\sc{C}_\pi$. For any $\sigma\in \sc{C}_\pi$, all anyons of the form $\sc{A}\times \sigma$ are still in $\sc{C}_\pi$, so to get $2q^2$ anyons instead of $4q^2$, there must be precisely one nontrivial anyon $x\in \sc{A}$ such that $x\times\sigma=\sigma$, and we have shown earlier that charge conservation implies $x$ must be $a^qc$. We thus conclude that $\sigma\times a^qc=\sigma$.
\end{proof}

\begin{lemma}
If $|\sc{A}/\sc{N}|=q^2$ then there is a single non-Abelian anyon $\gamma$ that satisfies:
\begin{align*}
Q_{\gamma} &= 0 \pmod{2}, \\
\theta_{a,\gamma} &= \pi Q_a,\\
\theta_\gamma &= 0, \\
\gamma \times \gamma &= v + va^qc.
\end{align*}
\label{lemma:gamma_properties}
\end{lemma}

\begin{proof}
From the previous lemma we know that all non-Abelian anyons in $\sc{C}$ lie in $\sc{A}\times \gamma$, where $\gamma$ is some non-Abelian anyon such that $\theta_{a^q,\gamma} =\pi$. Let us now use that fact to prove this theorem.

First we know that $v^q=1$. This means we must have that:
\begin{equation}
0 = \theta_{v^q,\gamma} = 2\pi qQ_\gamma \pmod{2\pi}.
\end{equation}
Thus we conclude that $Q_\gamma = l/q \pmod{1}$ for some $0\leq l \leq q-1$. We also know that $Q_a = 1/q$, so the charge of $a^{-l}\times \gamma$ is zero modulo one. Further since $a^q\in \sc{N}$ we have
\begin{equation}
\theta_{a^q,a^{-l}\times \gamma} = \theta_{a^q,a^{-l}} + \theta_{a^q,\gamma} = \pi.
\end{equation}
So without loss of generality we take $\gamma$ to have charge zero modulo one, where $\theta_{a,\gamma} = \pi p/q$ for some odd $p = 2p'+1$. We then consider $v^k\times \gamma$, we see that:
\begin{equation}
\theta_{a,v^k\times \gamma} = \frac{2\pi k}{q} + \frac{2\pi p' + \pi}{q}.
\end{equation}
Taking $k=-p'$ will ensure that the above is equal to $\pi/q = \pi Q_a$. Thus we see that there is a single non-Abelian anyon in $\sc{A}/\sc{N}\times \gamma$ such that $Q_\gamma = 0 \pmod{1}$ and $\theta_{a,\gamma} = \pi Q_a$. To conclude we note that the two\footnote{There are only two since $a^qc\times \gamma = \gamma$.} anyons in $\sc{N}\times \gamma$ both braid identically with $a$, but have charge that differs by $1\pmod{2}$. Thus without loss of generality we can choose $Q_\gamma = 0\pmod{2}$. We thus see that there is a single non-Abelian anyon in $\sc{A}\times \gamma$ such that:
\begin{equation}
Q_{\gamma} = 0 \pmod{2} \ \ \mathrm{and} \ \  \theta_{a,\gamma} = \pi Q_a.
\end{equation}

Next, we prove 
$\theta_\gamma=0$. For this we need the following relation:
\begin{equation}
e^{2\pi i(c_{-}-\sigma_H)/8} = \frac{1}{\sqrt{2}\sc{D}_{\sc{C}}} \sum_{\sigma \in \sc{C}} d_\sigma^2 e^{i(\theta_\sigma + \pi Q_\sigma)},
\end{equation}
which was shown in Ref.~\cite{lapa_anomaly_2019} for any fermionic TO using arguments about gauging the fermion parity. Since we have assumed time-reversal invariance the left-hand side must equal one.

We then need to compute the sum on the right-hand side. Recall that in the previous proofs we have established that if $|\sc{A}/\sc{N}|=q^2$ then $\sc{A}=\{1,c\}\boxtimes \langle a,v\rangle$, where $a$ has order $2q$, $a^q$ is a boson in $\sc{N}$, and $v$ has order $q$. We further showed that $\sc{C}=\sc{A}\cup \sc{A}\times \gamma$ where we can take $\gamma$ to have zero charge modulo two and $\theta_{a,\gamma}=\pi Q_a$. We also had that $a^qc\times \gamma = \gamma$ and $d_{\gamma}^2=2$.

We now do some algebra, building off of our previous results. First we have:
\begin{equation}
\sc{D}_{\sc{C}}^2 = |\sc{A}|+d_{\gamma}^2|\sc{A}\times \gamma|=8q^2.
\end{equation}
For compactness let us write $\tilde{\theta}_\sigma = \theta_\sigma + \pi Q_\sigma$ for an anyon $\sigma \in \sc {C}$. Then we can easily find:
\begin{align*}
\tilde{\theta}_{a^kv^lc^n} &= 2\pi klQ_a + \pi kQ_a+  k^2\theta_a\\
\tilde{\theta}_{a^kv^lc^n\times \gamma} &= 2\pi klQ_a +2\pi kQ_a +k^2\theta_a + \theta_\gamma.
\end{align*}
 Then we see that:
\begin{align*}
1 &= \frac{1}{4q}\left(\sum_{k,l,n} e^{i\tilde{\theta}_{a^kv^lc^n}} + 2\sum_{k,l}e^{i\tilde{\theta}_{a^kv^l\times \gamma}}\right)\\
&=\frac{1}{2q}\sum_{k=0}^{2q-1}e^{ik^2\theta_a}\Big(e^{i\pi kQ_a} + e^{i\theta_\gamma}e^{2\pi ikQ_a}\Big) \sum_{l=0}^{q-1}e^{2\pi i klQ_a}.
\end{align*}
Note that in the second sum we do not need to sum over $n$ because $a^qc\times\gamma=\gamma$.

The sums over $l$ will be zero unless $k$ is a multiple of $q$, so $k$ can only take two values $k=0,q$ and the sum over $l$ gives $q$. Note that $q^2\theta_a=\theta_{a^q}=0$. 
With $Q_a=1/q$,
\begin{align*}
1 
&=\sum_{k=0,q}\Big(e^{i\pi k/q} +e^{i\theta_\gamma} \Big)=e^{i\theta_\gamma}.
\end{align*}
Therefore $\theta_\gamma=0$.

Now consider the fusion product of $\gamma\times \gamma$, we know that it is equal to $b+ba^qc$ for some $b\in \sc{A} = \langle a,v\rangle \times \{1,c\}$. Since $Q_\gamma = 0\pmod{2}$ this means $\theta_{\gamma, v} = 0$. This, along with the fact that $\theta_{\gamma, a} = \pi Q_a$, means that any fusion product $b\in \gamma \times \gamma$ must be such that $\theta_{b,v} = 0$ and $\theta_{b,a} = 2\pi Q_a$. Thus $b$ braids identically to the vison, which implies $b=v$ or $b=vc$ by braiding nondegeneracy. In addition, $Q_\gamma=0$ mod 2 implies $Q_v=0$ mod 2, which shows $b=v$.
\end{proof}

\begin{remark}
After all of this work we have found that all non-Abelian anyons are of the form $\sc{A}\times \gamma$, where $\gamma$ can be thought of as the non-Abelian square root of the vison. Indeed this is what we might have guessed! We needed to find some particle that had $\pi$-braiding with $a^q$, which one might guess looks like the square root of the vison. Given this guess it is good we found something that agrees with our intuition.
\end{remark}

Lastly, we rule out the $|\sc{A}/\sc{N}|=q^2$ case.

\begin{theorem}
 $|\sc{A}/\sc{N}|=4q^2, |\sc{N}|=2$ and hence $\sc{C}=\sc{A}$ is an Abelian TO.
\end{theorem}
\begin{proof}
Suppose that $|\sc{A}/\sc{N}|=q^2$. Then we established that $\sc{C} = \sc{A}\cup \sc{A}\times \gamma$, where $\gamma$ is an anyon with the following properties given in Lemma~\ref{lemma:gamma_properties}. It will turn out that this theory is not consistent with time reversal.

We can see this by condensing the bosonic vison $v$, forming a child TO $\sc{C}'$. This will confine any anyon which braids nontrivially with $v$, i.e., has charge not equal to an integer. It will further identify all powers of $v$ with the vacuum. It is then clear that:
\begin{equation}
\sc{C}' = \sc{C}'' \times \{1,c\} \text{ where } \sc{C}'' = \{1,a^qc, \gamma\},
\end{equation}
and where
\begin{equation}
\gamma \times \gamma = 1 + a^qc, \ a^qc\times \gamma = \gamma,
\end{equation}
and $(a^qc)^2 = 1$ so $\sc{C}''$ is closed under fusion. We note that $\sc{C}''$ has the fusion rules of the Ising TO, but the Ising anyon $\gamma$ has $\theta_\gamma=0$. However, all TOs with these fusion rules were classified in Ref.~\cite{kitaev_anyons_2006}, and in all of them the non-Abelian Ising anyon satisfies $8\theta_\gamma=\pi\pmod{2\pi}$. 

We have arrived at a contradiction and thus conclude that $|\sc{A}/\sc{N}|\neq q^2$.
\end{proof}

\section{Uniqueness results}
\label{app:uniqueness}

In this section we prove a number of uniqueness results. Here we allow translation to permute anyons, so it may no longer be the case that $a=\sc{T}a$.

In the proofs we need the following fact:

 \begin{theorem}
     Let $\sc{A}$ be a fermionic Abelian TO of size $|\sc{A}|=2N^2$ and $\sc{A}$ has an order $N$ boson. Then $\sc{A}$ is a twisted $\bb{Z}_N$ gauge theory, and can be represented by the following $K$ matrix:
     \begin{equation}
K = \begin{pmatrix} 0& N& 0\\ N& -n& 0\\ 0& 0& 1\end{pmatrix}.
\label{eqn:SPTKmat}
\end{equation}
\label{theoremDW}
 \end{theorem}

\begin{proof}
    
 To prove this theorem, consider the braiding with the boson $b$. Since $b^N=1$, denote by $\sc{A}_k$ the set of anyons with $\theta_{\cdot, b}=\frac{2\pi}{N}k$. Clearly, each of them must have $2N$ anyons. Since $\theta_{b,b}=0$, we must have $\sc{A}_0=\{1,b,\cdots, b^{N-1}\}\boxtimes \{1,c\}$.
 
 Denote by $\phi$ to be an arbitrary anyon in $\sc{A}_1$. Then we may write
 \begin{equation}
     \sc{A}_k=\phi^k\times \sc{A}_0.
 \end{equation}

 The only remaining uncertainty in the fusion rule is $\varphi^N$, which must belong to $\sc{A}_0$. Suppose $\varphi^N=b^{s}c^t$.  Then
 \begin{equation}
    2N\theta_{\phi}= \theta_{\phi, \phi^N}=\theta_{\phi, b^s}=\frac{2\pi s}{N}.
 \end{equation}
 We find $\theta_\phi=\frac{\pi s}{N^2}+\frac{\pi l}{N}$.
 In addition, because
     $\theta_{\phi^N}=N^2\theta_\phi=\pi t$, it follows that $s+Nl+t$ is even. We can thus write $\theta_\phi=\frac{\pi n}{N^2}$ where $n=s+Nl$. Clearly, $n$ takes any integer value. Together with $\theta_b=0, \theta_{b,\phi}=\frac{2\pi}{N}$, the statistics of all anyons are all fixed. It is easy to verify that the anyon theory agrees with that of the $K$ matrix with $n=2r+t$, which can take any integer value. The precise identification is $b=(0,1), \phi=(1,0)$.
\end{proof}

\begin{remark}
    This theorem can also be deduced by gauging the 1-form symmetry generated by the order $N$ boson, resulting in a fermionic SPT state protected by $\bb{Z}_N\times\bb{Z}_2^f$ symmetry~\cite{wang_interacting_2017, cheng_classification_2018}. The parent TO is then obtained from gauging the $\bb{Z}_N$ symmetry in the fermionic SPT phase. 

    Let us elaborate on this construction. Since we require that the parent TO is Abelian, the fermionic SPT can be described by the following Abelian CS theory
    \cite{lu_theory_2012}:
    \begin{equation}
        \sc{L}_{\rm fSPT}=\frac{1}{2\pi}a_1da_2 - \frac{n}{4\pi}a_2da_2 + \frac{1}{2\pi}a_1dA,
    \end{equation}
    Here we also introduce the background gauge field $A$. Gauging the $\bb{Z}_N$ symmetry, we add to the Lagrangian the term $\frac{N}{2\pi}AdB$. We use equation of motion to integrate out $a_1$ and set $a_2=-A$, after which the theory becomes
    \begin{equation}
        \sc{L}_{\rm gauged}=- \frac{n}{4\pi}AdA + \frac{N}{2\pi}AdB,
    \end{equation}
    which is precisely the upper $2\times 2$ block of the $K$ matrix given in Eq.~\eqref{eqn:SPTKmat}. The lower $1$ on the diagonal is due to the stacking of this bosonic theory with $\{1,c\}$.
\end{remark}

By a GL$(3,\bb{Z})$ transformation we can show $n\sim n+2N$. In addition, for odd $N$, one can further show that $n\sim n+N$ by the following GL($3,\bb{Z})$ transformation:
\begin{equation}
\begin{split}
    W^TKW&=\begin{pmatrix} 0& N& 0\\ N& -n-N& 0\\ 0& 0& 1\end{pmatrix},\\
    W&=\left(
\begin{array}{ccc}
 1 & -\frac{1+N}{2} & 1 \\
 0 & 1 & 0 \\
 0 & N & 1 \\
\end{array}
\right).
\end{split}
\end{equation}
It follows that when $N$ is odd, one can always choose $n$ to be even.

\subsection{The case of $q$ odd}

\begin{theorem}
Let $\sc{A}$ be a fermionic Abelian TO of size $|\sc{A}|=2q^2$ that has an order $q$ boson, where $q$ is odd. If $\sc{A}$ is further time-reversal invariant then it must be $\bb{Z}_q$ gauge theory fused with the fundamental fermion.
\end{theorem}
\begin{proof}
According to Theorem \ref{theoremDW}, the Abelian TO can be described by a $K$ matrix
\begin{equation}
K = \begin{pmatrix} 0& q& 0\\ q& -2n& 0\\ 0& 0& 1\end{pmatrix}, \ n\in \bb{Z}_q.
\end{equation}
Here we have set the diagonal element to be an even, which is always possible for $q$ odd.
 We show that the only $K$ matrix of this form consistent with time reversal is the one where $n=0\pmod{q}$.

The anyons $d\in \sc{A}$ in this theory can be written as $d=(l,m,k)$ for some integers $l,m,k$ where $k=0,1$. They will have topological spin given by:
\begin{equation}
\theta_{(l,m,k)} = \pi (l,m,k)K^{-1}\begin{pmatrix} l\\ m \\ k\end{pmatrix} = 2\pi \left(\frac{ml}{q} + \frac{nl^2}{q^2}\right) + k\pi \nonumber.
\end{equation}
Now we know that time reversal is an automorphism on this TO. So the anyon $(1,0,0)$ must be mapped to some $(l,m,k)$ under time reversal. This tells us that:
\begin{equation}
2\pi\frac{n}{q^2} = -2\pi\left(\frac{nl^2}{q^2} + \frac{ml}{q}\right) + \pi k \pmod{2\pi}.
\end{equation}
Multiplying by $q$ we see that
\begin{equation}
2\pi \frac{n(l^2+1)}{q} = \pi k \pmod{2\pi},
\end{equation}
since $q$ is odd. Multiplying by $q$ again we see that $\pi k = 0\pmod{2\pi}$ and thus $k=0$. We may therefore conclude that:
\begin{equation}
n(l^2+1) = 0\pmod{q}.
\end{equation}

Further we note that the anyon $(0,1,0)$ must be mapped to some $(l',m',k')$ under time reversal. Then we must have:
\begin{equation}
2\pi \left(\frac{nl^{\prime 2}}{q^2}+\frac{m'l'}{q}\right)+\pi k' = 0\pmod{2\pi}.
\end{equation}
We can again multiply by $q^2$ and since $q$ is odd this tells us that $k'=0$. Then multiplying by $q$ again tells us that:
\begin{equation}
nl^{\prime 2} = 0 \pmod{q}.
\end{equation}

Now we use the fact that $\sc{T}$ is an order two operation. This means that:
\begin{align*}
(1,0,0)&= \sc{T}^2(1,0,0)\\
&= \sc{T}(l,m,0)\\
&=l\sc{T}(1,0,0) + m\sc{T}(0,1,0)\\
&= (l^2+ml',lm+mm',0). 
\end{align*}
So then we see that $(l^2+ml'-1,lm+mm',0)$ must be identified with $(0,0,0)$. This means that it must braid trivially with all other anyons. In particular the requirement that it braid trivially with $(0,1,0)$ means that:
\begin{equation}
l^2+ml'-1 = 0 \pmod{q}.
\end{equation}
We then multiply by $nl'$ and use the fact that $nl^{\prime 2}=0\pmod{q}$ to simplify to $nl'(l^2-1)=0\pmod{q}$. We have already seen that $n(l^2+1)=0 \pmod{q}$, so we must have $nl'(l^2+1) = 0\pmod{q}$. Subtracting these reveals that $2nl' = 0\pmod{q}$. Since $q$ is odd this must mean that $2$ does not divide $q$ and hence $nl'=0 \pmod{q}$. Then we can multiply $l^2+ml'-1=0\pmod{q}$ by $n$ to see that:
\begin{equation}
n(l^2-1)=0\pmod{q}.
\end{equation}
But if we now subtract this from $n(l^2+1)=0\pmod{q}$ we see that we must have $2n=0\pmod{q}$. Since $q$ is odd this must mean that $n=0\pmod{q}$.

We therefore conclude that for $q$ odd the \textit{only} fermionic Abelian TO, $\sc{A}$, which is time-reversal invariant, has an order $q$ boson, and is of size $|\sc{A}|=2q^2$ is given by the $K$-matrix
\begin{equation}
K = \begin{pmatrix} 0& q& 0\\ q& 0& 0\\ 0& 0& 1\end{pmatrix}.
\end{equation}
This is precisely the $K$-matrix of $\bb{Z}_q$ gauge theory stacked with $\{1,c\}$.
\end{proof}

\begin{corollary}
Let $\sc{A}\in mG^{(\nu)}_{\rm GSD}$ for the group $G$ in Eq.~\eqref{eqn:symm_gp}, where $q$ is odd. Then $\sc{A}$ must be $\bb{Z}_q$ gauge theory stacked with $\{1,c\}$.
\end{corollary}
\begin{proof}
We showed in the main text that there must exist a bosonic vison $v\in \sc{A}$ which has order $q$. We further showed that $|\sc{A}|=2q^2$ and was Abelian if $\sc{A}\in mG^{(\nu)}_{\rm GSD}$. Applying the results of the previous theorem gives us our proof.
\end{proof}

\subsection{The case of $q$ even}

We first want to prove that there exists an order $2q$ boson in this theory, $\gamma$, even if translation permutes anyons. In fact we will be able to do better and find a unique such boson. Since our fermionic TO is Abelian we can decompose $\sc{A}=\sc{A}_b\boxtimes \{1,c\}$, where $\sc{A}_b$ is a bosonic TO which is closed under fusion. In Appendix \ref{app:sub_charge_time} we reviewed the proof that there is a unique $\gamma \in \sc{A}_b$ which has the property that:\footnote{All results thus far about $\gamma$ in Abelian orders are proven in Ref.~\cite{lapa_anomaly_2019}.}
\begin{equation}
\theta_{\gamma, b} = \pi Q_b \ \mathrm{for \ all} \ b\in \sc{A}_b.
\end{equation}
In particular, $\theta_{\gamma, a^q}=\pi$.
Moreover, there we showed that $\gamma^2 = v$. Since $v$ has order exactly $q$, this means that $\gamma$ must have order exactly $2q$. Additionally, since $v$ is a boson we must have that $4\theta_\gamma = 0$. We can now fully enumerate the anyon content of our theory.

\begin{lemma}
Let $\sc{A}\in mG^{(\nu)}_{GSD}$ for the group $G$ in Eq.~\eqref{eqn:symm_gp}. Then $\sc{A}$ is made up entirely of: powers of $a$, the unique $\gamma$ discussed above, and $c$, i.e.
\begin{equation}
\sc{A} = (\langle a\rangle \times \langle \gamma\rangle) \boxtimes \{1,c\}.
\end{equation}
\label{lem:full_group}
\end{lemma}
\begin{proof}
First recall from Thm.~\ref{thm:decomp} in Appendix \ref{app:abelian_results} that $\sc{A} = \sc{A}_b\boxtimes \{1,c\}$, where $\sc{A}_b$ is closed under fusion. We already know that $\gamma \in \sc{A}_b$. Since $q$ is even it must always be the case that $a^q\in \sc{A}_b$. Suppose then that $a\notin \sc{A}_b$, then $ac\in \sc{A}_b$. Since $\sc{A}_b$ is closed under fusion $a^{q+1}c\in \sc{A}_b$. We see that $a^{q+1}c$ is an order $2q$ anyon with charge exactly $1/q$ modulo two. Then without loss of generality we can assume $a\in \sc{A}_b$, since all we know about $a$ are these two facts.

Now since $\sc{A}=\sc{A}_b\boxtimes \{1,c\}$ and $c\notin \sc{A}_b$ it must be the case that $|\sc{A}_b|=|\sc{A}|/2=4q^2$. So if we can show there are $4q^2$ anyons contained in $\langle \gamma\rangle \times \langle a\rangle \subseteq \sc{A}_b$, then we are done. Suppose there are fewer than $4q^2$ anyons in $\langle \gamma\rangle \times \langle a\rangle$, then there are some nontrivial $r$ and $s$ such that $a^r=\gamma^s$. We know that $\gamma$ has integer charge, so $r=q$ is the only possibility. Squaring the relation reveals that $1=\gamma^{2s}=v^s$, so $s=q$. But this cannot be the case because $\gamma^q$ will then have even charge, and we know that $a^q$ has odd charge. We thus have the contradiction and the result follows.
\end{proof}

\begin{lemma}
Let $\sc{A}\in mG^{(\nu)}_{\rm GSD}$ for the group $G$ in Eq.~\eqref{eqn:symm_gp}. Then if the unique $\gamma$ discussed earlier is not a boson $a^q$ must be a fermion. This further implies that: $2q\theta_a = \pi$, $q/2$ must be odd, and $\theta_{a^{q+1}}=\theta_a$.
\label{lem:aq_fermion}
\end{lemma}
\begin{proof}
To show this we again use the relation \begin{equation}
e^{2\pi i(c_{-}-\sigma_H)/8} = \frac{1}{4q} \sum_{b \in \sc{A}} e^{i(\theta_b + \pi Q_b)},
\end{equation}
proven in Ref.~\cite{lapa_anomaly_2019} for any fermionic TO. Time-reversal invariance will demand that the left-hand side is equal to one. Then it is a simple matter of evaluating the right-hand side for all anyons; these are of the form $a^r\gamma^s c^n$ by Lem.~\ref{lem:full_group}.

For compactness we again write $\tilde{\theta}_b = \theta_b + \pi Q_b$. Then since $\pi Q_\gamma = \theta_{\gamma,\gamma} = 2\theta_\gamma$ we have that:
\begin{align}
\tilde{\theta}_{a^r\gamma^s c^n} &= r^2\theta_a + \frac{\pi r(s+1)}{q}+(s^2+2s)\theta_\gamma \\
=& r^2\theta_a + \frac{\pi r(2s'+1)}{q} + \begin{cases} 0 &\mbox{if } s=2s'\\ \pi r/q - \theta_\gamma &\mbox{if } s=2s'+1\end{cases}. \nonumber
\end{align}
If we then sum over $s'=0,\ldots,q-1$ we see that $r=0,q$ are the only options that do not evaluate to zero. The relation will then become:
\begin{equation}
1 = \frac{1}{2}\left(1 - e^{i\theta_{a^q}} + e^{-i\theta_\gamma} + e^{-i\theta_\gamma + i\theta_{a^q}}\right).
\end{equation}

Now we know that $a^{2q}=1$ and thus $0 = \theta_{a,a^{2q}}=4q\theta_a$. If we multiply by $q/2$ we then have $0=2q^2\theta_a$, so $a^q$ must be a boson or fermion. Thus if $a^q$ is not a fermion, then it must be a boson. From the above we see this means $1 = e^{-i\theta_\gamma}$, so $\gamma$ must be a boson. By negation if $\gamma$ is not a boson, then $a^q$ must be a fermion.

Now suppose that $a^q$ is a fermion. We know that $4q\theta_a = 0$, so then $2q\theta_a = 0,\pi$. Suppose that it equaled zero, then multiplying by $q/2$ we see that $q^2\theta_a = 0$, a contradiction. So $2q\theta_a = \pi$. We can again multiply by $q/2$ to see that $q^2\theta_a = \pi q/2$. Since $q^2\theta_a = \pi$ this must mean that $q/2$ is odd. Finally we have that $\theta_{a^{q+1}} = (q^2+2q+1)\theta_a = \theta_a$, since $2q\theta_a$ and $q^2\theta_a$ both equal $\pi$.
\end{proof}

\begin{theorem}
Let $\sc{A}\in mG^{(\nu)}_{\rm GSD}$ for the group $G$ in Eq.~\eqref{eqn:symm_gp}. Then the unique $\gamma$ discussed earlier must be a boson or fermion and there is thus an order $2q$ boson in the theory.
\label{thm:gam_is_boson_fermion}
\end{theorem}
\begin{proof}
We know from Lem.~\ref{lem:full_group} that the unique $\gamma\in \sc{A}$ is an order $2q$ anyon with $4\theta_\gamma =0$. Suppose by contradiction that $2\theta_\gamma \neq 0$, then $2\theta_\gamma = \pi$. In particular since $\pi Q_\gamma = \theta_{\gamma,\gamma} = 2\theta_\gamma = \pi$ this means that $Q_\gamma = 1$. We further know from Lem.~\ref{lem:aq_fermion} that $2\theta_\gamma = \pi$ will mean: $a^q$ is a fermion, $2q\theta_a = \pi$, $q/2$ is odd, and $a^{q+1}$ has the same topological spin as $a$.

We now want to establish the possibilities that $a$ can be mapped to under time reversal. By Lem.~\ref{lem:full_group} we know that all anyons in $\sc{A}$ can be written as
\begin{equation}
a^n\gamma^mc^k, \ \mathrm{where} \ 0\leq n,m\leq q-1, k=0,1.
\end{equation}
The only anyons with the same $\U$ charge modulo two as $a$ are $a\gamma^m$ or $a^{q+1}c\gamma^m$ for $m$ even, and $a\gamma^mc$ or $a^{q+1}\gamma^m$ for $m$ odd. Since time reversal preserves the $\U$ charge of anyons it must be the case that $\sc{T}a$ is equal to one of these possibilities.

It must also be the case that $-\theta_a = \theta_{\sc{T}a}$. We now show that none of the possibilities mentioned above can satisfy this property. 
\begin{enumerate}
\item Suppose $\sc{T}a = a\gamma^m$ for some even $m$. Then:
\begin{align*}
0&= q(\theta_a + \theta_{\sc{T}a})\\
&= q\left(2\theta_a + \frac{\pi m}{q}\right) \text{ since } m^2\theta_\gamma = 0 \text{ for } m \text{ even,}\\
&= \pi \text{ since } 2q\theta_a = \pi \text{ and } m \text{ even},
\end{align*}
clearly a contradiction. 

\item Suppose that $\sc{T}a = a^{q+1}c\gamma^m$ for $m$ even. Then:
\begin{align*}
0 &= q(\theta_a + \theta_{\sc{T}a})\\
&= q\left(\pi + 2\theta_a + \frac{\pi m(q+1)}{q}\right) \text{ since } \theta_{a^{q+1}}=\theta_a\\
&= \pi \text{ since } 2q\theta_a = \pi \text{ and } m \text{ even},
\end{align*}
again a contradiction. 

\item Suppose that $\sc{T}a = a\gamma^mc$ for some odd $m$. Then:
\begin{align*}
0 &= q(\theta_a + \theta_{\sc{T}a})\\
&= q\left(\pi + 2\theta_a +\frac{\pi m}{q} + \theta_\gamma\right) \text{ since } m^2\theta_\gamma = \theta_\gamma \text{ for } m \text{ odd,}\\
&= \pi + \pi + q\theta_\gamma \text{ since } 2q\theta_a = \pi \text{ and } m \text{ odd},\\
&= \pi,
\end{align*}
since $q\theta_\gamma = q/2(2\theta_\gamma)=q\pi/2$ and $q/2$ must be odd. So we have again arrived at a contradiction. 
\item Finally, suppose that $\sc{T}a = a^{q+1}\gamma^m$ for some odd $m$. Then:
\begin{align*}
0 &= q(\theta_a + \theta_{\sc{T}a})\\
&= q\left(2\theta_a + \frac{\pi m(q+1)}{q} + \theta_\gamma\right) \text{ since } \theta_{a^{q+1}}=\theta_a, m^2\theta_\gamma = \theta_\gamma \\
&= \pi + \pi + q\theta_\gamma\text{ since } 2q\theta_a = \pi \text{ and } m \text{ odd},\\
&= \pi,
\end{align*}
since $q\theta_\gamma = \pi$. We have arrived at another contradiction. 
\end{enumerate}

We thus see that if $2\theta_\gamma = \pi$ there is no possible choice of time reversal action that is also consistent with charge conservation. 
By contradiction we conclude that $2\theta_\gamma =0$. It follows that either $\gamma$ or $\gamma c$ is an order $2q$ boson.
\end{proof}

\begin{lemma}
Let $\sc{A}$ be a fermionic Abelian TO of size $|\sc{A}|=8q^2$ that has an order $2q$ boson. If $\sc{A}$ is time-reversal invariant then it must either be $\bb{Z}_{2q}$ gauge theory or $\U_{2q}\boxtimes \U_{-2q}$ (stacked with $\{1,c\}$).
\label{lem:two_Q_gauge}
\end{lemma}
\begin{proof}
According to Theorem \ref{theoremDW}, we know that $\sc{A}$ can be described by a $K$ matrix:
\begin{equation}
K = \begin{pmatrix} 0& 2q& 0\\ 2q& -n& 0\\ 0& 0& 1\end{pmatrix}, \ n\in \bb{Z}_{2q}.
\end{equation}

The anyons $d\in \sc{A}$ in this theory can again be written as $d=(l,m,k)$ for some integers $l,m,k$ where $k=0,1$. They will have topological spin given by
\begin{equation}
\theta_{(l,m,k)} = \pi \left(\frac{ml}{q}+\frac{nl^2}{4q^2}\right)+k\pi.
\end{equation}

We can now repeat a very similar argument to the case of $q$ odd. Suppose that $\sc{T}(1,0,0)=(l,m,k)$ then we see that
\begin{equation}
\pi \frac{n}{4q^2} = -\pi \left(\frac{nl^2}{4q^2}+\frac{ml}{q}\right) + \pi k\pmod{2\pi}.
\end{equation}
If we multiply by $2q$ then we see that
\begin{equation}
n(l^2 +1)=0 \pmod{4q}.
\end{equation}

Similarly if we have that $\sc{T}(0,1,0)=(l',m',k')$ then we have that:
\begin{equation}
\pi \left(\frac{nl^{\prime 2}}{4q^2} + \frac{m'l'}{q}\right) + \pi k' = 0\pmod{2\pi}.
\end{equation}
If we now multiply by $2q$ then we see that:
\begin{equation}
nl^{\prime 2} = 0 \pmod{4q}.
\end{equation}

Because $\sc{T}^2$ is the identity, $\sc{T}^2(1,0,0)= (l^2+ml', lm+mm', lk+k)$ should be equivalent to $(1,0,0)$. Thus $(l^2+ml'-1, lm+mm', lk+k)$ must braid trivially with all anyons. The requirement that it braid trivially with $(0,1,0)$ will imply that:
\begin{equation}
l^2 + ml'-1 = 0 \pmod{2q}.
\end{equation}
If we multiply by $2nl'$ then the fact that $nl^{\prime 2} = 0 \pmod{4q}$ will tell us that:
\begin{equation}
2nl'(l^2-1) = 0 \pmod{4q}.
\end{equation}
We can then multiply $n(l^2+1)=0\pmod{4q}$ by $2l'$ and subtract to see that $4nl'=0\pmod{4q}$, which implies $nl'\equiv 0\pmod{q}$. We can thus multiply $l^2+ml'-1=0\pmod{2q}$ by $n$ to see that:
\begin{equation}
n(l^2-1) = 0\pmod{q}.
\end{equation}
Substracting $n(l^2+1)=0\pmod{4q}$ gives us that $2n = 0\pmod{q}$.

Suppose that $n=q/2,3q/2,5q/2,7q/2\pmod{4q}$ then the requirement that $n(l^2+1)=0\pmod{4q}$ will mean that $l^2+1=0\pmod{8}$. Similarly if $n=q,3q$, then we need $l^2+1=0\pmod{4}$. Both obviously imply $l$ must be odd i.e., $l=2l'+1$ for some integer $l'$. But then it can easily be checked that $l^2+1=2\pmod{8}$ and thus $l^2+1\neq 0$ either modulo $4$ or $8$. So then we must have that $n=0,2q$. Clearly $n=0$ is simply $\bb{Z}_{2q}$ gauge theory, which is consistent with time reversal.

Suppose now that $n=2q$. Then our $K$-matrix is given by
\begin{equation}
K = \begin{pmatrix} 0& 2q& 0\\ 2q& -2q& 0\\ 0& 0& 1\end{pmatrix}.
\end{equation}
Applying the transformation:
\begin{equation}
W = \begin{pmatrix} 1& 0& 0\\ 1& 1& 0\\ 0& 0& 1\end{pmatrix}, \ W \in {\rm GL}(3,\bb{Z}),
\end{equation}
takes $K\rightarrow W^TKW$ and produces $\U_{2q}\boxtimes \U_{-2q}$. This is also clearly consistent with time reversal.
\end{proof}

\begin{remark}
We could stop here, but it will turn out that once we demand invariance under translation $\U_{2q}\boxtimes \U_{-2q}$ will no longer be an acceptable theory.
\end{remark}

\begin{lemma}
Let $\sc{A}\in mG^{(\nu)}_{GSD}$ for the group $G$ in Eq.~\eqref{eqn:symm_gp}. Then it must be the case that $a^2$ is invariant under translation i.e., $T_xa^2=a^2=T_ya^2$. This then implies that $a^2$ is a boson. It will also further imply that the unique $\gamma$ discussed earlier is a boson.
\label{lem:a_sQ_boson}
\end{lemma}
\begin{proof}
We follow the same logic as Thm.~\ref{thm:gam_is_boson_fermion}, but now with a unitary symmetry. Additionally we can now use the fact that since $2\theta_\gamma = 0$ then $Q_\gamma = 0\pmod{2}$ for the unique $\gamma$ discussed earlier. Then the only anyons with the same $\U$ charge modulo two as $a$ are $a\gamma^m$ and $a^{q+1}c\gamma^m$. Since translation preserves the $\U$ charge of anyons it must be the case that $T_xa=a\gamma^m$ or $T_xa=a^{q+1}c\gamma^m$ for some $m$. 

Since translation is a unitary symmetry of the theory it should also leave the topological spin of the anyons invariant. This therefore means that if $T_xa=a\gamma^m$ then:
\begin{align}
\theta_a &= \theta_{a\gamma^m}\\
&= \theta_a +\frac{\pi m}{q} + m^2\theta_\gamma.
\end{align}
Since $\theta_\gamma = 0,\pi$ we must have $m$ is a multiple of $q$. Since $q$ is even this means $m^2\theta_\gamma = 0$ and we must therefore have that $m$ is a multiple of $2q$. But $\gamma^{2q}=1$. So $T_xa\neq a\gamma^m$ unless $m=0$. Next suppose that $T_xa=a^{q+1}c\gamma^m$. Then:
\begin{align}
\theta_{a} &= \theta_{a^{q+1}c\gamma^m}\\
&= \pi + (q^2+2q+1)\theta_a + \frac{\pi m(q+1)}{q} + m^2\theta_\gamma.
\end{align}
If we subtract $\theta_a$ from both sides and multiply by two we see that:
\begin{align}
0 &= (2q^2+4q)\theta_a + \frac{2\pi m}{q},
\end{align}
since $2\theta_\gamma = 0$. Now we note that $a^{2q}=1$ so it must braid trivially with $a$, i.e., $0=\theta_{a,a^{2q}}=4q\theta_a$. Since $q$ is even this means that $(2q^2+4q)\theta_a=0$. Thus we have that $0=2\pi m/q$ so $m$ must be a multiple of $q$. The only anyons that $a$ can be taken to under translation are thus:
\begin{equation}
T_xa=a,a^{q+1}c, a^{q+1}c\gamma^q.
\end{equation}
In any case we see that $T_xa^2 = (T_xa)(T_xa)=a^2$. The same logic will hold for $T_y$. Thus we see that $a^2$ must be invariant under translation.

Next we appeal to the logic of Appendix \ref{app:non_sub_G}. There we demonstrated that
\begin{equation}
\theta_{a,d} = \theta_{\sc{T}a,d} \ \forall d\in \sc{C} \ \mathrm{such \ that} \ T_xd=d=T_yd.
\end{equation}
So then we may conclude:
\begin{align}
\theta_{a,a^2}&= \theta_{\sc{T}a,a^2}\\
&= 2\theta_{\sc{T}a,a}\\
&=0,
\end{align}
where the last line follows from the fact that $\theta_{\sc{T}a,a} = -\theta_{a,\sc{T}a}$. So then $0=\theta_{a,a^2}=4\theta_a=\theta_{a^2}$ and we can conclude that $a^2$ is a boson.

Finally, since $a^2$ is a boson its powers must be also. In particular, raising it to the $q/2$th power will imply that $a^q$ is a boson. But then by the negation of Lem.~\ref{lem:aq_fermion} we must have that $\gamma$ is a boson.
\end{proof}

\begin{theorem}
Let $\sc{A}\in mG^{(\nu)}_{\rm GSD}$ for the group $G$ in Eq.~\eqref{eqn:symm_gp}, where $q$ is even. Then $\sc{A}$ must be $\bb{Z}_{2q}$ gauge theory stacked with the fundamental fermion.
\end{theorem}
\begin{proof}
 We know by Lem.~\ref{lem:two_Q_gauge} that $\sc{A}$ must be $\bb{Z}_{2q}$ gauge theory or $\U_{2q}\boxtimes \U_{-2q}$ stacked with the fundamental fermion. Our task is thus just to show that $\U_{2q}\boxtimes \U_{-2q}$ cannot work.

We know by Lem.~\ref{lem:full_group} that all anyons in our TO can be written as $a^r\gamma^s c^n$. 
Then we see that:
\begin{align}
2q\theta_{a^r\gamma^s c^n} &= 2q\left(n\pi + r^2\theta_a + s^2\theta_\gamma + \frac{\pi rs}{q}\right)\\
&= \frac{q}{2}r^2\theta_{a^2} + \frac{q}{2}s^2\theta_{\gamma^2}\\
&= 0,
\end{align}
where this follows from the fact that $q$ is even and $a^2,\gamma^2$ are bosons.
It is easy to see that all anyons $b$ in our TO $\sc{A}$ have the property that $2q\theta_b = 0$.

But obviously in the TO $\U_{2q}\boxtimes \U_{-2q}$ there is some anyon $b$ such that $\theta_b = \pi/2q$ and thus $2q\theta_b = \pi$. We have therefore arrived at a contradiction, and we see that $\U_{2q}\boxtimes \U_{-2q}$ cannot be an acceptable TO.
\end{proof}

\bibliographystyle{apsrev4-1_custom}
\bibliography{LSM.bib}

\end{document}